%% file: tobit-lad-mle-draft9a.tex
\documentclass[12pt,british,letter]{article}
\usepackage[T1]{fontenc}
\usepackage[latin9]{inputenc}
\usepackage[a4paper]{geometry}
\usepackage{fancyhdr}
\usepackage{color}
\usepackage{babel}
\usepackage{verbatim}
\usepackage{mathtools}
\usepackage{enumitem}
\usepackage{amsmath}
\usepackage{amsthm}
\usepackage{amssymb}
\usepackage{setspace}
\usepackage[authoryear,longnamesfirst]{natbib}
\usepackage[pdfusetitle,
 bookmarks=true,bookmarksnumbered=false,bookmarksopen=false,
 breaklinks=false,pdfborder={0 0 0},pdfborderstyle={},backref=false,colorlinks=false]
 {hyperref}
\usepackage{booktabs}
\usepackage{graphicx}

\makeatletter

\numberwithin{equation}{section}
\numberwithin{figure}{section}
\numberwithin{table}{section}
\theoremstyle{remark}
\newtheorem*{notation*}{\protect\notationname}
\theoremstyle{plain}
\newtheorem{assumption}{\protect\assumptionname}
\theoremstyle{plain}
\newtheorem{thm}{\protect\theoremname}[section]   
\newtheorem{lem}[thm]{\protect\lemmaname}         
\newtheorem{prop}[thm]{\protect\propositionname}  
\theoremstyle{definition}

\theoremstyle{remark}

\setlist[enumerate,1]{label=\upshape{(\roman*)}, ref=(\roman*)}
\setlist[enumerate,2]{label=\upshape{(\alph*)}, ref=(\alph*)}
\setlist[enumerate,3]{label=\upshape{\roman*.}, ref=\roman*}

\makeatletter
  \@ifpackageloaded{mathrsfs}{}{\usepackage{mathrsfs}}
\makeatother

\date{}

\usepackage{nextpage}

\allowdisplaybreaks[1]

\usepackage{scalefnt}
\usepackage[group-separator={,}]{siunitx}
\newcommand\smaller[2][0.85]{{\scalefont{#1}#2}}
\newcommand{\ass}[1]{{\upshape{\smaller[0.76]{#1}}}}

\newcommand{\assumpname}[1]{%
  \renewcommand{\theassumption}{\ass{#1}}%
}

\makeatletter
\newsavebox{\@brx}
\newcommand{\dbllangle}[1][]{\savebox{\@brx}{\(\m@th{#1\langle}\)}%
  \mathopen{\copy\@brx\kern-0.5\wd\@brx\usebox{\@brx}}}
\newcommand{\dblrangle}[1][]{\savebox{\@brx}{\(\m@th{#1\rangle}\)}%
  \mathclose{\copy\@brx\kern-0.5\wd\@brx\usebox{\@brx}}}
\makeatother

\usepackage{changepage}

\newcounter{subremark}[rem]

\usepackage{mathdots}

\usepackage[usestackEOL]{stackengine}
\setstackgap{S}{5pt}
\newcommand{\authaffil}[2]{\Shortunderstack{#1\\\small{#2}}}

\usepackage{needspace}

\usepackage{caption}
\usepackage{subcaption}

\makeatother

\providecommand{\assumptionname}{Assumption}
\providecommand{\examplename}{Example}
\providecommand{\lemmaname}{Lemma}
\providecommand{\notationname}{Notation}
\providecommand{\propositionname}{Proposition}
\providecommand{\remarkname}{Remark}
\providecommand{\theoremname}{Theorem}

\begin{document}
\include{macros-v4}

\global\long\def\vec#1{\boldsymbol{#1}}%

\global\long\def\jsr{{\scriptstyle \mathrm{JSR}}}%

\global\long\def\lpar{p}%

\global\long\def\lcof{r}%

\global\long\def\parset{{\cal P}}%

\global\long\def\ucc{L^{\mathrm{ucc}}}%

\global\long\def\lad{{\scriptscriptstyle \mathsf{L}}}%

\global\long\def\mle{{\scriptscriptstyle \mathsf{M}}}%

\global\long\def\S{\mathbb{S}}%

\global\long\def\W{\mathbb{W}}%

\global\long\def\U{\mathbb{U}}%

\global\long\def\M{\mathbb{M}}%

\global\long\def\N{\mathbb{N}}%

\global\long\def\L{\mathbb{L}}%

\global\long\def\Z{\mathcal{Z}}%

\title{Estimation of a Dynamic Tobit Model\\ with a Unit Root
\thanks{The authors would like to thank Vadim Gorin for helpful comments and Yiting Fu for excellent research assistance.}}
\author{\authaffil{Anna Bykhovskaya}{Duke University}\hspace{3cm}
\authaffil{James A.\ Duffy}{University of Oxford}}

\date{\vspace*{0.3cm}\today}

\maketitle




\setcounter{footnote}{0}
\begin{abstract}
\noindent
This paper studies robust estimation in the dynamic Tobit model under local-to-unity (LUR) asymptotics. We establish consistency of both Gaussian maximum likelihood (ML) and censored least absolute deviations (CLAD) estimators under LUR asymptotics, complementing their known consistency in the stationary region and hence providing robustness to persistence in the sense of consistency in both regimes. We derive their asymptotic distributions under correct specification and show that the estimators of the short-run parameters are asymptotically Gaussian, yielding standard normal $t$-statistics. These findings contrast with OLS: it is inconsistent in the stationary region and, although consistent under LUR, its $t$-statistics are not asymptotically standard normal. Consequently, inference and sequential lag selection based on ML and CLAD $t$-tests are valid, with consistent lag selection obtained using suitably diverging critical values. Gaussian quasi-ML remains consistent under LUR when the innovations are non-Gaussian, but its limiting distributions generally differ from those under correct specification. An application to U.S. Treasury bill rates illustrates the practical relevance of the results.


\medskip

\end{abstract}
\vfill

\thispagestyle{plain}

\pagenumbering{roman}

\newpage{}

\thispagestyle{plain}

\setcounter{tocdepth}{2}


\newpage{}

\pagenumbering{arabic}

\section{Introduction}
Many macroeconomic and financial time series are highly persistent,
often exhibiting the random wandering that is the hallmark of a stochastic
trend. While it is relatively straightforward to accommodate these
series within the framework of a linear (vector) autoregressive model,
even in this setting such persistence creates significant challenges
for inference, giving rise to limiting distributions that are non-standard
and dependent on nuisance parameters. Accordingly, a substantial literature has developed on inference methods that are broadly robust to temporal dependence, remaining valid for both stationary and stochastically trending data (see e.g., \citealp{Hansen99REStat};
\citealp{JM06Ecta}; \citealp{Mik07Ecta,mikusheva2012one}; \citealp{EMW15Ecta}).

This literature has now attained a high degree of maturity, but with the significant limitation that it is entirely devoted to linear
models. Although nonlinear autoregressive models have been widely
studied, this has been almost exclusively in the context of stationary
processes (as in e.g., \citealp{Tong90}; \citealp{Chan09}; \citealp{TTG10}).
Indeed, for a long time the only available dynamic models that permitted
some conjunction of nonlinearities with stochastic trends were those
in which any nonlinearities were confined to the short-run dynamics,
i.e.\ to first differences and equilibrium error components (as in
the nonlinear VECM models considered by \citealp{BF97IER}; \citealp{HS02JoE};
\citealp{KR10JoE}). This excluded, for example, nonlinear autoregressive
models with regimes endogenously dependent on the \emph{level} of
a stochastically trending series. Only recently has it
been shown possible to configure nonlinear (vector) autoregressive
models so as to accommodate both nonlinearity in levels and stochastic
trends (\citealp{cavaliere2005}; \citealp{BD22}; \citealp{DM24}; \citealp{DMW22}).

A major impetus for this recent work has come from the desire to adequately
model highly persistent series that are subject to occasionally binding
constraints, of which the zero lower bound on nominal interest rates
is a leading example (\citealp{Mav21Ecta}, \citealp{AMSV21JoE}).
For such series, a plausible descriptive model is the dynamic Tobit,
which in a stationary setting has a lengthy pedigree in the empirical
analysis of constrained, i.e.\ censored, time series (see e.g., \citealp{Demiralp_Jorda2002},
\citealp{stat_paper}, \citealp{DSK12AE}, \citealp{LMS19}, \citealp{BMMV21JBF},
\citealp{bykh_JBES}). \citet{BD22} recently determined conditions
under which this model may also generate stochastically trending,
$I(1)$-like series, and obtained the asymptotic distribution of the OLS
estimators for this model in a local-to-unity framework. However, while their results
are sufficient for the purposes of conducting unit root tests, there
is no possibility of extending the validity of OLS beyond their local-to-unity setting,
since even its consistency is known to fail for the stationary dynamic
Tobit.

The present work is thus motivated by the need to develop methods
of estimation and inference, in a \emph{nonlinear} time series setting,
that are robust to the degree of persistence of the data generating
mechanism. More specifically, we seek estimators that remain consistent
under both stationary and local-to-unity regimes, just as OLS does
in a linear autoregressive model.
This leads us to consider two alternative estimators that are known to be consistent and asymptotically normal
for the stationary dynamic Tobit: Gaussian maximum likelihood (ML)
and \citeauthor{Pow84JoE}'s \citeyearpar{Pow84JoE} censored least
absolute deviations (CLAD; see \citealp{stat_paper}, and \citealp{bykh_JBES}).

The principal contribution of this paper is to establish consistency and derive the asymptotic
distributions of the ML and CLAD estimators for a dynamic Tobit model,
in a local-to-unity framework. When the model is
rendered in ADF form, the joint limiting distribution of the estimators
of the intercept and the sum of the autoregressive coefficients is
non-standard but nevertheless usable for the purposes of asymptotic inference.

Of particular significance are the estimates of the `short run'
coefficients, i.e.\ those on lagged differences. For both ML and CLAD,
these are asymptotically Gaussian, and so the
corresponding $t$-statistics are asymptotically standard normal. Thus
inferences on these coefficients and on smooth functionals of them, such as short-horizon impulse responses, remain
standard in both the stationary and local-to-unity regimes (exactly as
is the case for OLS estimation in a \emph{linear} autoregressive model).
By contrast, this is not true of the OLS $t$-statistics,
which have non-standard limits under local-to-unity asymptotics. Our results additionally
permit the model lag order to be consistently estimated, on the basis of
sequential $t$-tests with critical values that diverge at an appropriate rate.


We also show that under non-Gaussian innovations, the Gaussian quasi-ML estimator
remains consistent in the local-to-unity regime. In the usual
Tobit model, such misspecification can instead lead to inconsistency
\citep{arabmazar1982investigation}. The difference is due to the
slower accumulation of censored observations in our setting.
Consistency does not, however, imply that the limiting distributions
are unaffected by misspecification. These generally differ from
the Gaussian case, and the short-run $t$-statistics need no longer
be asymptotically standard normal.

At a technical level, our results contribute to the literature on
the asymptotics of nonlinear extremum estimators with highly persistent
processes (see e.g., \citealp{PP00Ecta,PP01Ecta}; \citealp{PJH07JoE};
\citealp{Xiao09JoE}; \citealp{CW15JoE}). Relative to previous work,
the analysis is complicated by two aspects of our problem. Firstly,
we consider a nonlinear \emph{autoregressive} model, whereas that literature
has been concerned with regression (or discrete choice) models where
the regressor $x_{t}$ is a \emph{linear} unit root process, and the
object is to estimate some (possibly) nonlinear transformation of
$x_{t}$. Secondly, in our analysis of the CLAD estimator, we cannot
rely on the convexity of the criterion function, something which greatly
facilitates the derivation of the asymptotics of the \emph{uncensored}
LAD estimator, as in \citet{Poll91ET}, \citet{Herce96ET}, \citet{LL09ET}
and \citet{Xiao09JoE}. (By contrast, convexity may be fruitfully
exploited in the analysis of the Gaussian MLE.) Nor may approaches
employed in the i.i.d.\ or stationary cases \citep{Pow84JoE,stat_paper,bykh_JBES}
be readily transposed to our setting. Thus deriving the asymptotics
of the (centred) CLAD criterion function, in particular, requires
some delicate and novel arguments. We expect these arguments will
also prove useful for the analysis of other nonlinear extremum estimators,
in the presence of highly persistent data.

The remainder of this paper is organised as follows. Section~\ref{sec_model} introduces the dynamic Tobit model and our assumptions. Section~\ref{sec_estimation} derives the asymptotic distributions of ML and CLAD estimators. Section~\ref{sec_modelselection} considers lag selection, and Section~\ref{sec_discussion} discusses the results. An empirical illustration is presented in Section~\ref{sec_empirical}. Finally, Section~\ref{sec_conclusion} concludes. All proofs appear in the appendices.

\begin{notation*}

All limits are taken as $T\goesto\infty$ unless otherwise stated.
$\inprob$ and $\indist$ respectively denote convergence in probability
and in distribution (weak convergence). We write `$X_{T}(\lambda)\indist X(\lambda)$
on $D_{\reals^{m}}[0,1]$' to denote that $\{X_{T}\}$ converges
weakly to $X$, where these are considered as random elements of $D_{\reals^{m}}[0,1]$,
the space of cadlag functions $[0,1]\setmap\reals^{m}$, equipped
with the uniform topology; we denote this as $D[0,1]$ whenever the
value of $m$ is clear from the context. $\smlnorm{\cdot}$ denotes
the Euclidean norm on $\reals^{m}$, and the matrix norm that it induces.
For $X$ a random vector and $p\geq1$, $\smlnorm X_{p}\defeq(\expect\smlnorm X^{p})^{1/p}$.
\end{notation*}

\section{Model}\label{sec_model}
We suppose $\{y_t\}$ is generated by a dynamic Tobit model of order $k$, written in
augmented Dickey--Fuller (ADF) form as
\begin{equation}\label{eq:tobitark}
y_{t}=\left[\alpha_0+\beta_0 y_{t-1}+\vec{\phi}_0^{\trans}\Delta\vec y_{t-1}+u_{t}\right]_{+},\qquad t=1,\ldots T,
\end{equation}
where $[x]_+:=\max\{x,0\}$, $\vec y_{t-1}:=(y_{t-1},\ldots, y_{t-k+1})^{\trans}$, $\vec{\phi}_0=(\phi_{1,0},\ldots,\phi_{k-1,0})^{\trans}$ and $\Delta y_{t}:=y_t-y_{t-1}$. The model has two attractive features. Firstly, it is Markovian -- with the state vector defined by $k$ consecutive lags of $y_t$ -- making it well-suited for forecasting. Secondly, the presence of the positive part on the right of the equality enforces a lower bound of zero on $y_t$, making the model appropriate for series that are constrained to take non-negative values.

We assume that the parameters $\rho_0:=(\alpha_0,\beta_0,\vec{\phi}_0^{\trans})^{\trans}\in\mathbb{R}^{k+1}$, innovations $\{u_t\}$ and initial conditions satisfy the following assumptions. 

\assumpname{A1}
\begin{assumption}
\label{ass:INIT} $\{y_{t}\}$ is initialised by (possibly) random
variables $\{y_{-k+1},\ldots,y_{0}\}$. Moreover, $T^{-1/2}y_{0}\inprob b_{0}$
for some $b_{0}\geq0$.
\end{assumption}
\assumpname{A2}
\begin{assumption}
\label{ass:DGP} $\{y_{t}\}$ is generated according to (\ref{eq:tobitark}),
where:
\begin{enumerate}[label=\ass{\arabic*.}, ref=\ass{.\arabic*}]
\item \label{enu:DGP:ut}$\{u_{t}\}_{t\in\integers}$ is independently
and identically distributed with $\expect u_{t}=0$ and
$\expect u_{t}^{2}=\sigma_{0}^{2}$.
\item \label{enu:DGP:LU}$\alpha_0=\alpha_{T,0}\defeq T^{-1/2}a_{0}$ and
$\beta_0=\beta_{T,0}=1+T^{-1}c_{0}$, for some $a_{0},c_{0}\in\reals$.
\end{enumerate}
\end{assumption}

\assumpname{A3}
\begin{assumption}
\label{ass:MOM} There exist $\delta_{u}>0$ and $C<\infty$ such
that:
\begin{enumerate}[label=\ass{\arabic*.}, ref=\ass{.\arabic*}]
\item $\expect\smlabs{u_{t}}^{2+\delta_{u}}<C$.
\item $\expect\smlabs{T^{-1/2}y_{0}}^{2+\delta_{u}}<C$, and $\expect\smlabs{\Delta y_{i}}^{2+\delta_{u}}<C$
for $i\in\{-k+2,\ldots,0\}$.
\end{enumerate}
\end{assumption}
Assumption \ref{ass:INIT} allows the initialisation to be of the same order of magnitude as the process $y_t$. This is appropriate because the first observation in any sample is unlikely to be the `true' starting point of the process, but simply the point at which data collection begins. Assumption \ref{ass:DGP}\ref{enu:DGP:LU} ensures that the data is highly persistent, as a consequence of the autoregressive polynomial having a root local to unity. Assumption~\ref{ass:DGP}\ref{enu:DGP:ut} specifies i.i.d.\ innovations, while Assumption~\ref{ass:MOM} imposes moment conditions
on both the innovations and the initial observations. We retain the
i.i.d.\ assumption to simplify the analysis and to allow auxiliary
results from \citet{BD22} to be applied directly. Extensions to
dependent innovations, and their implications for estimation, are
discussed in Section~\ref{sec_estimation}.


As shown in \citet[Appendix C]{BD22}, owing to the nonlinearity of the model additional conditions -- which go beyond the classical requirements on the stability of the roots of the autoregressive polynomial -- are needed to rule out explosive behaviour of the first differences $\{\Delta y_t\}$. Here we ensure this through a constraint on the joint spectral radius of two auxiliary companion-form matrices. For $\delta\in[0,1]$, define
\begin{equation}
F_{\delta}\defeq\begin{bmatrix}\phi_{1,0}\delta & \phi_{2,0} & \cdots & \phi_{k-2,0} & \phi_{k-1,0}\\
\delta & 0 & \cdots & 0 & 0\\
0 & 1\\
 &  & \ddots\\
 &  &  & 1 & 0
\end{bmatrix}\label{eq:Fdef}
\end{equation}
and let $\lambda_{\jsr}({\cal A})$ denote the joint spectral radius
of a bounded collection of matrices ${\cal A}$, which may
be defined as
\[
\lambda_{\jsr}(\mathcal{A})\defeq\limsup_{n\goesto\infty}\sup_{M\in\mathcal{A}^{n}}\lambda(M)^{1/n}
\]
where $\lambda(M)$ denotes the spectral radius of $M$, and $\mathcal{A}^{n}\defeq\{\prod_{i=1}^{n}A_{i}\mid A_{i}\in\mathcal{A}\}$
(cf., \citealp[Defn.~1.1]{Jungers09}).

\assumpname{A4}
\begin{assumption}
\label{ass:JSR}$\lambda_{\jsr}(\{F_{0},F_{1}\})<1$.
\end{assumption}
%


Approximate upper bounds for the joint spectral radius (JSR) can be computed numerically to an arbitrarily high degree of accuracy using semidefinite programming \citep{PJ08LAA}, making it possible to verify whether the condition is satisfied for a given set of parameter values (see \citealp{DMW23stat}, for a further discussion). Assumption \ref{ass:JSR} ensures that the process $y_t$ does not exhibit explosive behavior by jointly controlling the dynamics across both the censored ($\delta=0$) and uncensored ($\delta=1$) `regimes'. Without this restriction the interaction between the two regimes could generate a `bouncing' effect, where hitting the lower bound triggers a strong rebound, potentially leading to exponential growth (for a numerical example, see \citealp{BD22}, Appendix C).

\citet[Theorem 3.2]{BD22} show that under Assumptions \ref{ass:INIT}-\ref{ass:JSR}
\begin{equation}\label{eq_LUR_limit}
\frac{1}{T^{1/2}}y_{\lfloor \tau T\rfloor} \indist \phi(1)^{-1}e^{c_0\tau/\phi(1)}
   \left(K(\tau) +\sup_{\tau^{\prime}\leq \tau}\left[-K(\tau^{\prime})\right]_{+} \right)
   \eqdef Y(\tau),
\end{equation}
on $D[0,1]$, where
\begin{equation*}
K(\tau)\defeq\phi(1)b_0+a_0\int_0^{\tau}e^{-c_0 r / \phi(1)}\diff r  +\sigma_0\int_0^{\tau}e^{-c_0 r / \phi(1)} \diff W(r)
\end{equation*}
for $\phi(1)=1-\sum_{i=1}^{k-1}\phi_{i,0}>0$, and $W(\cdot)$ a standard Brownian motion. The limiting distribution of the estimators of $\alpha$ and $\beta$, though not of $\vec{\phi}$, will be shown below to depend on the process $Y(\cdot)$ -- a dependence similar to that of the OLS estimators on the limiting Brownian motion (or Ornstein--Uhlenbeck process) in a linear autoregressive model with a root at (or local to) unity. In fact, the process $K(\cdot)$ -- which $Y(\cdot)$ is the constrained counterpart of -- corresponds exactly to the limit of a local-to-unity autoregressive process. So surprisingly, despite the nonlinearity of the dynamic Tobit model, a similar asymptotic structure emerges: the limiting distributions of our estimators will closely resemble those of a linear autoregression, albeit with a nonlinear limiting process $Y(\cdot)$.

\section{Estimation and inference}\label{sec_estimation}
This section introduces two estimation methods -- Gaussian maximum likelihood (ML) and censored least absolute deviations (CLAD) -- and establishes their asymptotic properties. We show that the estimators of the coefficients $\vec{\phi}_0$ on the stochastically bounded (though not necessarily stationary) lags $\Delta \vec {y}_{t-1}$ are asymptotically normal under both methods, subject to their respective assumptions, permitting standard inferences to be drawn. As discussed further in Section~\ref{sec_OLS_MLE_LAD}, this contrasts sharply with the behaviour of ordinary least squares estimators, which have a limiting bias (of order $T^{-1/2}$) due to the censoring, even in the presence of a root local to unity.

\subsection{Maximum likelihood}\label{sec_MLE}

We first study the Gaussian ML estimator under correct specification.
We then relax the Gaussianity assumption on the innovations and examine
the consistency and asymptotic behaviour of the resulting quasi-ML
estimator.

\subsubsection{Gaussian innovations}

To permit estimation by maximum likelihood, we need to make a parametric assumption regarding the distribution of the innovations $\{u_{t}\}$. A conventional choice -- which is particularly attractive
for computational reasons, because it yields a loglikelihood that is concave after reparametrisation -- is for $u_{t}$ to be Gaussian, as per the following.

\assumpname{B}
\begin{assumption}
\label{ass:TOB} $u_{t}\sim N[0,\sigma_{0}^{2}]$, and $\delta_{u}>2$ in \ref{ass:MOM}.
\end{assumption}

Note that for $\{u_{t}\}$ itself, the maintained Gaussianity makes the condition on $\delta_{u}$ redundant; this condition is nonetheless still needed to ensure that the initial conditions have a little more than four finite moments.

Under Assumption \ref{ass:TOB}, when $\{y_{t}\}$ is generated according to (\ref{eq:tobitark}),
the conditional density of $y_{t}$ given $(y_{t-1},\ldots,y_{t-k})$,
evaluated at the values $y_{t-i}=\mathsf{y}_{t-i}$ for $i \in \{0, \ldots, k\}$, is
\begin{equation}
f_{(\alpha,\beta,\vec{\phi},\sigma)}(\mathsf{y}_{t}\mid\mathsf{y}_{t-1},\ldots,\mathsf{y}_{t-k})=\begin{cases}
\sigma^{-1}\varphi[\sigma^{-1}(\mathsf{y}_{t}-\alpha-\beta\mathsf{y}_{t-1}-\vec{\phi}^{\trans}\Delta\vec{\mathsf{y}}_{t-1})] & \text{if }\mathsf{y}_{t}>0,\\
1-\Phi[\sigma^{-1}(\alpha+\beta\mathsf{y}_{t-1}+\vec{\phi}^{\trans}\Delta\vec{\mathsf{y}}_{t-1})] & \text{if }\mathsf{y}_{t}=0;
\end{cases}\label{eq:conddens}
\end{equation}
where $\vec{\mathsf{y}}_{t-1} = (\mathsf{y}_{t-1},\ldots,\mathsf{y}_{t-k+1})^{\trans}$, and $\Phi(x)$ and $\varphi(x)$ denote the standard normal cdf and pdf, respectively, i.e.\ $\varphi(x)\defeq\tfrac1{\sqrt{2\pi}}\e^{-x^{2}/2}$. We consider a conditional ML estimator that maximises the loglikelihood of $\{y_{t}\}_{t=1}^{T}$ conditional on the initial values
$\{y_{i}\}_{i=-k+1}^{0}$. Thus, it maximises
\begin{equation}\label{eq:gaus_likel}
{\cal L}_{T}(\alpha,\beta,\vec{\phi},\sigma)\defeq\sum_{t=1}^{T}\log f_{(\alpha,\beta,\vec{\phi},\sigma)}(y_{t}\mid y_{t-1},\ldots,y_{t-k})
\end{equation}
with respect to $\alpha,\beta,\vec{\phi},\sigma$. Define
\[
(\hat{\alpha}_{T}^{\mle},\hat{\beta}_{T}^{\mle},\hat{\vec{\phi}}{}_{T}^{\mle},\hat{\sigma}_{T}^{\mle})\in\argmax_{(\alpha,\beta,\vec{\phi},\sigma)\in\reals^{k+1}\times\reals_{+}}{\cal L}_{T}(\alpha,\beta,\vec{\phi},\sigma)
\]
to be a sequence of maximisers of ${\cal L}_{T}$. Let $\Omega\defeq\plim_{T\to\infty}\frac1{T}\sum_{t=1}^{T}\Delta\mathbf{y}_{t-1}\Delta\mathbf{y}_{t-1}^{\trans}$, which exists and is positive definite by Lemma~B.4 in \citet{BD22}.

\begin{thm}\label{thm:tobitmle}
Suppose \ref{ass:INIT}--\ref{ass:JSR} and \ref{ass:TOB} hold. Then
\begin{equation}\label{eq_MLE_LUR}
\begin{bmatrix}T^{1/2}(\hat{\alpha}_{T}^{\mle}-\alpha_{0})\\
T(\hat{\beta}_{T}^{\mle}-\beta_{0})
\end{bmatrix}\xrightarrow[T\to\infty]{d}\sigma_0\begin{bmatrix}1 & \int_{0}^{1}Y(\tau)\diff\tau\\
\int_{0}^{1}Y(\tau)\diff\tau & \int_{0}^{1}Y^{2}(\tau)\diff\tau
\end{bmatrix}^{-1}\begin{bmatrix}W(1)\\
\int_{0}^{1}Y(\tau)\diff W(\tau)
\end{bmatrix}
\end{equation}
jointly with, and independently of
\begin{equation}\label{eq_MLE_short_run}
T^{1/2}\begin{bmatrix}\hat{\vec{\phi}}{}_{T}^{\mle}-\vec{\phi}_{0}\\
(\hat{\sigma}_{T}^{\mle})^{2}-\sigma_{0}^{2}
\end{bmatrix}\xrightarrow[T\to\infty]{d}\begin{bmatrix}\sigma_{0}^{2}\Omega^{-1} & 0\\
0 & 2\sigma_{0}^{4}
\end{bmatrix}^{1/2}\mathcal{N}[0,I_{k}].
\end{equation}
\end{thm}

In contrast to the linear autoregressive model, where the Gaussian MLE is numerically identical to OLS, and therefore has the same asymptotic distribution, Theorem~\ref{thm:tobitmle} and the simulations in Section~\ref{sec_OLS_MLE_LAD} show that censoring breaks this equivalence (see Theorem 3.4 of \citealp{BD22}, for the corresponding OLS limit theory). Interestingly, the limiting distributions in \eqref{eq_MLE_LUR} and \eqref{eq_MLE_short_run} directly parallel their OLS counterparts in a linear autoregression -- as given e.g.\ in (17.7.25) and (17.7.27) of \citet{Hamilton94} -- modulo the limiting $Y$ being a different process.

The asymptotic variance in \eqref{eq_MLE_short_run} is given by the inverse of the information matrix, i.e., the inverse of the negative Hessian (see Proposition~\ref{prop:mlescorehess}), matching the corresponding result in the stationary case \citep[Theorem 3]{stat_paper}. Unlike in the stationary case, where censored observations occur much more frequently (accumulating at rate $T$ rather than $T^{1/2}$), the limiting distribution of the short-run parameter estimators in \eqref{eq_MLE_short_run} is entirely unaffected by censoring, while that of the long-run parameter estimators in \eqref{eq_MLE_LUR} is affected only through the nonlinear limiting process $Y$.

The proof of Theorem~\ref{thm:tobitmle} relies on the
concavity of the Gaussian loglikelihood for the (dynamic) Tobit model,
which holds after reparametrising the model as in \citet{Ols78Ecta}.
This allows the limiting distribution of the MLE to be deduced straightforwardly
from the pointwise asymptotics of the (rescaled) likelihood, by utilising
the convexity lemma of \citet{Kni89CJS}; see also \citet{Poll91ET}.
The major technical challenges then concern the limiting behaviour
of the score and the Hessian, both of which require non-standard arguments,
because the elevated persistence of $\{y_{t}\}$ prevents the application
of limit theorems for weakly dependent, stationary processes. The
analysis of the score relies on its martingale property
(at the true parameters), which facilitates a simultaneous application
of the martingale central limit theorem (for the weakly persistent
components in \eqref{eq_MLE_short_run}) and theorems on the convergence
to stochastic integrals (for the more persistent components in \eqref{eq_MLE_LUR})
to determine the joint limiting distribution of all components.

\subsubsection{Non-Gaussian innovations}

The preceding analysis maintains that the model is correctly specified, including in particular the Gaussianity assumption on the innovations. Except in certain special cases -- such as regression models with innovations whose distribution belongs to a linear exponential family -- in general we would anticipate that relaxing this assumption would yield a quasi-maximum likelihood estimator (QMLE) that is consistent merely for some `pseudotrue' value of the parameters, distinct from their true values \citep{AM96ET}. However, in the case of a highly persistent (i.e.\ local-to-unity) process, the slower accumulation of censored observations allows the QMLE to remain consistent for the true parameters, irrespective of the distribution of the innovations, as stated in the following theorem.


\begin{thm}\label{thm:tobitqmle}
Suppose \ref{ass:INIT}--\ref{ass:JSR} hold with $\delta_u>2$. Then
\begin{equation*}
(\hat{\alpha}_{T}^{\mle}, \hat{\beta}_{T}^{\mle}, \hat{\vec{\phi}}{}_{T}^{\mle}, (\hat{\sigma}_{T}^{\mle})^{2})
-(\alpha_{0}, \beta_{0}, \vec{\phi}_{0}, \sigma_{0}^{2})
\xrightarrow[T\to\infty]{p}0.
\end{equation*}
\end{thm}

To this extent, the Gaussian QMLE enjoys a certain robustness to departures
from exact Gaussianity of the innovations, i.e.\ to departures from
Assumption~\ref{ass:TOB}. On the other hand, the limiting distributions
in Theorem~\ref{thm:tobitmle} will not, in general, be those of
the QMLE, because the martingale property of the (quasi) loglikelihood
score is generally lost under misspecification. 
More precisely, the score decomposes as the sum of a martingale 
and an additional term associated with the censored observations. Under
Gaussianity, the latter has conditional mean zero, and its contribution
to the normalised score is asymptotically negligible. If the innovation
distribution is misspecified, however, that conditional mean will
generally be nonzero, and of comparable magnitude to the martingale
component of the score.



The asymptotics of this predictable component do not appear to be covered, at
least not at the level of generality of Assumptions~\ref{ass:INIT}--\ref{ass:JSR},
by any existing literature. In the absence of censoring, it would
be possible to apply results from the theory of convergence of additive
functionals of integrated processes to local times (\citealp{Jeg04};
\citealp{WP09ET}), but the nonlinearity in the data generating mechanism
prevents this. In certain special cases -- in particular, if $k=1$,
$\alpha_{0}=0$, $\beta_{0}=1$, and $y_0=0$ -- it is possible to use ergodic
theory for null recurrent Markov processes to determine the weak limit
of the score, and thence of the QMLE (cf.\ \citealp{GLY13JoE}).
Appendix~\ref{sec:QMLE_example} derives this analytically for the particularly tractable case
where $u_{t}$ has a zero-mean Laplace distribution with scale $2^{-1/2}$, so that $\mathbb{E}u_t^2=1$. In this case,
$$
\begin{bmatrix}T^{1/2}(\hat{\alpha}_{T}^{\mle}-\alpha_{0})\\
T(\hat{\beta}_{T}^{\mle}-\beta_{0})
\end{bmatrix}\xrightarrow[T\to\infty]{d}\sigma_0\begin{bmatrix}1 & \int_{0}^{1}Y(\tau)\diff\tau\\
\int_{0}^{1}Y(\tau)\diff\tau & \int_{0}^{1}Y^{2}(\tau)\diff\tau
\end{bmatrix}^{-1}\begin{bmatrix}W(1)+2A_L\sup_{\tau\leq 1}\left[-K(\tau)\right]_{+}\\
\int_{0}^{1}Y(\tau)\diff W(\tau)
\end{bmatrix},
$$
where $A_{L}$ is the nonzero constant defined in \eqref{eq:lap00-AL}.

Figure~\ref{fig:MLEvsQMLE_mixture} provides further simulation-based
evidence for the manner in which the distribution of the QMLE may differ
from that obtained under correct specification, by supposing that
the innovations follow a symmetric Gaussian mixture of the form
\[
u_t\thicksim
\begin{cases}
N(0,0.05^2), & \text{with probability }0.99,\\
N(0,\mathfrak a^2), & \text{with probability }0.01,
\end{cases}
\qquad
\mathfrak a^2=100-99(0.05)^2,
\]
so that $Eu_t=0$ and $Eu_t^2=1$. This is a symmetric distribution,
but one which is deliberately chosen to be extreme in terms of its
departure from Gaussianity: its kurtosis is almost $300$, rather than
$3$. The figure shows that, in this case, the distribution of
the estimated coefficient on the lagged difference, $\hat\phi_{1,T}$,
diverges strikingly from its distribution under Gaussianity. At the
same time, engineering this result requires
substantial non-Gaussianity in $\{u_t\}$, and as shown in Section~\ref{sec_OLS_MLE_LAD} below,
under more moderate departures from Gaussianity, the distribution of the QMLE
is substantially closer to its Gaussian counterpart.


As with OLS, the consistency of the Gaussian QMLE under
misspecification relies on the local-to-unity parametrisation, which
attenuates the effects of censoring, and does not in general extend
into the stationary region. While the Gaussian MLE remains a useful
benchmark against which to gauge the efficiency of alternative
estimators, these limitations motivate our subsequent development of
the asymptotics of the censored least absolute deviations (CLAD)
estimator. This permits valid inferences to be drawn for a range of
non-Gaussian innovation distributions under both stationary and
local-to-unity specifications, provided that $f_u$ has zero median
and suitable regularity conditions hold.


\begin{figure}[t]
\begin{subfigure}{.3\textwidth}
\centering \includegraphics[width=1\linewidth]{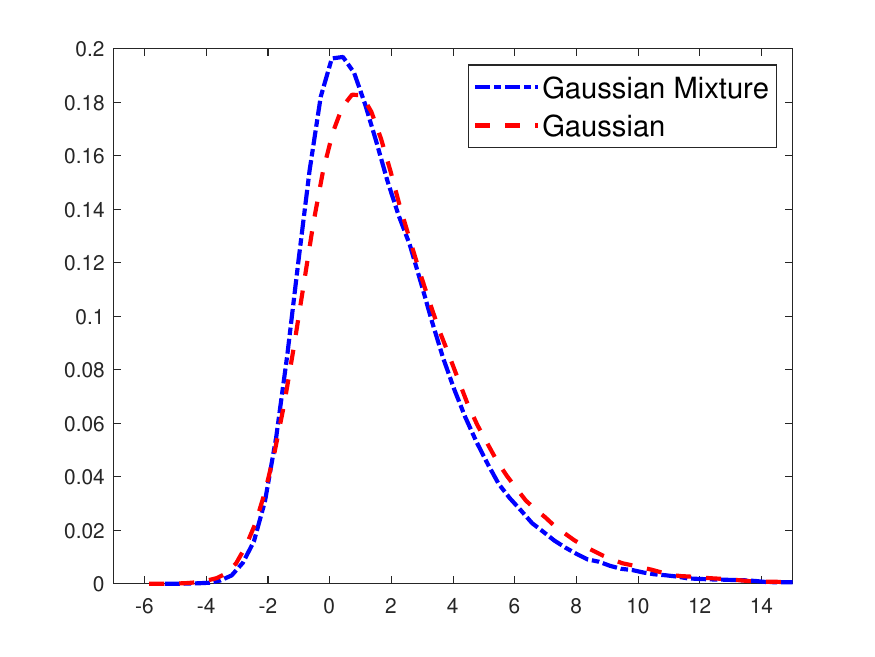}
\caption{$\sqrt{T}\hat{\alpha}$.}
\end{subfigure}
\begin{subfigure}{.3\textwidth}
\centering \includegraphics[width=1\linewidth]{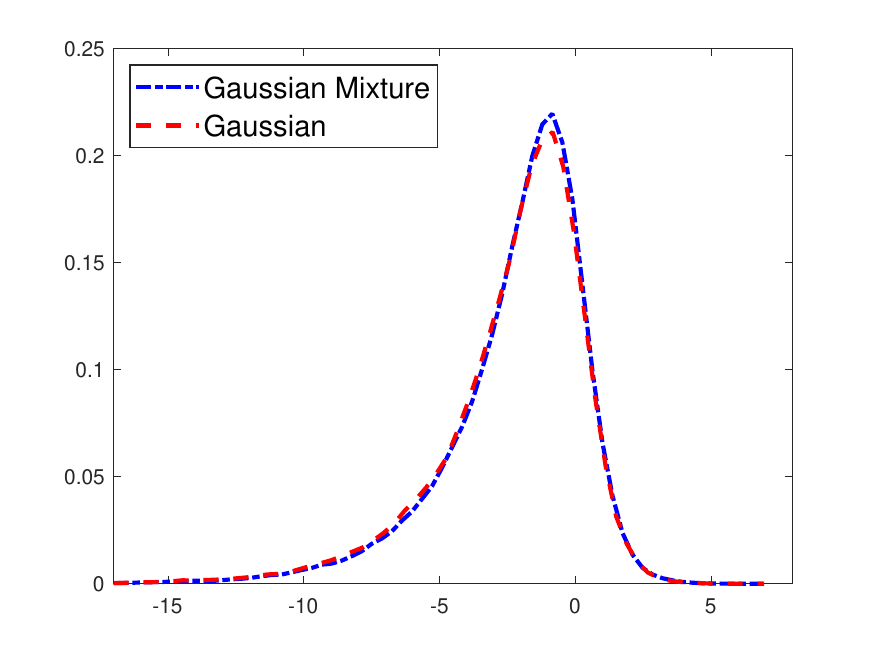}
\caption{$T(\hat{\beta}-1)$.}
\end{subfigure}
\begin{subfigure}{.3\textwidth}
\centering \includegraphics[width=1\linewidth]{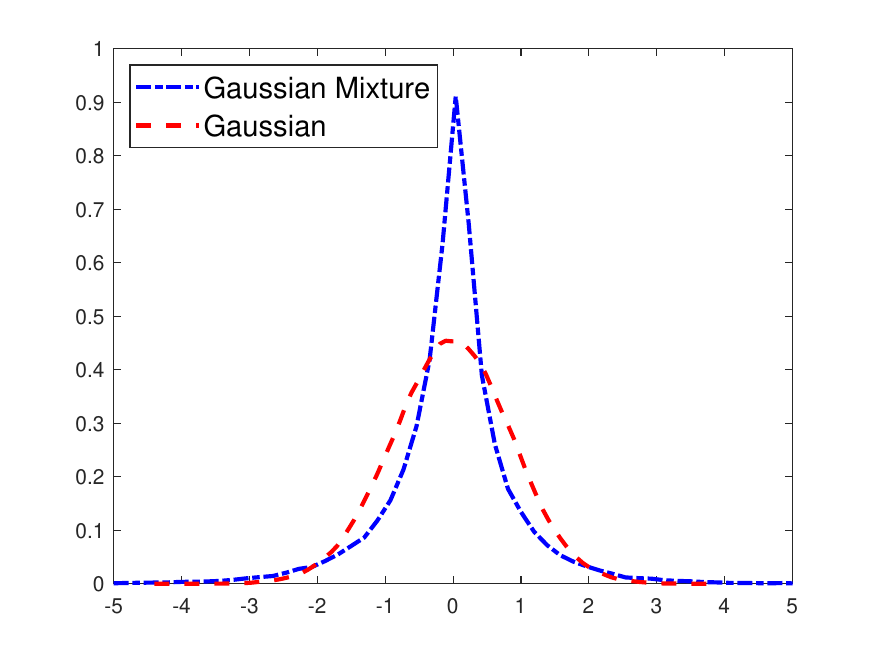}
\caption{$\sqrt{T}(\hat{\phi}_{1}-\phi_{1,0})$.}
\end{subfigure}
\caption{Asymptotic distributions of the Gaussian MLE under different error distributions: $u_{t}\thicksim\text{i.i.d.}~\mathcal{N}(0,1)$ or $u_{t}\thicksim\text{i.i.d.}~\mathfrak{B}_t\mathcal{N}(0,0.05^2)+(1-\mathfrak{B}_t)\mathcal{N}(0,9.9876^2)$, where $\mathfrak{B}_t\thicksim\text{i.i.d.}~Bernoulli(0.99)$. Data generating process: $y_{t}=\left[y_{t-1}+0.5\Delta y_{t-1}+u_{t}\right]_{+}$, $y_{0}=y_{-1}=0$, $T=10^4$, $MC=10^5$ replications.}
\label{fig:MLEvsQMLE_mixture}
\end{figure}

\subsection{Censored least absolute deviations}\label{sec_CLAD}

The censored least absolute deviations (CLAD) estimator of (\ref{eq:tobitark}), which we denote as $(\hat{\alpha}_{T}^{\lad}, \hat{\beta}_{T}^{\lad}, \hat{\vec{\phi}}{}_T^{\lad})$,  minimises
\begin{equation}\label{eq_LAD_min}
  S_T(\alpha,\beta, \vec{\phi}) \defeq \sum_{t=1}^T  \left| y_t-
 \left[\alpha+\beta y_{t-1}+\vec{\phi}^{\trans}\Delta\vec y_{t-1}\right]_+\right|
\end{equation}
as per \citet{Pow84JoE}.
The presence of the positive part -- reflecting Tobit-type censoring as in \eqref{eq:tobitark} -- within the objective function \eqref{eq_LAD_min} renders the criterion function non-convex, making the analysis of the CLAD estimator rather more challenging than that of its uncensored counterpart. A further complication arises from the high persistence in $\{y_t\}$, which prevents the application of maximal inequalities or uniform central limit theorems appropriate to stationary processes.

\assumpname{C}
\begin{assumption}
\label{ass:LAD}~
\begin{enumerate}[label=\ass{\arabic*.}, ref=\ass{.\arabic*}]
\item $u_{t}$ is continuously distributed, with density $f_{u}$ that is positive at zero, bounded,
and continuously differentiable with bounded derivatives, $\med(u_{t})=0$
and $\delta_{u}>2$ in \ref{ass:MOM}; and
\item $(\alpha_0,\beta_0,\vec{\phi}^{\trans}_0)^{\trans}\in\Pi$, for some compact
$\Pi\subset\reals^{k+1}$, and $(0,1,\vec{\phi}_{0}^{\trans})^{\trans}\in\intr\Pi$.
\end{enumerate}
\end{assumption}
The zero-median restriction in Assumption~\ref{ass:LAD} is imposed
in addition to the zero-mean restriction in
Assumption~\ref{ass:DGP}\ref{enu:DGP:ut}.
This ensures that the conditional median of $\{y_t\}$ takes a simple form, as
\begin{equation}\label{eq:ymedian}
\med(y_t\mid\filt_{t-1})=[\alpha_0+\beta_0 y_{t-1}+\vec{\phi}_0^{\trans}\Delta\vec y_{t-1}]_+,
\end{equation}
even in the presnce of censoring; here $\filt_t$ denotes the $\sigma$-field generated by $\{y_{-k+1},\ldots,y_0\}$ and $\{u_s\}_{s=1}^{t}$.
The CLAD estimator uses this restriction to identify the coefficients
$(\alpha_0,\beta_0,\vec{\phi}_0)$ in \eqref{eq:tobitark}. The conditional
mean of the observed process is by contrast
\[
\expect_{t-1}y_t
  =\int_{\reals}[\alpha_0+\beta_0 y_{t-1}+\vec{\phi}_0^{\trans}\Delta\vec y_{t-1}+u]_+f_u(u)\diff u,
\]
which generally differs from $[\alpha_0+\beta_0 y_{t-1}+\vec{\phi}_0^{\trans}\Delta\vec y_{t-1}]_+$ and depends on the
innovation distribution. Thus, under the present assumptions, in the local-to-unity regime CLAD,
MLE and OLS are consistent for exactly the same parameters values, even though the conditional
mean and median of the observed process $y_t$ need not coincide.

We expect that our requirement that $u_t$ have a little more than four finite moments ($\delta_{u}>2$) could be relaxed to two finite moments. Obtaining the limiting distribution in the latter case would necessitate a sharper bound on the order of
\begin{equation*}
\sum_{t=1}^{T}\indic\{\alpha_{0}+\beta_{0}y_{t-1}+\vec{\phi}_{0}^{\trans}\Delta\vec y_{t-1}\leq0\},
\end{equation*}
which is closely related to the rate at which censored observations occur. The arguments used to prove the following theorem rely on an estimate of $o_p(T)$ for such terms, but this could likely be improved to  $O_p(T^{1/2})$.

\begin{thm}\label{thm:tobitLAD}
Suppose that Assumptions \ref{ass:INIT}-\ref{ass:JSR} and \ref{ass:LAD} hold. Then
\begin{equation}\label{LAD_limits_LUR1}
\begin{bmatrix}T^{1/2}(\hat{\alpha}_{T}^{\lad}-\alpha_{0})\\
T(\hat{\beta}_{T}^{\lad}-\beta_{0})
\end{bmatrix}
\xrightarrow[T\to\infty]{d}  \frac1{2f_u(0)}\begin{bmatrix}1 & \int_{0}^{1}Y(\tau)\diff\tau\\
\int_{0}^{1}Y(\tau)\diff\tau & \int_{0}^{1}Y^{2}(\tau)\diff\tau
\end{bmatrix}^{-1}\begin{bmatrix}\widetilde{W}(1)\\
\int_0^1 Y(\tau)d\widetilde{W}(\tau)
\end{bmatrix}
\end{equation}
jointly with, and independently of
\begin{equation}\label{LAD_limits_LUR2}
\sqrt{T}(\hat{\vec{\phi}}{}_T^{\lad}-\vec{\phi}_0)
\xrightarrow[T\to\infty]{d}  \frac1{2f_u(0)}\mathcal{N}(0,\Omega^{-1}),
\end{equation}
where $\Omega=\plim_{T\to\infty}\frac1{T}\sum_{t=1}^{T}\Delta\mathbf{y}_{t-1}\Delta\mathbf{y}_{t-1}^{\trans}$ and $\widetilde{W}(\cdot)$ is a $1$-dimensional Brownian motion, such that the covariance matrix of $(\sigma_0 W(\cdot),\widetilde{W}(\cdot))^{\trans}$ is
$$\Sigma=\begin{pmatrix}
          \sigma_0^2 & \mathbb{E}|u_t|  \\
          \mathbb{E}|u_t| & 1
        \end{pmatrix}.$$
\end{thm}


The limiting distribution for the CLAD estimator closely mirrors that of the (uncensored) LAD estimator in the linear regression setting (\citealp[Theorem~1]{Herce96ET}), with the only difference being the limiting $Y$ processes that appear in those two cases. We may also compare our results for the local-to-unity setting with their stationary counterparts (\citealp[Theorem 5]{stat_paper}, and \citealp[Theorem 6]{bykh_JBES}). In both cases, the limiting distributions feature the same factor $\frac{1}{2f_u(0)}$; the density of the errors at zero, which governs the behavior of the estimator when $y_{t-1}$ is close to zero,
is imparted here, in both cases, by the Hessian of the limiting criterion function.
Indeed, the asymptotics for the short-run coefficients $\hat{\vec{\phi}}^{\lad}$ are exactly the same as under stationarity, since in our setting $y_t=0$ occurs sufficiently infrequently that $\Omega$ is also the probabilty limit of $\frac1{T}\sum_{t=1}^{T}\indic\{y_{t}>0\}\Delta\mathbf{y}_{t-1}\Delta\mathbf{y}_{t-1}^{\trans}$.


Deriving the limiting distribution of the CLAD estimator involves
two key steps: the first establishes their rates of convergence, while
the second yields a quadratic approximation to the centred CLAD
criterion function $S_T$ on a shrinking neighbourhood of the true
parameters. Both steps are here made more challenging by the
strong dependence and nonstationarity of the data. This is particularly evident
in the fact that $\frac{1}{T}S_T$ need not converge to its expectation $\frac{1}{T}\mathbb{E}S_T$,
as it would in the stationary and ergodic case. The role of the `population criterion function' in our arguments
must instead be played by the predictable component
\begin{equation*}
\bar{S}_T(\alpha,\beta, \vec{\phi}) \defeq
\sum_{t=1}^{T}\expect_{t-1}\smlabs{y_{t}-[\alpha+\beta y_{t-1}+\vec{\phi}^{\trans}\Delta\vec y_{t-1}]_{+}},
\end{equation*}
which, though random, is a.s.\ minimised at the true parameters, and differs from $S_T$ by a martingale term.

In the first step,  $\frac{1}{T}\bar{S}_T$ is shown
to have a well-separated minimum (at the truth), deviations from which are lower bounded by a quadratic function on a (fixed)
neighbourhood in the rescaled coordinates (Proposition~\ref{prop:centring}). We also need to control the oscillations of the martingale
$\frac{1}{T}(S_T-\bar{S}_T)$ over that same neighbourhood.
If we were willing to assume that Assumption~\ref{ass:LAD} held for some sufficiently large $\delta_u$, we could directly bound sufficiently
high moments of the martingale increments across parameter values to control these oscillations. Here, a stopping time argument allows us to effectively follow this approach while avoiding any such requirement on $\delta_u$ (Proposition~\ref{prop:mg}).
The rates of convergence of the estimators then follow by an adaptation of the proof of Theorem~3.2.5 in \citet{VVW96}.

For the second step, we establish a quadratic approximation to the
CLAD criterion $S_T$ on a neighbourhood of the true parameters that
shrinks proportionately with the rates of convergence obtained in the previous step.
The main difficulties concern the uniformity of the approximation
and the (joint) limiting behaviour of its coefficients.
The former involves a second application of the arguments developed
in the first step for controlling the oscillations of $S_T-\bar{S}_T$.
Regarding the terms in the approximation, the quadratic term
derives from the predictable component $\bar{S}_T$, with $f_u(0)$ multiplying
the limiting regressor moment matrix. The linear (i.e.\ score)
term is associated with the martingale $S_T-\bar{S}_T$, evaluated at the true parameters, and the asymptotics
rely importantly on the martingale difference property of $\sgn(u_t)$. Here we combine
the martingale CLT for the weakly
persistent components in \eqref{LAD_limits_LUR2} with theorems on
convergence to stochastic integrals for the more persistent components
in \eqref{LAD_limits_LUR1}. These limits are established jointly with
the partial sums of $u_t$, thereby determining the joint distribution
of all components, including the dependence between $W$ and
$\widetilde W$. An application of the argmin continuous mapping
theorem identifies the limiting distribution of the estimator with that of the minimiser
of the limiting quadratic criterion, and thus completes the proof.

We conclude the section by considering how the LAD results might extend to
dependent innovations. As noted in Section~\ref{sec_model} above, the i.i.d.\ assumption
simplifies the analysis and allows auxiliary results from \citet{BD22}
to be applied directly. The CLAD estimator would nevertheless remain
consistent if the i.i.d.\ assumption on $\{u_t\}$ were relaxed so as
to allow it to be a stationary and ergodic m.d.s., subject to suitable
moment and density conditions. In that case, our conditions on the
(unconditional) density $f_u$ would need to be transposed to the
(random) \emph{conditional} density $f_{u\mid\filt_{t-1}}$ of $u_t$
given $\filt_{t-1}$, with in particular that density being assumed
to have zero median (a.s.), i.e.
\[
\med(u_t\mid\filt_{t-1})=0\qquad\text{a.s.}.
\]
so that the conditional median of $y_t$ remained as in \eqref{eq:ymedian} above.

Inspection of the proof suggests that the consistency of the CLAD
estimator would still obtain under these assumptions. Moreover, at least
for stationary GARCH innovations with symmetric i.i.d.\ shocks, it appears
likely that the limiting distributions
of the estimators of the `long run' parameters $(\alpha,\beta)$ would
remain unchanged from \eqref{LAD_limits_LUR1}, while those of the
`short run' parameters $\vec{\phi}$ would remain Gaussian, and
independent of the long-run limit, though with a covariance matrix
that generally differs from \eqref{LAD_limits_LUR2} due to the
dependence between $f_{u\mid\filt_{t-1}}(0)$ and lags of $\Delta y_t$.
In these expressions, $\sigma_0^2$ would correspond to the \emph{unconditional} variance, and
$f_u$ is the unconditional density, of $u_t$. Formally
establishing these limits, and considering more general cases,
is deferred to future research.



\section{Model lag order selection}\label{sec_modelselection}
Obtaining a correctly specified model is important both for inference and for forecasting. In empirical applications, having too few lags may leave serial dependence in the innovations, while including unnecessary lags may impair estimator precision and forecasting accuracy. The results in Section \ref{sec_estimation} provide a basis for selecting the lag order $k$ of the dynamic Tobit by sequential testing. Although the limiting distributions associated with $(\alpha_0,\beta_0)$ are nonstandard, the $t$-statistics for the coefficients $\phi_i$ are asymptotically standard normal under both Gaussian MLE and CLAD, rendering them amenable to this purposes. (By constrast, this asyptotic normality is \emph{not} exhibited by the OLS $t$-statistics, nor indeed by those of the Gaussian MLE when the errors are not Gaussian.)

In this section, we consider a procedure in which one begins with a conservatively large (but fixed) lag order $k=\bar{k}$, and proceeds to test the null hypothesis $H_0:\phi_{k-1}=0$, reducing $k$ by one whenever the absolute value of the corresponding $t$-statistic fails to exceed a specified critical value, re-estimating the model at each step. The procedure halts at the first value of $k$ for which the null is rejected, and that value of $k$ is selected. If no null is rejected for $k\geq2$, the procedure selects $k=1$. Below we present two theorems establishing the consistency of this procedure when the model is estimated by Gaussian MLE or CLAD, under an appropriate sequence of critical values. Let $k_0$ denote the true (minimal) lag order, so that $\phi_{k-1,0}\neq0$ for $k=k_0$ (if $k_0>1$) and $\phi_{k-1,0} = 0$ for all $k \geq k_0$, and let $k^*$ denote the order selected by the procedure. By consistency of the selection procedure, we mean that
$$
  \Prob(k^*=k_0)\to1
  \quad\text{as}\quad T\to\infty.
$$


If fixed standard normal critical values are employed -- such as $1.96$ for a two-sided $5\%$ test -- the procedure will be conservative rather than consistent, in the sense of having a nonzero asymptotic probability of selecting $k>k_0$, because the asymptotic size of the test is non-zero. To obtain a consistent procedure, we need instead to choose a \emph{sequence} of critical values $cv(T)$ such that the size of the $t$-test converges to zero, which obtains for any $cv(T)$ diverging with the sample size. At the same time, we must also retain the ability to discriminate between zero and non-zero coefficients, so that $H_0:\phi_{k-1}=0$ is still correctly rejected under fixed alternatives, with probability approaching one. Balancing these two requirements leads to a sequence of critical values that satisfy\footnote{The use of diverging critical values to obtain consistent order selection dates back to the 1980s; see, e.g., \citet{potscher1983order} for sequential tests in ARMA models.}
\begin{equation}\label{eq_cvT}
  cv(T)\to\infty,\qquad \frac{cv(T)}{\sqrt T}\to0.
\end{equation}

Calculation of the $t$-statistics requires consistent estimation of the terms entering the asymptotic variances of the estimators. The estimator
\begin{equation}\label{eq_Omega_hat}
  \hat{\Omega}_T=T^{-1}\sum_{t=1}^{T}(\Delta\vec y_{t-1})(\Delta\vec y_{t-1})^{\trans},
\end{equation}
is consistent for $\Omega$ by Lemma \ref{lem:bd2022}\ref{enu:deltayt}. This consistency result continues to hold for any fixed $k\geq k_0$, since a model of order $k_0$ can always be embedded in a model of a higher order by setting $\phi_j=0$ for all $j\geq k_0$. MLE additionally yields a consistent estimator of the error variance, $(\hat{\sigma}_{T}^{\mle})^{2}$. By the same logic, its consistency holds for any fixed $k\geq k_0$. CLAD, in turn, additionally requires a consistent estimator of $f_u(0)$: see Section~\ref{sec_guid} for details.
\begin{thm}\label{thm:ttest_MLE}
Let $k_0\leq\bar{k}<\infty$. Suppose Assumptions~\ref{ass:INIT}--\ref{ass:JSR} and \ref{ass:TOB} hold. If the critical values $cv(T)$ satisfy \eqref{eq_cvT}, then the sequential $t$-test model-selection procedure based on MLE with critical values $cv(T)$ is consistent.
\end{thm}
\begin{thm}\label{thm:ttest_LAD}
Let $k_0\leq\bar{k}<\infty$. Suppose that Assumptions~\ref{ass:INIT}-\ref{ass:JSR} and \ref{ass:LAD} hold, and suppose that $\hat{f}_{u,k,T}$ is a consistent estimator of $f_u(0)$ for all $k\in[k_0,\bar{k}]$. If the critical values $cv(T)$ satisfy \eqref{eq_cvT}, then the sequential $t$-test model-selection procedure based on CLAD with critical values $cv(T)$ is consistent.
\end{thm}

Theorems \ref{thm:ttest_MLE} and \ref{thm:ttest_LAD} suggest that one can use a BIC-inspired critical value
$$cv(T)=\sqrt{\log T}.$$
Table~\ref{tab:lag_selection_MC} compares the sequential procedures
in a model with true order $k_0=2$, starting the search at
$\bar{k}=6$. The data are generated as
\[
 y_t=\left[(1+c/T)y_{t-1}+\phi_{1,0}\Delta y_{t-1}+u_t\right]_+,
 \qquad c=-5,0,\qquad \phi_{1,0}=0.2,0.5,
\]
with zero initial observations, $T=500,1000$, and $1000$
replications of each design. For Panel A, $u_t\sim N(0,1)$;
for Panel B, its distribution is $0.8N(0,0.5^2)+0.2N(0,2^2)$,
so that the variance is again one but the Gaussian likelihood
is misspecified. Both $cv(T)=1.96$ and
$cv(T)=\sqrt{\log T}$ are applied to the same estimated models,
re-estimating the coefficients and nuisance quantities at every
order. The entries U, C and O give the frequencies of selecting
$k^*=1$, $k^*=2$ and $k^*>2$, respectively.

\begin{table}[t]
\centering
\small
\setlength{\tabcolsep}{3pt}
\begin{tabular}{ccc*{4}{rrr}}
\toprule
 & & & \multicolumn{6}{c}{$cv(T)=1.96$}
     & \multicolumn{6}{c}{$cv(T)=\sqrt{\log T}$} \\
\cmidrule(lr){4-9}\cmidrule(lr){10-15}
 & & & \multicolumn{3}{c}{ML} & \multicolumn{3}{c}{CLAD}
     & \multicolumn{3}{c}{ML} & \multicolumn{3}{c}{CLAD} \\
\cmidrule(lr){4-6}\cmidrule(lr){7-9}
\cmidrule(lr){10-12}\cmidrule(lr){13-15}
$T$ & $c$ & $\phi_1$
 & U & C & O & U & C & O & U & C & O & U & C & O \\
\midrule
\multicolumn{15}{l}{\textit{Panel A: Gaussian innovations}} \\
\addlinespace
500 & $-5$ & 0.2 & 0.005 & 0.768 & 0.227 & 0.062 & 0.679 & 0.259 & 0.034 & 0.891 & 0.075 & 0.174 & 0.723 & 0.103 \\
 &  & 0.5 & 0.000 & 0.770 & 0.230 & 0.000 & 0.700 & 0.300 & 0.000 & 0.940 & 0.060 & 0.000 & 0.897 & 0.103 \\
 & $0$ & 0.2 & 0.007 & 0.786 & 0.207 & 0.051 & 0.697 & 0.252 & 0.030 & 0.919 & 0.051 & 0.154 & 0.749 & 0.097 \\
 &  & 0.5 & 0.000 & 0.814 & 0.186 & 0.000 & 0.731 & 0.269 & 0.000 & 0.945 & 0.055 & 0.000 & 0.904 & 0.096 \\
\addlinespace
1000 & $-5$ & 0.2 & 0.000 & 0.822 & 0.178 & 0.003 & 0.781 & 0.216 & 0.000 & 0.965 & 0.035 & 0.014 & 0.929 & 0.057 \\
 &  & 0.5 & 0.000 & 0.780 & 0.220 & 0.000 & 0.747 & 0.253 & 0.000 & 0.954 & 0.046 & 0.000 & 0.937 & 0.063 \\
 & $0$ & 0.2 & 0.000 & 0.804 & 0.196 & 0.002 & 0.773 & 0.225 & 0.000 & 0.963 & 0.037 & 0.015 & 0.934 & 0.051 \\
 &  & 0.5 & 0.000 & 0.808 & 0.192 & 0.000 & 0.801 & 0.199 & 0.000 & 0.960 & 0.040 & 0.000 & 0.954 & 0.046 \\
\midrule
\multicolumn{15}{l}{\textit{Panel B: Gaussian mixture innovations}} \\
\addlinespace
500 & $-5$ & 0.2 & 0.007 & 0.796 & 0.197 & 0.000 & 0.722 & 0.278 & 0.030 & 0.902 & 0.068 & 0.002 & 0.908 & 0.090 \\
 &  & 0.5 & 0.000 & 0.769 & 0.231 & 0.000 & 0.737 & 0.263 & 0.000 & 0.932 & 0.068 & 0.000 & 0.891 & 0.109 \\
 & $0$ & 0.2 & 0.007 & 0.767 & 0.226 & 0.000 & 0.729 & 0.271 & 0.027 & 0.903 & 0.070 & 0.001 & 0.911 & 0.088 \\
 &  & 0.5 & 0.000 & 0.827 & 0.173 & 0.000 & 0.734 & 0.266 & 0.000 & 0.951 & 0.049 & 0.000 & 0.903 & 0.097 \\
\addlinespace
1000 & $-5$ & 0.2 & 0.000 & 0.790 & 0.210 & 0.000 & 0.778 & 0.222 & 0.000 & 0.962 & 0.038 & 0.000 & 0.944 & 0.056 \\
 &  & 0.5 & 0.000 & 0.800 & 0.200 & 0.000 & 0.760 & 0.240 & 0.000 & 0.963 & 0.037 & 0.000 & 0.939 & 0.061 \\
 & $0$ & 0.2 & 0.000 & 0.783 & 0.217 & 0.000 & 0.790 & 0.210 & 0.000 & 0.953 & 0.047 & 0.000 & 0.941 & 0.059 \\
 &  & 0.5 & 0.000 & 0.778 & 0.222 & 0.000 & 0.774 & 0.226 & 0.000 & 0.963 & 0.037 & 0.000 & 0.950 & 0.050 \\
\bottomrule
\end{tabular}
\caption{Lag-order selection frequencies. Entries are the frequencies of under-selection (U),
correct selection (C), and over-selection (O), based on $1{,}000$
replications. The data generating process is
$y_t=[(1+c/T)y_{t-1}+\phi_1\Delta y_{t-1}+u_t]_+$, $y_0=y_{-1}=0$. The true lag order is $k_0=2$,
and the sequential procedure starts from $\bar{k}=6$.
Innovations are i.i.d.\ $N(0,1)$ in Panel A and follow
$0.8N(0,0.5^2)+0.2N(0,2^2)$ in Panel B. ML denotes Gaussian maximum likelihood,
interpreted as quasi-maximum likelihood in Panel B.\label{tab:lag_selection_MC}}
\end{table}

With a fixed critical value, over-selection remains frequent
even at $T=1000$, occurring in roughly one fifth to one quarter
of the replications. The diverging critical value substantially
reduces this frequency. At $T=1000$, it yields correct selection
in $95.3$--$96.5\%$ of replications for ML/QML and
$92.9$--$95.4\%$ for CLAD. Raising the critical value does,
however, make it harder to detect a small nonzero coefficient.
This is most apparent for CLAD under Gaussian innovations when
$T=500$ and $\phi_{1,0}=0.2$: the frequency of under-selection
is $15.4$--$17.4\%$, falling to $1.4$--$1.5\%$ at $T=1000$.
Thus, both procedures perform well with the diverging critical
value in the larger samples, while detecting the smaller
short-run coefficient by CLAD remains more difficult in the
smaller Gaussian samples.

\section{Discussion and practical recommendations}\label{sec_discussion}

\subsection{Comparison of estimators: OLS, MLE, and CLAD}\label{sec_OLS_MLE_LAD}

Under our assumptions, OLS, CLAD and the Gaussian MLE are consistent for the same parameter values, in the local-to-unity regime. Their comparison therefore concerns the precision of estimation and the validity of inference under different innovation distributions, as well as their behaviour in the stationary regime.


We first compare MLE and CLAD. When $u_t$ is Gaussian, both Theorems~\ref{thm:tobitmle} and \ref{thm:tobitLAD} are applicable, allowing us to compare the variances of the short-run coefficients in MLE and CLAD. In this case, $\tfrac{1}{2f_u(0)} = \sigma_0\sqrt{\pi/2} \approx 1.25\sigma_0$, implying roughly a $25\%$ increase in the standard deviation of CLAD relative to the (optimal) ML estimator. The increased asymptotic variance of all CLAD coefficients compared to MLE is illustrated in Figure~\ref{fig:LADvsMLE_gaussian}. We observe that, because the asymptotic distributions of both CLAD and MLE estimators of $\alpha_0$ and $\beta_0$ involve the non-standard limiting process $Y(\cdot)$, they are skewed and centered at positive values for $\alpha_0$ and negative values for $\beta_0$, with the MLE exhibiting tighter concentration than CLAD. In contrast, the short-run estimators are all centered at zero and approximately normal, though again with CLAD displaying higher variance than MLE.

On the other hand, relative to the (Gaussian) MLE, a major advantage of the CLAD estimator is that it is semiparametric, insofar as it does not require a parametric assumption on the distribution of the errors $u_t$. As illustrated in \citet{stat_paper} and \citet{bykh_JBES}, real-world data can deviate substantially from Gaussianity, in which case the CLAD estimator is likely to yield more reliable inferences.

\begin{figure}[t]
\begin{subfigure}{.3\textwidth}
\centering \includegraphics[width=1\linewidth]{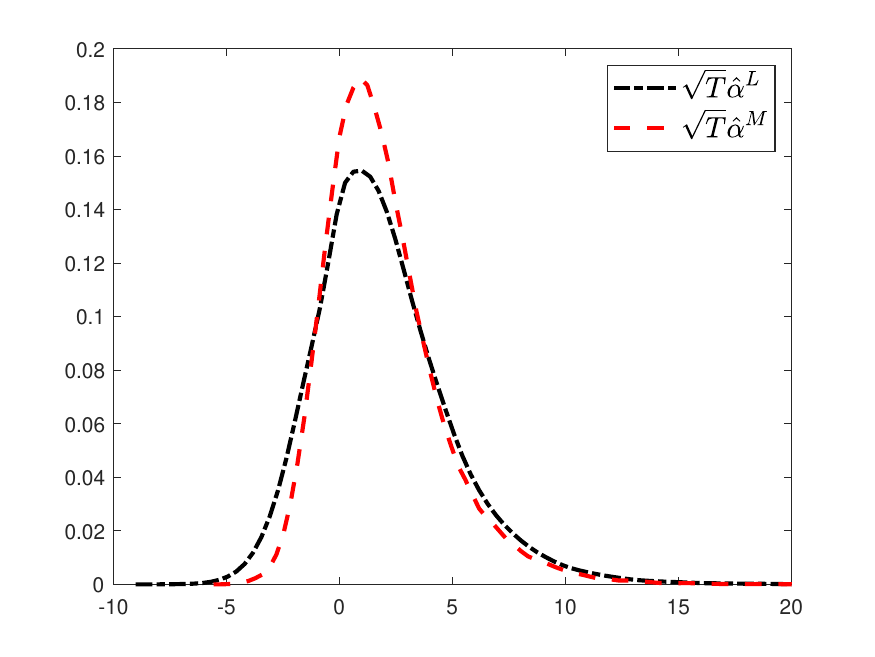}
\caption{$\sqrt{T}\hat{\alpha}$.}
\end{subfigure}
\begin{subfigure}{.3\textwidth}
\centering \includegraphics[width=1\linewidth]{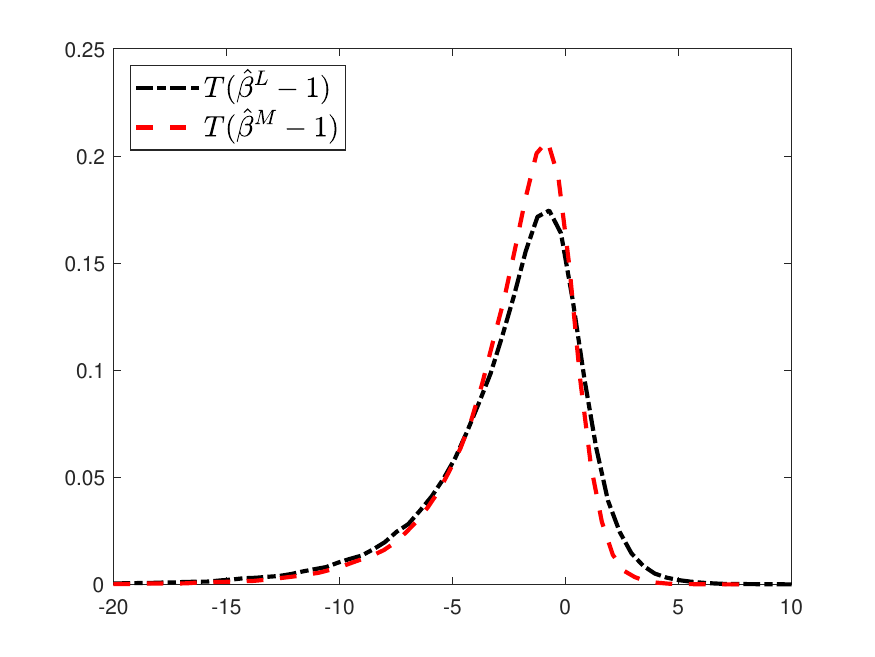}
\caption{$T(\hat{\beta}-1)$.}
\end{subfigure}
\begin{subfigure}{.3\textwidth}
\centering \includegraphics[width=1\linewidth]{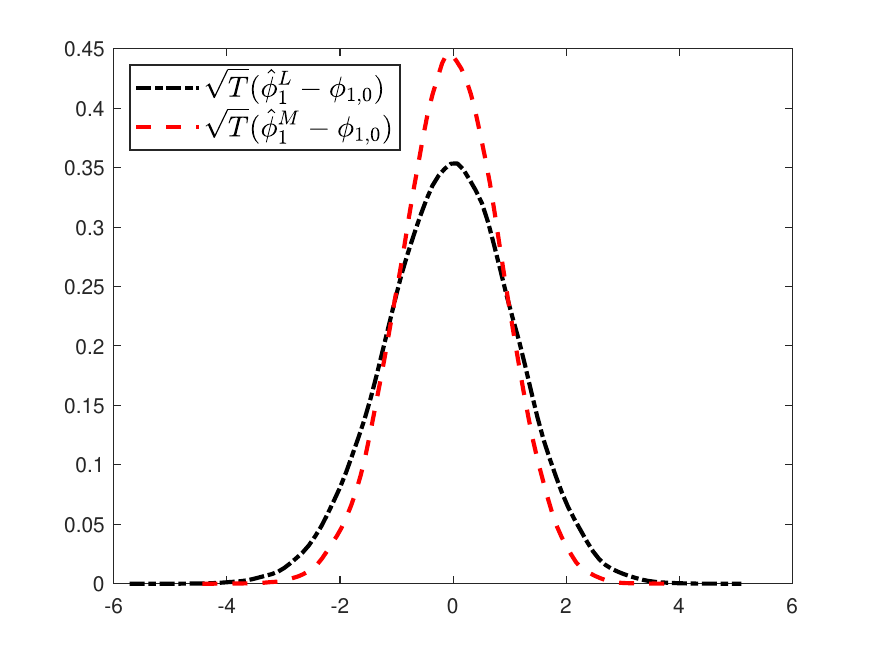}
\caption{$\sqrt{T}(\hat{\phi}_1-\phi_{1,0})$.}
\end{subfigure}
\caption{Asymptotic distributions based on MLE and CLAD. Data generating process: $y_{t}=\left[y_{t-1}+0.5\Delta y_{t-1}+u_{t}\right]_{+}$,
$y_{0}=y_{-1}=0$, $u_{t}\thicksim\text{i.i.d.}~\mathcal{N}(0,1)$, $T=1000$, $MC=100,000$ replications.}
\label{fig:LADvsMLE_gaussian}
\end{figure}

\begin{figure}[t]
\begin{subfigure}{.3\textwidth}
\centering \includegraphics[width=1\linewidth]{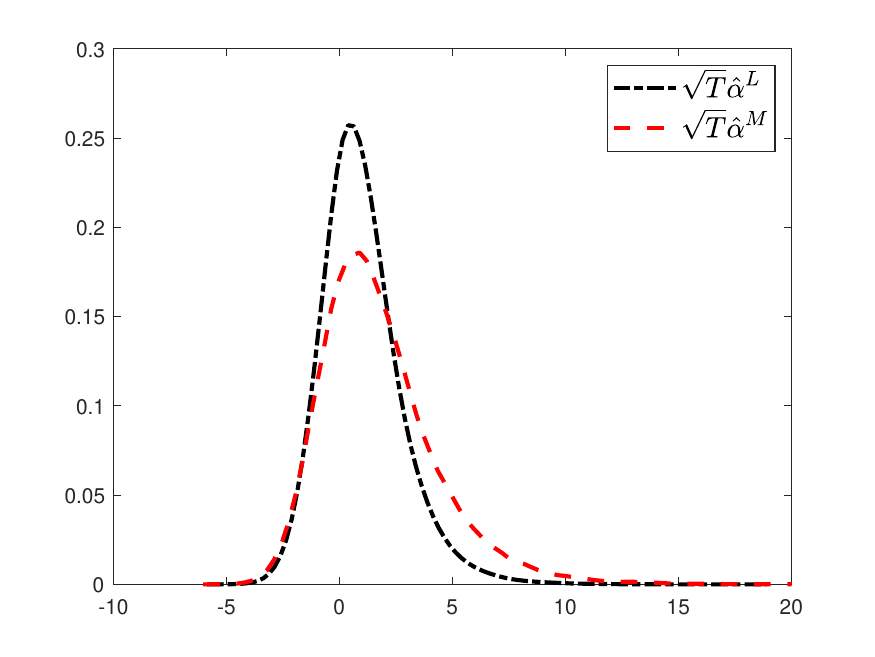}
\caption{$\sqrt{T}\hat{\alpha}$.}
\end{subfigure}
\begin{subfigure}{.3\textwidth}
\centering \includegraphics[width=1\linewidth]{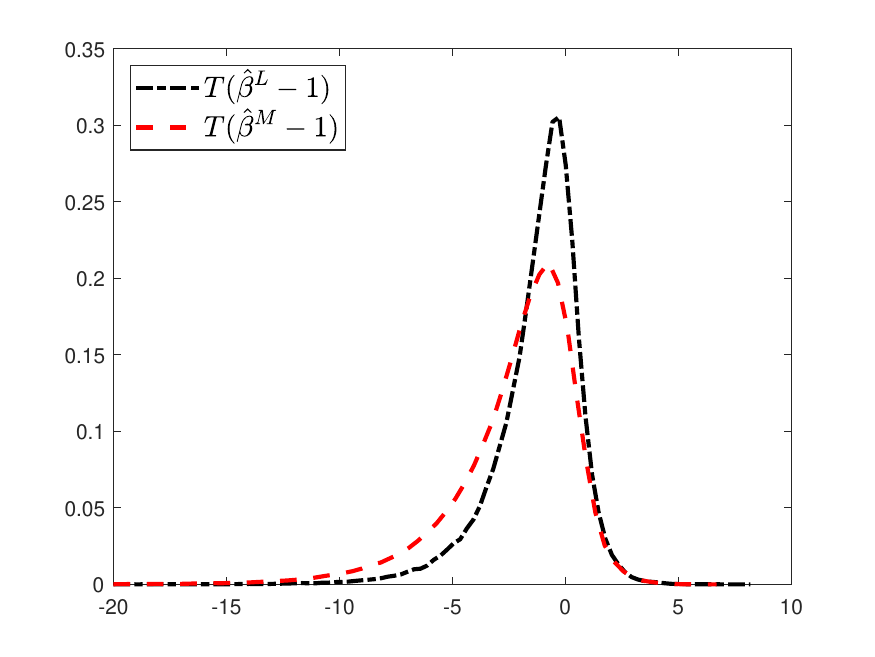}
\caption{$T(\hat{\beta}-1)$.}
\end{subfigure}
\begin{subfigure}{.3\textwidth}
\centering \includegraphics[width=1\linewidth]{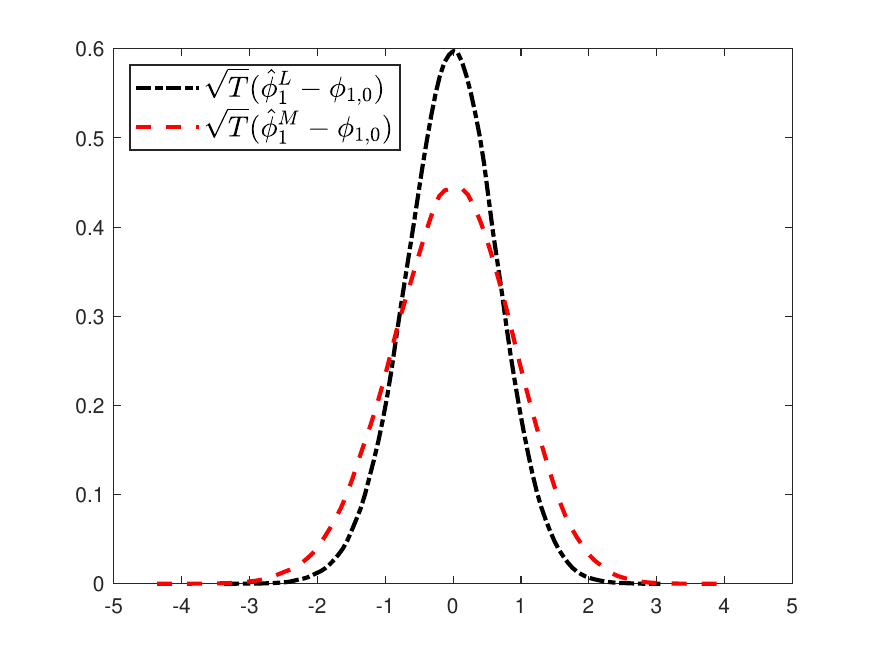}
\caption{$\sqrt{T}(\hat{\phi}_1-\phi_{1,0})$.}
\end{subfigure}
\caption{Asymptotic distributions based on Gaussian QMLE and CLAD when errors follow a Laplace distribution. Data generating process: $y_{t}=\left[y_{t-1}+0.5\Delta y_{t-1}+u_{t}\right]_{+}$,
$y_{0}=y_{-1}=0$, $u_{t}\thicksim\text{i.i.d.}~Laplace(0,1/\sqrt{2})$, $T=1000$, $MC=100,000$ replications.}
\label{fig:LADvsMLE_Laplace}
\end{figure}

Theorem~\ref{thm:tobitqmle} shows that the Gaussian QMLE remains consistent under non-Gaussian innovations in the local-to-unity regime, but its limiting distributions generally differ from those under correct specification. Our simulations suggest that the distribution of the QMLE changes little under \emph{moderate} departures from Gaussianity, although this resemblance does not hold for all innovation distributions, as Figure~\ref{fig:MLEvsQMLE_mixture} above illustrates. However, depending on the value of $f_u(0)$, CLAD may now be more efficient than the quasi-MLE. For example, if $u_t$ follows a Laplace$(0,1/\sqrt{2})$ distribution, so that $\mathbb{E}u_t=0$ and $\mathbb{E} u_t^2=1$, we have $f_u(0) = 1/\sqrt{2}$ and $1/(2f_u(0))= 1/\sqrt{2} \approx 0.7 < 1$. In the simulations with $k=2$, the Gaussian QMLE standard deviation for the short-run coefficients is almost $1.5$ times larger than that of CLAD.
This is illustrated in Figure~\ref{fig:LADvsMLE_Laplace}, where the behaviour of the Gaussian QMLE is very similar to that observed for Gaussian errors in Figure~\ref{fig:LADvsMLE_gaussian}, while the CLAD estimator is now much more tightly distributed, reflecting the higher value of the density at zero ($1/\sqrt{2}$ for the Laplace distribution versus $1/\sqrt{2\pi}$ for the standard normal). For Student-$t$ innovations rescaled to have unit variance, the simulations also favour CLAD when $\nu=2.5,3$ or $4$. At $T=1000$, the short-run standard deviations are nearly equal for $\nu=4.5$.

The main drawbacks of the CLAD estimator come from two aspects. First, the CLAD optimisation problem is non-convex, potentially causing numerical optimisers to converge to local rather than global optima, and thus to be sensitive to their initialisation. The Gaussian loglikelihood, by contrast, may be reparameterised so as to render the objective function concave, thereby avoiding this problem.\footnote{As per \citet{Ols78Ecta}, the new parameters are $\alpha/\sigma,\, \beta/\sigma,\,\vec{\phi}/\sigma$, and $\sigma^{-1}$.} (This may in turn be used to provide reasonable starting values for the optimisation of the CLAD criterion function.) The second drawback stems from the necessity of estimating $f_u(0)$ for the purposes of inference (i.e.\ to construct confidence intervals or conduct hypothesis tests). We tackle this by means of kernel density estimation, as proposed in \citet[(5.2) and (5.5)]{Pow84JoE}, though this has the undesirable consequence of rendering inferences dependent on the choice of a bandwidth parameter.

There are two significant advantages to using MLE or CLAD relative to OLS. The first advantage is that, under their respective assumptions, the former two are consistent both under stationary and (near) unit root regimes (see e.g., \citet{stat_paper} and \citet{bykh_JBES}), while OLS is consistent only under the (near) unit root regime (see, e.g., \citet[Supplementary material]{bykh_JBES} for inconsistency results and \citet{BD22} for consistency results). For the Gaussian MLE, this comparison presumes Gaussian innovations. Under LUR, censored observations accumulate only at rate $T^{1/2}$, so that their proportion in the sample tends to zero. Their contribution is small enough for OLS to remain consistent, but need not vanish from the limiting distributions of the normalised estimators. Since an empirical researcher is typically a priori uncertain about the level of persistence in the data, it is useful to have estimators whose consistency does not depend on the local-to-unity parametrisation.

The second advantage is related to the estimation of the `short-run' parameters $\phi_i$, for $i \in \{1,\ldots,k-1\}$, i.e.\ the coefficients on the lagged differences $\Delta y_{t-i}$. As is shown in Theorems~\ref{thm:tobitmle} and \ref{thm:tobitLAD}, under their respective assumptions both the MLE and CLAD estimators of these are asymptotically normal, leading to standard normal $t$-statistics. By contrast, as illustrated in Figure~\ref{fig:phi1}, the asymptotic distributions of $\sqrt{T}(\hat{\phi}_{1,\mathrm{OLS}}-\phi_{1,0})$ and its corresponding $t$-statistic (solid blue curves) are not centered at zero, as a consequence of the censoring. In particular, the mean of the OLS $t$-statistic in Figure~\ref{fig:phi1} is approximately $-0.4$. This reflects the contribution of the censoring term $-\frac1{\sqrt{T}}\sum_t y_t^{-}\Delta y_{t-1}$ to the normalised OLS score, where $y_t^{-}=\min\{0,\alpha_0+\beta_0 y_{t-1}+\vec{\phi}_0^{\trans}\Delta\vec y_{t-1}+u_{t}\}$. The intuition is that the positive-part binds $O_p(\sqrt{T})$ times, so there are $O_p(\sqrt{T})$ nonzero $y_t^{-}$, whose contribution need not vanish when scaled by $\tfrac{1}{\sqrt{T}}$. In contrast, the simulated distributions of the $t$-statistics for the ML (red dashed curve) and CLAD (black dash-dotted curve) estimators align closely with the standard normal density.

\begin{figure}[t]
\begin{subfigure}{.49\textwidth}
\centering \includegraphics[width=1\linewidth]{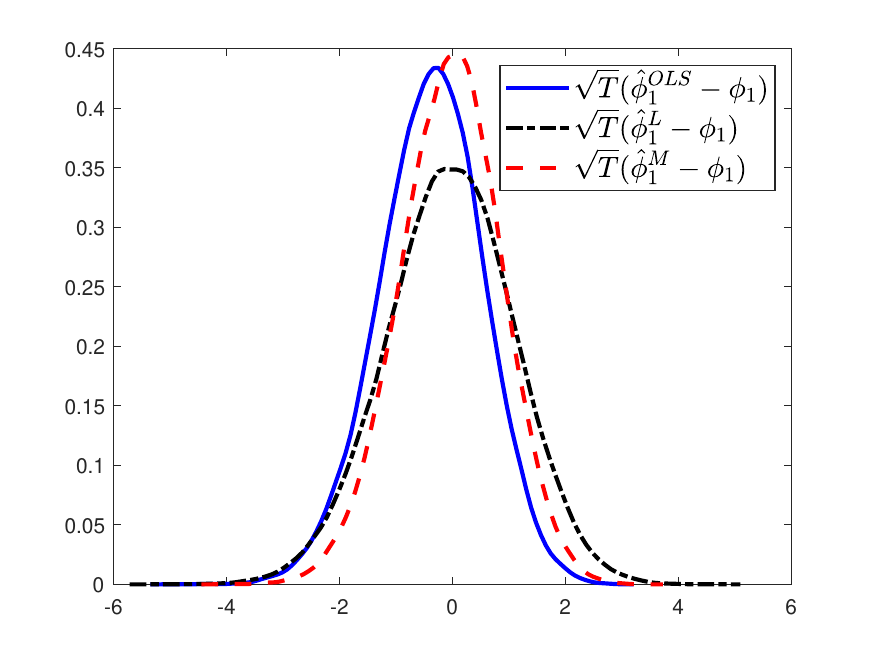}
\caption{$\sqrt{T}(\hat{\phi}_1-\phi_{1,0})$.}
\end{subfigure}
\begin{subfigure}{.49\textwidth}
\centering \includegraphics[width=1\linewidth]{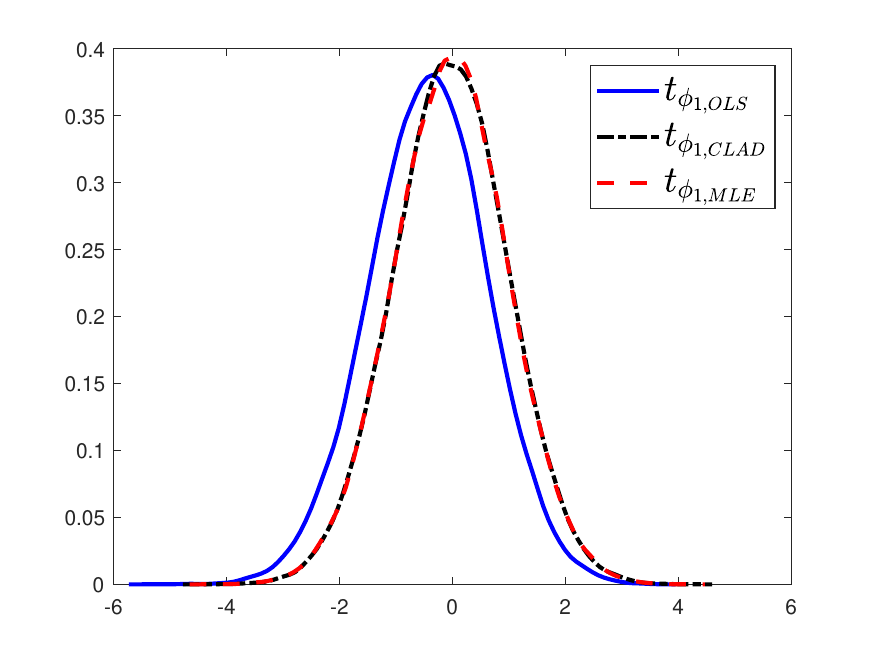}
\caption{$t$-statistic for $\phi_1$.}
\end{subfigure}
\caption{Asymptotic distributions of $\sqrt{T}(\hat{\phi}_1-\phi_{1,0})$ and of the $t$-statistic for $\phi_1$ based on three estimation procedures. Data generating process: $y_{t}=\left[y_{t-1}+0.5\Delta y_{t-1}+u_{t}\right]_{+}$,
$y_{0}=y_{-1}=0$, $u_{t}\thicksim\text{i.i.d.}~\mathcal{N}(0,1)$, $T=1000$, $MC=100,000$ replications.}
\label{fig:phi1}
\end{figure}

Table~\ref{tab:short_run_CI} reports the coverage and mean length
of nominal $95\%$ confidence intervals for $\phi_{1,0}$, for the
same designs as Table~\ref{tab:lag_selection_MC}. There are
$1000$ replications of each design, with $T=500,1000$. We fit
the model with $k=2$ throughout, estimating all three coefficients,
so that lag selection does not enter this comparison. The ML and
CLAD intervals use the standard errors described in the
practitioner's guide in Section~\ref{sec_guid} below, with $d=1.06$ in the CLAD bandwidth.
To see how much difference the estimation of $f_u(0)$ makes,
we also construct CLAD intervals with the true density in place
of its estimate; the remaining estimates are kept the same.

\begin{table}[t]
\centering
\small
\setlength{\tabcolsep}{3pt}
\begin{tabular}{rrr*{4}{rr}}
\hline
 & & &
 \multicolumn{2}{c}{OLS} &
 \multicolumn{2}{c}{ML/QML} &
 \multicolumn{2}{c}{CLAD: $\hat f_u(0)$} &
 \multicolumn{2}{c}{CLAD: $f_u(0)$} \\
$T$ & $c$ & $\phi_{1,0}$ &
Cov. & Length & Cov. & Length &
Cov. & Length & Cov. & Length \\
\hline
\multicolumn{11}{c}{Panel A: Gaussian innovations} \\
\hline
500 & -5 & 0.2 & 92.5 & 0.171 & 93.6 & 0.181 & 92.2 & 0.227 & 93.8 & 0.228 \\
500 & -5 & 0.5 & 86.6 & 0.154 & 94.6 & 0.163 & 94.8 & 0.203 & 95.6 & 0.205 \\
500 & 0 & 0.2 & 93.8 & 0.171 & 95.5 & 0.176 & 94.1 & 0.220 & 94.8 & 0.222 \\
500 & 0 & 0.5 & 89.9 & 0.153 & 93.3 & 0.157 & 91.1 & 0.196 & 93.2 & 0.198 \\
\hline
1000 & -5 & 0.2 & 90.9 & 0.121 & 94.8 & 0.126 & 94.3 & 0.158 & 95.2 & 0.159 \\
1000 & -5 & 0.5 & 86.9 & 0.109 & 95.7 & 0.113 & 94.5 & 0.142 & 95.3 & 0.142 \\
1000 & 0 & 0.2 & 94.1 & 0.121 & 94.2 & 0.124 & 93.5 & 0.155 & 94.7 & 0.156 \\
1000 & 0 & 0.5 & 91.4 & 0.108 & 94.6 & 0.110 & 95.0 & 0.138 & 94.5 & 0.138 \\
\hline
\multicolumn{11}{c}{Panel B: Gaussian mixture innovations} \\
\hline
500 & -5 & 0.2 & 92.6 & 0.171 & 93.6 & 0.181 & 92.0 & 0.138 & 93.3 & 0.139 \\
500 & -5 & 0.5 & 81.7 & 0.155 & 93.1 & 0.163 & 91.9 & 0.125 & 93.1 & 0.125 \\
500 & 0 & 0.2 & 94.3 & 0.171 & 95.1 & 0.176 & 94.6 & 0.133 & 94.8 & 0.133 \\
500 & 0 & 0.5 & 90.3 & 0.153 & 93.6 & 0.157 & 92.8 & 0.119 & 93.7 & 0.118 \\
\hline
1000 & -5 & 0.2 & 92.8 & 0.121 & 95.0 & 0.126 & 92.2 & 0.095 & 93.8 & 0.095 \\
1000 & -5 & 0.5 & 79.2 & 0.109 & 93.2 & 0.113 & 93.5 & 0.085 & 94.1 & 0.085 \\
1000 & 0 & 0.2 & 94.1 & 0.121 & 94.8 & 0.124 & 95.3 & 0.093 & 94.7 & 0.093 \\
1000 & 0 & 0.5 & 91.9 & 0.108 & 95.9 & 0.110 & 94.9 & 0.082 & 95.1 & 0.082 \\
\hline
\end{tabular}
\caption{Coverage and mean length of nominal $95\%$ confidence
intervals for $\phi_{1,0}$. The DGP is
$y_t=[(1+c/T)y_{t-1}+\phi_{1,0}\Delta y_{t-1}+u_t]_+$, $y_0=y_{-1}=0$, $MC=1000$. All models are fitted with $k=2$, estimating
$(\alpha,\beta,\phi_1)$. Innovations are i.i.d.\ $N(0,1)$
in Panel A and $0.8N(0,0.5^2)+0.2N(0,2^2)$ in Panel B.
Cov.\ is coverage in percent; Length is mean interval length.
Every interval uses the standard normal critical value $1.96$.
OLS uses $\hat\sigma_{\mathrm{OLS}}^2\hat\Omega_T^{-1}/T$,
with residual variance computed using divisor $T$. ML/QML uses
$(\hat\sigma^{\mle})^2\hat\Omega_T^{-1}/T$.
CLAD uses $\hat\Omega_T^{-1}\hat f_u(0)^{-2}/4T$, with the
density estimator as in Section~\ref{sec_guid} and
$h_T=1.06T^{-1/5}$ times the median positive residual among
observations with positive fitted index. The last two columns
replace $\hat f_u(0)$ by its true value:
$1/\sqrt{2\pi}$ in Panel A and $1.7/\sqrt{2\pi}$ in Panel B.
All other estimates remain unchanged.\label{tab:short_run_CI}}
\end{table}

The Gaussian MLE gives coverage close to $95\%$ in Panel A
($93.3$--$95.7\%$), at the same time producing shorter intervals
than CLAD. The length comparison is reversed in Panel B, where
the CLAD intervals are about a quarter shorter than those based
on Gaussian QMLE. For CLAD, the coverage with estimated density
is $91.1$--$95.3\%$, generally improving when the density is known.
OLS can fare considerably worse: its coverage falls to $79.2\%$.
Gaussian QMLE performs reasonably well for this mixture,
although this does not remove the difficulty, discussed above,
of using its Gaussian critical values under misspecification.
Overall, the Gaussian MLE intervals perform well when the
likelihood is correctly specified; the feasible CLAD intervals
also give reasonably good coverage under both distributions,
with some undercoverage remaining at these sample sizes.

\subsection{Uniform confidence intervals}\label{sec_univofrmCI}

In the spirit of \citet{Hansen99REStat} and \citet{Mik07Ecta}, Theorems~\ref{thm:tobitmle} and \ref{thm:tobitLAD}, together with their stationary counterparts, suggest a way to conduct uniformly valid inference in the dynamic Tobit model. In particular, the limiting distributions in \eqref{eq_MLE_LUR} and \eqref{LAD_limits_LUR1} provide a basis for the construction of a joint confidence set for the intercept $\alpha_0$ and the sum of the autoregressive coefficients $\beta_0$. Both localizing parameters, $a_0$ and $c_0$, enter the limiting process $Y(\cdot)$ and therefore affect the asymptotic distributions of the ML and CLAD estimators.

In contrast to the linear autoregressive model, where persistence is governed solely by $\beta_0$, here both coefficients influence the degree of persistence. As a result, we consider a two-dimensional confidence ellipse for $(\alpha_0,\beta_0)$. This region can be constructed by forming a Wald-type statistic for the hypothesis $H_0:(\alpha_0,\beta_0)=(\alpha,\beta)$ and inverting the corresponding tests over candidate values of $(\alpha,\beta)$. Individual confidence sets for each coefficient can then be obtained by projection onto the coordinate axes.

For CLAD, however, the procedure is more complicated than in the linear autoregressive case. Its LUR limit depends on the innovation distribution through three quantities that need to be estimated: the density at zero $f_u(0)$, the variance $\sigma^2_0$, and the first absolute moment $\mathbb{E}|u_t|$. Let $\hat x_{t-1}$ denote the fitted latent index defined below and $\hat u_t=y_t-\hat x_{t-1}$, using the unrestricted CLAD estimates. For details on how to estimate $f_u(0)$, see Section~\ref{sec_guid} below.

Since $\mathbb{E}u_t=0$, we have $\mathbb{E}|u_t|=2\mathbb{E}[u_t]_+$. We therefore estimate this quantity using the positive parts of the residuals:
\begin{equation*}
\hat m_T=\frac{2\sum_{t=1}^T\indic\{\hat{x}_{t-1}>0\}[\hat{u}_t]_+}{\sum_{t=1}^T\indic\{\hat{x}_{t-1}>0\}},\qquad
\hat{x}_{t-1}  =\hat{\alpha}^{\lad}+\hat{\beta}^{\lad}y_{t-1}+\left(\hat{\vec{\phi}}^{\lad}\right)^{\trans}\Delta\vec y_{t-1},\qquad
  \hat{u}_t=y_t-\hat{x}_{t-1}.
\end{equation*}
At the true coefficients, $x_{t-1}>0$ implies $[y_t-x_{t-1}]_+=[u_t]_+$, even when $u_t$ itself is censored. Moreover, $u_t$ is independent of $x_{t-1}$. The most challenging part is estimating the variance $\sigma^2_0$. Additional assumptions are needed to ensure that the estimator works both under LUR and stationarity. Here we propose to estimate the variance under additional assumption of symmetry, so that $\mathbb{E}[u_t]_+^2=\mathbb{E}[u_t]_-^2$. Under this restriction, $\sigma_0^2=2\mathbb{E}([u_t]_+^2)$, and we may use
\begin{equation*}
\hat{\sigma}_T^2=\frac{2\sum_{t=1}^T\indic\{\hat{x}_{t-1}>0\}[\hat{u}_t]_+^2}{\sum_{t=1}^T\indic\{\hat{x}_{t-1}>0\}}.
\end{equation*}

With a consistent estimate $\hat f_u(0)$, define the Wald statistic:
$$
W_T(\alpha,\beta)
 =\left(\sqrt T(\hat\alpha_T^{\lad}-\alpha),T(\hat\beta_T^{\lad}-\beta)\right)\hat V_T^{-1}\left(\sqrt T(\hat\alpha_T^{\lad}-\alpha),T(\hat\beta_T^{\lad}-\beta)\right)^{\trans},
$$
where
$$\hat V_T
 =\frac{1}{4\hat f_u(0)^2}
   \left[\left\{\sum_{t=1}^T
   \indic\{\hat x_{t-1}>0\}z_{t-1,T}z_{t-1,T}^{\trans}
   \right\}^{-1}\right]_{1:2,1:2},\qquad z_{t-1,T}
 =\begin{bmatrix}T^{-1/2}\\T^{-1}y_{t-1}\\T^{-1/2}\Delta\vec y_{t-1}\end{bmatrix}.$$
Here $[\cdot]_{1:2,1:2}$ denotes the leading $2\times2$ block, taken after inversion of the full matrix to account for estimation of $\vec\phi_0$. The indicator, although asymptotically negligible under LUR, is required under stationarity.

In practice, critical values for $W_T(\alpha,\beta)$ can be obtained by a grid bootstrap, since its LUR limit is nonstandard. The bootstrap samples must be generated by drawing innovations from a distribution consistent with the quantities $f_u(0)$, $\sigma_0^2$ and $\mathbb{E}\smlabs{u_t}$ in the underlying data generating process, since these appear in the limiting distributions in Theorem~\ref{thm:tobitLAD}. To that end, consider the following three-parameter distribution:
\[
g(u;w,s_1,s_2)
 =\frac{w}{s_1}\varphi(u/s_1)
  +\frac{1-w}{s_2}\varphi(u/s_2),
\qquad 0<w<1,\quad 0<s_1\leq s_2,
\]
where $\varphi$ is the standard normal density. That is, a bootstrap innovation $u_t^{b}$ will be drawn from $N(0,s_1^2)$ with probability $w$ and from $N(0,s_2^2)$ otherwise; its mean and median are both zero. The parameters $(\hat w,\,\hat s_1,\,\hat s_2)$ are specified so as to match the (unrestricted) estimates of $f_u(0)$, $\sigma_0^2$ and $\mathbb{E}\smlabs{u_t}$ obtained from the data, as per\footnote{For positive estimates, a mixture of two centred normal distributions with different variances and positive weights can match all three quantities if and only if $\pi\hat f_u(0)\hat m_T>1$ and $\hat\sigma_T^2>(\pi/2)\hat m_T^2$. When both inequalities become equalities, a single normal distribution matches the estimates. In all other cases, this family cannot match the three estimates exactly and a different family is needed.}
\[
\begin{split}
\frac{\hat w/\hat s_1+(1-\hat w)/\hat s_2}{\sqrt{2\pi}}
 &=\hat f_u(0),\\
\hat w\hat s_1^2+(1-\hat w)\hat s_2^2
 &=\hat\sigma_T^2,\\
\sqrt{\frac2\pi}\{\hat w\hat s_1+(1-\hat w)\hat s_2\}
 &=\hat m_T.
\end{split}
\]
For each candididate pair $(\alpha,\beta)$, the remaining coefficients $\vec{\phi}$ are estimated with the null imposed; bootstrap samples are generated conditional on the initial values for $\{y_t\}$, by simulating the model at those estimates, with innovations drawn from $g$ above. For each bootstrap sample, the model is (unrestrictedly) re-estimated by CLAD, and the Wald statistic (for the null of $(\alpha,\beta)$) re-computed. The resulting collection of bootstrap Wald statistics yields an estimate of the quantiles of $W_T(\alpha,\beta)$, and thus suitable critical values for the test.


Of course, we do not claim that our results are sufficient to establish the \emph{uniform} validity of the procedure outlined above. Theorem~\ref{thm:tobitLAD} and the existing CLAD literature together establish the limits of $W_T(\alpha,\beta)$ under (drifting) local-to-unity and (fixed) stationary parametrisations. As discussed in \citet{ACG20JoE}, rigorously proving that the above procedure controls asymptotic size would also require us to consider the asymptotics of $W_T$ under certain other drifting parameter sequences, such as those under which $\beta_T \rightarrow 1$ at a slower rate than $T^{-1}$. This is deferred to future research.

\subsection{Practitioner's guide}\label{sec_guid}





We provide the following guidance for empirical work, consistent with our implementation of these estimators in Section~\ref{sec_empirical} below.
\begin{enumerate}
\item Choose $\bar{k}$ large enough and run CLAD model selection as in Section~\ref{sec_modelselection}, using $cv(T)=\sqrt{\log T}$. If there is an a priori reason to believe in normality, then instead use MLE. The relevant consistency results assume a fixed $\bar{k}\geq k_0$; CLAD additionally requires the innovations to have zero median.

\item Estimate the coefficients with $k=k^*$ and calculate standard
errors for the short-run coefficients. The covariance matrix $\hat{\Omega}_T$ is given in \eqref{eq_Omega_hat}. For CLAD, let
\[
  \hat{x}_{t-1}
  =\hat{\alpha}+\hat{\beta}y_{t-1}
   +\hat{\vec{\phi}}^{\trans}\Delta\vec y_{t-1},
  \qquad
  \hat{u}_t=y_t-\hat{x}_{t-1}.
\]
Following \citet[(5.2)]{Pow84JoE}, the innovation density at zero
can be estimated by
\[
  \hat{f}_u(0)
  =
  \frac{
    \sum_{t=1}^{T}
    \indic\{\hat{x}_{t-1}>0\}
    \indic\{0\leq\hat{u}_t\leq h_T\}
  }{
    h_T\sum_{t=1}^{T}\indic\{\hat{x}_{t-1}>0\}
  }.
\]
Here $h_T$ is a bandwidth. The restriction to nonnegative residuals
ensures that censored observations do not enter the numerator.
One possible choice, following \citet[(5.5)]{Pow84JoE}, is
\[
  h_T=dT^{-1/5}
  \med\{\hat{u}_t:\hat{u}_t>0,\ \hat{x}_{t-1}>0\},
  \qquad d>0.
\]
Thus, under the LUR asymptotics, the estimated covariance matrix of
$\hat{\vec{\phi}}$ is
\[
  \widehat{\operatorname{Var}}(\hat{\vec{\phi}})
  =
  \begin{cases}
    \displaystyle
    \frac{1}{4T\hat{f}_u(0)^2}\hat{\Omega}_T^{-1},
      & \text{for CLAD},\\[6pt]
    \displaystyle
    \frac{(\hat{\sigma}_T^{\mle})^2}{T}\hat{\Omega}_T^{-1},
      & \text{for Gaussian MLE},
  \end{cases}
\]
and an approximate $95\%$ confidence interval for $\phi_i$ is $\hat{\phi}_i\pm1.96\,\widehat{\mathrm{se}}(\hat{\phi}_i)$.

\item For $(\alpha_0,\beta_0)$, form the joint confidence region
proposed in Section~\ref{sec_univofrmCI}. This allows one to assess,
for instance, whether the data are compatible with the restriction
$(\alpha_0,\beta_0)=(0,1)$. Separate confidence intervals follow by projection onto the two axes.


\item As a default, use the simple one-step-ahead forecast
\[
\left[\hat\alpha+\hat\beta y_T+
\hat{\vec\phi}^{\trans}\Delta\vec y_T\right]_+.
\]
Under the zero-median assumption, this estimates the one-step conditional median. For longer horizons, iterate the same formula, replacing unobserved lagged values by their forecasts and applying the positive part at each step. For conditional mean forecasts, a more computationally demanding alternative is to average the censored outcomes over the fitted innovation distribution.
\end{enumerate}

\section{Empirical illustration}\label{sec_empirical}
This section illustrates the proposed procedure. Understanding and forecasting short-term interest rates is important for monetary policy analysis and financial decision-making. Their high persistence and recurrent episodes near the effective lower bound motivate methods that accommodate censoring while allowing reliable inference on short-run dynamics. Treasury bill rates in the United States have been near the zero lower bound (ZLB) during three historical periods: during the $1930$s and $1940$s following the Great Depression, during and after the $2007$ Global Financial Crisis, and more recently during the COVID-$19$ pandemic. We use quarterly data on $3$-month Treasury bill rates. The data set is based on \citet{ramey2018government}, commences in the first quarter of $1920$ is extended using FRED data to the first quarter of $2025$, so that $T=421$.

We transform the $3$-month Treasury bill rate $r_t$ as $y_t=[r_t-0.5]_+$, treating rates at or below $50$ basis points as effectively at the lower bound. This threshold closely matches the historical dating of \citet{ramey2018government}, whose baseline lower-bound period begins in the second quarter of $1932$, when the $3$-month Treasury bill rate fell to approximately $0.5\%$. It also provides a conservative definition of the near-zero regime: since conventional Federal Reserve policy adjustments are typically made in $25$-basis-point increments, a $50$-basis-point threshold allows for small policy-rate movements without treating them as a return to normal-rate conditions. A $25$-basis-point threshold more closely reflects the modern Federal Reserve target range at the ZLB, but is less well suited to identifying the earlier historical lower-bound episode.



\begin{figure}[t]
\begin{subfigure}{.49\textwidth}
\centering \includegraphics[width=0.8\linewidth]{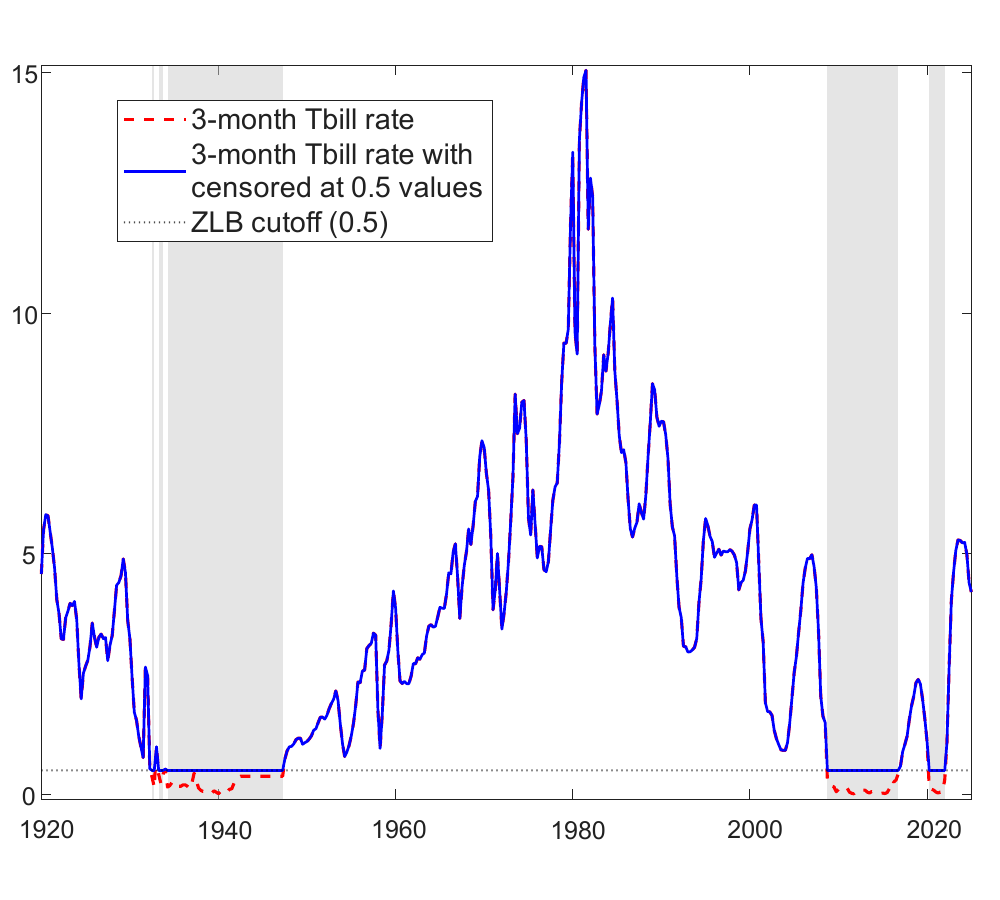}
\caption{3-month Tbills and ZLB.\\ Gray area represents ZLB periods.}\label{fig:3monthdata}
\end{subfigure}
\begin{subfigure}{.49\textwidth}
\centering \includegraphics[width=0.87\linewidth]{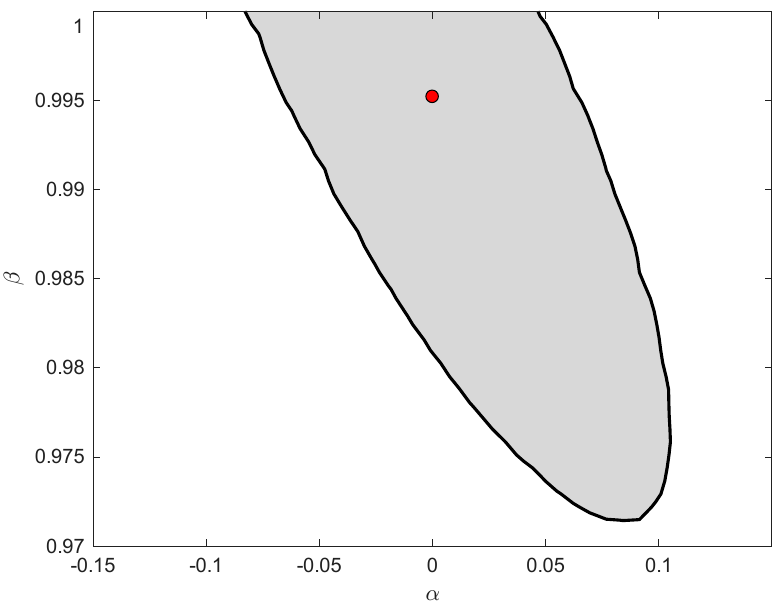}
\caption{$95\%$ joint grid-bootstrap confidence set (Wald inversion) for $\alpha$ and $\beta$.}\label{fig:CIab95}
\end{subfigure}
\caption{3-month Tbills $r_t$ and a joint confidence set for the intercept and persistence in $y_t=[r_t-0.5]_+$.}
\label{fig:3month}
\end{figure}

Figure~\ref{fig:3monthdata} displays the interest rate data. We estimate the transformed series using both Gaussian MLE and CLAD and perform model selection by sequential $t$-testing starting from $\bar{k}=20$, corresponding to up to $5$ years of quarterly lags. For CLAD, the density at zero is estimated as per the proposal in Section~\ref{sec_guid} above. We report the $t$-statistics in Table~\ref{tab:ttest_stat}. Using the conventional $5\%$ standard-normal critical value $1.96$, both procedures select $k=14$. As we discuss in Section \ref{sec_modelselection}, fixed critical values yield conventional hypothesis tests but do not deliver consistent lag-order selection; for the latter, the critical value must diverge with $T$. Using critical value $\sqrt{\log T}\approx2.46$, which satisfies the conditions for consistent selection, CLAD continues to select $k=14$, whereas Gaussian MLE selects $k=8$.
The CLAD choice of $k=14$ is also supported by a Wald test of $H_0:\phi_8=\cdots=\phi_{13}=0$ in Tobit($14$), which yields a $p$-value well below $0.01$. Thus, the CLAD choice of $k=14$ is robust to the critical-value specification, while the MLE evidence for the additional lags is more marginal. We therefore use $k=14$ as our baseline specification and report the corresponding estimates in Table~\ref{tab:estimates_tb3}.

The CLAD estimates have substantially smaller standard errors than their Gaussian MLE counterparts. Since correctly specified Gaussian MLE would be more efficient than CLAD under Gaussian innovations, this reversal is consistent with departures from Gaussianity and provides further motivation for preferring our CLAD estiamtes. The associated $95\%$ joint confidence set for $\alpha$ and $\beta$, obtained by inversion of the CLAD grid-bootstrap Wald test, is shown in Figure~\ref{fig:CIab95}. We restrict the reported joint confidence set to $\beta\leq1$, focusing on non-explosive long-run interest-rate dynamics while retaining an exact unit root as a boundary case.


\begin{table}[t]
\centering
\scriptsize
\setlength{\tabcolsep}{2.5pt}
\renewcommand{\arraystretch}{1.05}
\resizebox{\textwidth}{!}{%
\begin{tabular}{c*{19}{r}}
\hline
$k$
& 2 & 3 & 4 & 5 & 6 & 7 & 8 & 9 & 10
& 11 & 12 & 13 & 14 & 15 & 16 & 17 & 18 & 19 & 20 \\
\hline
\hline
MLE
& 5.0 & 4.2 & 4.5 & 1.6 & 2.9 & 3.4 & 3.6 & 1.4 & 0.1
& 0.0 & 2.3 & 0.5 & 2.0 & 0.5 & 1.9 & 0.3 & 0.8 & 1.1 & 1.3\\
CLAD
& 17.6 & 1.6 & 3.3 & 1.7 & 0.7 & 5.7 & 4.6 & 2.5 & 0.1
& 0.9 & 1.4 & 0.1 & 3.4 & 1.0 & 0.7 & 0.2 & 0.0 & 0.1 & 1.5 \\
\hline
\hline
\end{tabular}%
}
\caption{Absolute values of $t$-statistics for $\phi_{k-1}$.}
\label{tab:ttest_stat}
\end{table}

\begin{table}[t]
\centering
\scriptsize
\setlength{\tabcolsep}{3pt}
\renewcommand{\arraystretch}{0.95}
\resizebox{\textwidth}{!}{%
\begin{tabular}{c*{15}{r}}
\hline
& $\alpha$ & $\beta$ & $\phi_1$ & $\phi_2$ & $\phi_3$ & $\phi_4$ & $\phi_5$ & $\phi_6$ & $\phi_7$ & $\phi_8$ & $\phi_9$ & $\phi_{10}$ & $\phi_{11}$ & $\phi_{12}$ & $\phi_{13}$  \\
\hline
\hline
MLE
& $-0.2$ & $1.0$ & $0.4$ & $-0.4$ & $0.4$ & $-0.2$
& $0.1$ & $-0.1$ & $-0.3$ & $0.1$ & $-0.1$ & $0.1$
& $-0.2$ & $0.1$ & $-0.1$  \\
& \multicolumn{1}{c}{--} & \multicolumn{1}{c}{--}
& $(0.06)$ & $(0.06)$ & $(0.07)$ & $(0.07)$
& $(0.07)$ & $(0.07)$ & $(0.07)$ & $(0.07)$
& $(0.07)$ & $(0.07)$ & $(0.07)$ & $(0.06)$
& $(0.06)$  \\
\hline
CLAD & $0.0$ & $1.0$ & $0.5$ & $-0.2$ & $0.3$
& $-0.1$ & $0.1$ & $-0.1$ & $-0.3$ & $0.1$
& $-0.1$ & $0.1$ & $-0.1$ & $0.1$ & $-0.1$ \\
& \multicolumn{1}{c}{--} & \multicolumn{1}{c}{--}
& $(0.03)$ & $(0.03)$ & $(0.04)$ & $(0.04)$
& $(0.04)$ & $(0.04)$ & $(0.04)$ & $(0.04)$
& $(0.04)$ & $(0.04)$ & $(0.04)$ & $(0.03)$ & $(0.03)$ \\
\hline
\hline
\end{tabular}%
}
\caption{MLE and CLAD estimates with standard errors for $\hat{\phi}_i$ coefficients.}
\label{tab:estimates_tb3}
\end{table}

Finally, we evaluate the out-of-sample forecasting performance of the dynamic Tobit model for the transformed series $y_t=[r_t-0.5]_+$, which is the object directly modeled by the specification. Thus, movements in $r_t$ within the effective-lower-bound region are deliberately treated as observationally equivalent (to zero) and are not distinguished in the forecast evaluation. We consider one-step-ahead forecasts over the period 2000Q1--2025Q1 using an expanding estimation window. As a benchmark, we use the random-walk forecast $\hat y_{t+1|t}=y_t$. The alternative is the CLAD-estimated Tobit($k$) forecast
$$\hat y_{t+1|t}=
\left[\hat\alpha^{\lad}_t+\hat\beta^{\lad}_t y_t+\left(\hat{\vec\phi}{}^{\lad}_t\right)^{\trans}\Delta\vec y_t\right]_+.$$
Applying the sequential CLAD lag-selection procedure to the pre-2000 sample selects $k=14$, so we fix $k=14$ throughout the out-of-sample period. Table~\ref{tab:tbill_forecast} reports the mean absolute error (MAE) and the root mean squared error (RMSE) for the two forecasting methods. The dynamic Tobit model substantially outperforms the random-walk benchmark: its MAE is $22.9\%$ lower and its RMSE is $28.3\%$ lower. Given that the CLAD estimates of the intercept and persistence parameters are very close to the random-walk values $(0,1)$, the substantial forecasting gains underscore the importance of incorporating the short-run dynamics.

\begin{table}[t]
\centering
\begin{tabular}{lcc}
\hline
& MAE & RMSE \\
\hline\hline
CLAD Tobit($14$) & 0.1854 & 0.3037 \\
Random walk      & 0.2406 & 0.4237 \\
\hline
Relative improvement (\%) & 22.9 & 28.3 \\
\hline\hline
\end{tabular}
\caption{One-step-ahead forecasting performance, 2000Q1--2025Q1}
\label{tab:tbill_forecast}
\end{table}

\section{Conclusion}\label{sec_conclusion}

This paper develops asymptotic theory for Gaussian maximum likelihood
(ML) and censored least absolute deviations (CLAD) estimation in
dynamic Tobit models under local unit root (LUR) asymptotics. While
earlier work had established consistency (and asymptotic normality)
of these estimators in the stationary setting, we show that both ML
and CLAD remain consistent in the LUR regime and derive their
asymptotic distributions. A key finding is that the short-run parameter
estimators are asymptotically normal under both methods, as in the
stationary setting. This permits standard normal inference on the
short-run dynamics in either regime, despite the non-standard limits
for the intercept and persistence parameter under LUR. These results
also yield consistent lag-order selection via sequential
$t$-testing with appropriately diverging critical values.

These results contrast with the behaviour of ordinary least squares
(OLS), which, although consistent under LUR, yields $t$-statistics
for the short-run parameters that have an asymptotic bias and
non-standard distribution. Standard normal critical values are
therefore not generally valid for OLS-based inference.

Overall, our results support CLAD as the default method for estimation
and inference in dynamic censored models, under its maintained
assumptions. Although Gaussian QMLE remains consistent under non-Gaussian
innovations in the LUR regime, its limiting distributions generally
differ from those under correct specification. Under stationarity,
such misspecification can also lead to inconsistency. CLAD permits
standard inference on the short-run coefficients for a wider range
of innovation distributions, while Gaussian MLE remains a useful
benchmark under correct specification. Future work could establish
uniform validity of inference over parameter sequences spanning the
stationary and LUR regimes, and formally extend the results to
accommodate time-varying volatility, covariates or alternative forms
of censoring commonly encountered in economic and financial data.




\appendix
\section{Auxiliary lemmas}\label{sec_notation}

\begin{notation*}
$e_{m,i}$ denotes the $i$th column of an $m\times m$ identity matrix;
when $m$ is clear from the context, we write this simply as $e_{i}$.
For $S\subset\reals^{m}$, let $\ucc(S)$ denote the set of functions
that are uniformly bounded on compact subsets of $S$, equipped with
the topology of uniform convergence on compacta. (When $S$ is compact,
this coincides with the uniform topology.)
Let $\filt_{t}\defeq\sigma(\{y_{s}\}_{s\leq0},\{u_{s}\}_{s\leq t})$
denote the filtration generated by the initial conditions and the
innovations $\{u_{t}\}$.
\end{notation*}

For ease of reference, we start by presenting some technical results. We will use the following two properties. First, let $y_{t}^{-}\defeq[\alpha_{T,0}+\beta_{T,0}y_{t-1}+\vec{\phi}_{0}^{\trans}\Delta\vec y_{t-1}+u_{t}]_{-}$ and $v_{t}\defeq u_{t}-y_{t}^{-}$, where $[x]_-\defeq\min\{x,0\}$. Then the evolution of $\{y_{t}\}$ may also be described by
\begin{equation}
y_{t}=\alpha_{T,0}+\beta_{T,0}y_{t-1}+\vec{\phi}_{0}^{\trans}\Delta\vec y_{t-1}+v_{t} = x_{t-1} +v_t,\label{eq:v}
\end{equation}
where we have also defined
\begin{equation}\label{eq:xtm1}
x_{t-1} \defeq \alpha_{T,0}+\beta_{T,0}y_{t-1}+\vec{\phi}_{0}^{\trans}\Delta\vec y_{t-1}.
\end{equation}
Second, as noted in the
proof of \citet[Lemma~B.2]{BD22}, since the polynomial $\phi(z)\defeq1-\sum_{i=1}^{k-1}\phi_{i,0}z^{i}$
has all its roots outside the unit circle under \ref{ass:INIT}--\ref{ass:JSR},
it has a well-defined inverse $\phi^{-1}(z)\defeq\sum_{i=0}^{\infty}\gamma_{i}z^{i}$
for all $\smlabs z\leq1$.

We also define the scaled regressor process:
\begin{equation}\label{eq:ztm1}
z_{t-1,T} =\begin{bmatrix}T^{-1/2}\\
T^{-1}y_{t-1}\\
T^{-1/2}\Delta\vec y_{t-1}.
\end{bmatrix}
\end{equation}

\begin{lem}
\label{lem:bd2022}Suppose \ref{ass:INIT}--\ref{ass:JSR} hold.
Then
\begin{enumerate}
\item \label{enu:mombnd}there exists a $C<\infty$ such that $\max_{-k+2\leq t\leq T}\smlnorm{\Delta y_{t}}_{2+\delta_{u}}<C$;
\item \label{enu:deltayt}$T^{-1}\sum_{t=1}^{T}(\Delta\vec y_{t-1})(\Delta\vec y_{t-1})^{\trans}\inprob\Omega$,
where $\Omega=[\Omega_{ij}]_{i,j=1}^{k-1}$ is positive definite,
with $\Omega_{ij}=\sigma_{0}^{2}\sum_{n=0}^{\infty}\gamma_{n}\gamma_{n+\smlabs{i-j}}$;
\item \label{enu:leveldelta}$\sum_{t=1}^{T}y_{t-1}\Delta\vec y_{t-1}=O_{p}(T)$.
\end{enumerate}
\end{lem}

\begin{lem}
\label{lem:integ}Suppose \ref{ass:INIT}--\ref{ass:JSR} hold, $M>0$,
and let $g(x)$ be a bounded function such that $\smlabs{g(x)}\goesto0$
as $x\goesto+\infty$. Let $\{{\cal Y}_{t,T}\}_{t=1}^{T}$ be such
that $T^{-1/2}{\cal Y}_{\smlfloor{\tau T},T}\indist Y(\tau)$, where $Y(\cdot)$ is defined in (\ref{eq_LUR_limit}). Then
\begin{enumerate}
\item \label{enu:loctime}$\int_{0}^{1}\indic\{Y(\tau)\leq M\}\diff\tau\inprob0$
as $M\goesto0$;
\item \label{enu:indicbnd} $T^{-1}\sum_{t=1}^{T}\indic\{T^{-1/2}{\cal Y}_{t-1,T}<M\}\inprob0$
as $T\goesto\infty$ and then $M\goesto0$; and
\item \label{enu:integ}$\expect\smlabs{T^{-1}\sum_{t=1}^{T}g({\cal Y}_{t-1,T})}=o(1)$.
\end{enumerate}
In particular, the preceding holds with ${\cal Y}_{t-1,T}=\alpha_{T,0}+\beta_{T,0}y_{t-1}+\vec{\phi}_{0}^{\trans}\Delta\vec y_{t-1}$.
\end{lem}


\begin{lem}
\label{lem:moments}Suppose \ref{ass:INIT}--\ref{ass:JSR} hold.
Then
\begin{enumerate}
\item \label{enu:ZZ}$\sum_{t=1}^{T}z_{t-1,T}z_{t-1,T}^{\trans}\indist Q_{ZZ}$
jointly with $\frac1{\sqrt{T}}\sum_{t=1}^{\lfloor\tau T\rfloor}u_t\indist \sigma_0 W(\tau)$, where
\[
Q_{ZZ}\defeq\begin{bmatrix}1 & \int_{0}^{1}Y(\tau)\diff\tau & 0\\
\int_{0}^{1}Y(\tau)\diff\tau & \int_{0}^{1}Y^{2}(\tau)\diff\tau & 0\\
0 & 0 & \Omega
\end{bmatrix}
\]
is a.s.\ positive definite;
\item \label{enu:max} $\max_{0\leq t\leq T}\smlnorm{z_{t,T}}=O_{p}(T^{-\delta_{u}/2(2+\delta_{u})})=o_{p}(1)$;
\item \label{enu:zeros}$T^{-1}\sum_{t=1}^{T}\indic\{y_{t}=0\}=o_{p}(1)$;
and
\item \label{enu:xprodpos}$\sum_{t=1}^{T}z_{t-1,T}u_{t}\indic\{y_{t}>0\}=o_{p}(T^{1/2})$
\end{enumerate}
Suppose further that \ref{ass:MOM} holds with $\delta_{u}>2$. Then
\begin{enumerate}[resume]
\item \label{enu:fourthz}$\sum_{t=1}^{T}\smlnorm{z_{t-1,T}}^{4}=O_{p}(T^{-1})$;
\item \label{enu:vtzero}$\sum_{t=1}^{T}\indic\{y_{t}=0\}\smlabs{v_{t}}^{m}=o_{p}(T)$
for each $m\in[1,4]$.
\end{enumerate}
\end{lem}

\begin{lem}
\label{lem:mills} Let $\Phi(x), \varphi(x)$ denote the standard normal cdf and pdf functions and let $\lambda(x)\defeq\varphi(x)/[1-\Phi(x)]$ be the inverse Mills ratio. There exists a $C<\infty$ such that $\lambda(x)\leq C(1+[x]_{+})$
and $\smlabs{\lambda^{\prime}(x)}\leq C$ for all $x\in\reals$.
\end{lem}
\begin{proof}[Proof of Lemma~\ref{lem:bd2022}]
\ref{enu:mombnd}, \ref{enu:deltayt} and \ref{enu:leveldelta} correspond
to Lemma~B.2 and parts (iii) and (iv) of Lemma~B.4 in \citet{BD22}.

Regarding the positive definiteness of $\Omega$ asserted in part~\ref{enu:deltayt},
by defining the stationary autoregressive process
\[
w_{t}=\sum_{i=1}^{k-1}\phi_{i,0}w_{t-i}+u_{t}
\]
and setting $\vec w_{t}\defeq(w_{t},\ldots,w_{t-k+2})^{\trans}\in\reals^{k-1}$,
we obtain that $\Omega=\expect\vec w_{t}\vec w_{t}^{\trans}$ is necessarily
positive semi-definite. Further, by the independence of $u_{t}$ from
$\{w_{\tau}\}_{\tau< t}$, we have for $\vec a=(a_{1},\ldots,a_{k-1})^{\trans}$
that $\mathbb{V}(\vec a^{\trans}\vec w_{t})\geq\mathbb{V}(a_{1}u_{t})=a_{1}^{2}\sigma_{0}^{2}$,
and so is nonzero so long as $a_{1}\neq0$. If $a_{1}=0$, then we similarly get $\mathbb{V}(\vec a^{\trans}\vec w_{t})\geq\mathbb{V}(a_{2}u_{t-1})=a_{2}^{2}\sigma_{0}^{2}$, which is nonzero as long as $a_{2}\neq0$. Repeating the same argument, it follows that there does not exist any $\vec a\in\reals^{k-1}\backslash\{0\}$
such that $\vec a^{\trans}\Omega\vec a=\mathbb{V}(\vec a^{\trans}\vec w_{t})=0$.
\end{proof}
\begin{proof}[Proof of Lemma~\ref{lem:integ}]
\textbf{\ref{enu:loctime}.} Let $c_{\phi}\defeq\phi(1)^{-1}c_0$ and $Y^{\ast}(\tau)\defeq\phi(1)\e^{-c_{\phi}\tau}Y(\tau)$,
where $\phi(1)>0$ as a consequence of \ref{ass:JSR}, as noted in \citet[Remark~3.1]{BD22}. Then
\begin{equation*}\begin{split}
Y^{\ast}(\tau)&=K(\tau)+\sup_{\tau^{\prime}\leq\tau}[-K(\tau^{\prime})]_{+},\\
K(\tau)&=\phi(1)b_0+a_0\int\limits_0^{\tau}e^{-c_{\phi}r}dr  +\sigma_0\int\limits_0^{\tau}e^{-c_{\phi}r} dW(r).
\end{split}\end{equation*}
Now consider
\begin{align*}
\int_{0}^{1}\indic\{Y(\tau)\leq M\}\diff\tau & =\int_{0}^{1}\indic\{Y^{\ast}(\tau)\leq M\phi(1)\e^{-c_{\phi}\tau}\}\diff\tau\\
 & \leq\int_{0}^{1}\indic\{Y^{\ast}(\tau)<2M\phi(1)(1+\e^{-c_{\phi}})\}\diff\tau\\
 & =\int_{0}^{1}\indic\{0\leq Y^{\ast}(\tau)<2M\phi(1)(1+\e^{-c_{\phi}})\}\diff\tau
\end{align*}
since $Y^{\ast}(\tau)\geq0$ for all $\tau$. It remains to show that
the r.h.s.\ vanishes as $M\goesto0$. To that end, observe that $K(\cdot)$
is a continuous semimartingale, and that the (weakly) increasing process
$\sup_{\tau^{\prime}\leq\tau}[-K(\tau^{\prime})]_{+}$
is continuous and has finite variation. Hence $Y^{\ast}(\tau)$ is
also a continuous semimartingale, with quadratic variation
\[
[Y^{\ast}(\tau)]=[K(\tau)]=\sigma_{0}^{2}\int_{0}^{\tau}\e^{-2c_{\phi}s}\diff(s)=\frac{\sigma_{0}^{2}(1-\e^{-2c_{\phi}\tau})}{2c_{\phi}}\defeq\mu^{\ast}(\tau)
\]
and $\mu^{\ast}(\tau)=\sigma_{0}^{2}\tau$ when $c_0=0$. It follows by Corollary~VI.1.9 of \citet{Revuz_Yor}
that, almost surely
\[
\int_{0}^{1}\indic\{0\leq Y^{\ast}(\tau)<M^{\prime}\}\diff\mu^{\ast}(\tau)\goesto0
\]
as $M^{\prime}\goesto0$. Since $\mu^{\ast}$ and Lebesgue measure
on $[0,1]$ are mutually absolutely continuous, it follows that
\[
\int_{0}^{1}\indic\{0\leq Y^{\ast}(\tau)<M^{\prime}\}\diff\tau\goesto0
\]
as $M^{\prime}\goesto0$, as required.

\textbf{ \ref{enu:indicbnd}.}
Letting $h_{M}(x)$ be any smooth function such that $\indic\{x<M\}\leq h_{M}(x)\leq\indic\{x<2M\}$,
we have by the continuous mapping theorem (CMT) and the result of
part~\ref{enu:loctime} that
\begin{multline*}
T^{-1}\sum_{t=1}^{T}\indic\{T^{-1/2}{\cal Y}_{t-1,T}<M\}\leq T^{-1}\sum_{t=1}^{T}h_{M}(T^{-1/2}{\cal Y}_{t-1,T})\\
\indist\int_{0}^{1}h_{M}[Y(\tau)]\diff\tau\leq\int_{0}^{1}\indic\{Y(\tau)<2M\}\diff\tau\inprob0
\end{multline*}
as $T\goesto\infty$ and then $M\goesto0$.

\textbf{ \ref{enu:integ}.}
Let $\epsilon>0$ and $M>0$: then for all $T$ sufficiently large,
$\smlabs{g(x)}<\epsilon$ for all $x\geq T^{1/2}M$. Since there exists
a $C<\infty$ such that $\smlabs{g(x)}\leq C$ for all $x\in\reals$,
it follows that for such $T$,
\begin{align}
\abs{T^{-1}\sum_{t=1}^{T}g({\cal Y}_{t-1,T})} & \leq CT^{-1}\sum_{t=1}^{T}\indic\{T^{-1/2}{\cal Y}_{t-1,T}<M\}+\epsilon.\label{eq:integbnd}
\end{align}
Deduce, by the result of part~\ref{enu:indicbnd}, that the r.h.s.\ of
(\ref{eq:integbnd}) is bounded by $2\epsilon$ w.p.a.1\ as $T\goesto\infty$
and then $M\goesto0$; since $\epsilon$ was arbitrary, it follows
that $T^{-1}\sum_{t=1}^{T}g({\cal Y}_{t-1,T})=o_{p}(1)$. Since $g$
is bounded, this holds also in $L^{1}$.

Regarding the final claim, we need only to note that in this case
\[
T^{-1/2}{\cal Y}_{\smlfloor{\tau T},T}=T^{-1/2}\alpha_{T,0}+T^{-1/2}\beta_{T,0}y_{\smlfloor{\tau T}}+T^{-1/2}\vec{\phi}_{0}^{\trans}\Delta\vec y_{\smlfloor{\tau T}}\indist Y(\tau)
\]
by Lemma~\ref{lem:bd2022}\ref{enu:mombnd} and Theorem~3.2 in \citet{BD22}.
\end{proof}
\begin{proof}[Proof of Lemma~\ref{lem:moments}]
\textbf{\ref{enu:ZZ}.} We have
\begin{align*}
\sum_{t=1}^{T}z_{t-1,T}z_{t-1,T}^{\trans} & =\begin{bmatrix}1 & T^{-3/2}\sum_{t=1}^{T}y_{t-1} & T^{-1}\sum_{t=1}^{T}(\Delta\vec y_{t-1})^{\trans}\\
T^{-3/2}\sum_{t=1}^{T}y_{t-1} & T^{-2}\sum_{t=1}^{T}y_{t-1}^{2} & T^{-3/2}\sum_{t=1}^{T}y_{t-1}(\Delta\vec y_{t-1})^{\trans}\\
T^{-1}\sum_{t=1}^{T}\Delta\vec y_{t-1} & T^{-3/2}\sum_{t=1}^{T}(\Delta\vec y_{t-1})y_{t-1} & T^{-1}\sum_{t=1}^{T}(\Delta\vec y_{t-1})(\Delta\vec y_{t-1})^{\trans}
\end{bmatrix}\\
 & \indist\begin{bmatrix}1 & \int_{0}^{1}Y(\tau)\diff\tau & 0\\
\int_{0}^{1}Y(\tau)\diff\tau & \int_{0}^{1}Y^{2}(\tau)\diff\tau & 0\\
0 & 0 & \Omega
\end{bmatrix}
\end{align*}
by parts \ref{enu:deltayt} and \ref{enu:leveldelta} of Lemma~\ref{lem:bd2022},
Theorem~3.2 in \citet{BD22} and the CMT, noting in particular that
$\sum_{t=1}^{T}\Delta\vec y_{t-1}=\vec y_{T-1}-\vec y_{0}=O_{p}(T^{1/2})$.
The a.s.\ positive definiteness of the final matrix follows since
$\Omega$ is positive definite, by Lemma~\ref{lem:bd2022}\ref{enu:deltayt}, and $Y(\cdot)$ is non-constant, so that $\int_{0}^{1}Y^{2}(\tau)\diff\tau\neq\left(\int_{0}^{1}Y(\tau)\diff\tau\right)^2$.

\textbf{ \ref{enu:max}.} We have
\begin{align*}
\max_{1\leq t\leq T}\smlnorm{z_{t,T}}^{2} & \leq T^{-1}+T^{-1}\max_{1\leq t\leq T}(T^{-1/2}y_{t-1})^{2}+T^{-1}\max_{1\leq t\leq T}\smlnorm{\Delta\vec y_{t}}^{2}\\
 & =O_{p}(T^{-1})+O_{p}(T^{-1+2/(2+\delta_{u})})=O_{p}(T^{-\delta_{u}/(2+\delta_{u})})
\end{align*}
by Theorem~3.2 in \citet{BD22}, and Lemma~\ref{lem:bd2022}\ref{enu:mombnd}.

\textbf{\ref{enu:zeros}.} Let $\delta>0$, and $h_{\delta}(x)$ be
any smooth function such that $\indic\{0\leq x\leq\delta\}\leq h_{\delta}(x)\leq\indic\{0\leq x\leq2\delta\}$.
Then by the CMT and Lemma~\ref{lem:integ}\ref{enu:loctime},
\[
T^{-1}\sum_{t=1}^{T}\indic\{y_{t}=0\}\leq T^{-1}\sum_{t=1}^{T}h_{\delta}(T^{-1/2}y_{t})\indist\int_{0}^{1}h_{\delta}[Y(\tau)]\diff\tau\leq\int_{0}^{1}\indic\{0\leq Y(\tau)\leq2\delta\}\diff\tau\inprob0
\]
as $T\goesto\infty$ and then $\delta\goesto0$.

\textbf{\ref{enu:xprodpos}.} Decompose
\[
\sum_{t=1}^{T}z_{t-1,T}u_{t}\indic\{y_{t}>0\}=\sum_{t=1}^{T}z_{t-1,T}u_{t}-\sum_{t=1}^{T}z_{t-1,T}u_{t}\indic\{y_{t}=0\}.
\]
For the second r.h.s.\ term, two applications of the CS inequality
yield
\begin{align*}
\sum_{t=1}^{T}\smlnorm{z_{t-1,T}}\smlabs{u_{t}}\indic\{y_{t}=0\} & \leq\left(\sum_{t=1}^{T}\smlnorm{z_{t-1,T}}^{2}\right)^{1/2}\left(\sum_{t=1}^{T}\smlabs{u_{t}}^{4}\right)^{1/4}\left(\sum_{t=1}^{T}\indic\{y_{t}=0\}\right)^{1/4}\\
 & =O_{p}(1)\cdot O_{p}(T^{1/4})\cdot o_{p}(T^{1/4})=o_{p}(T^{1/2})
\end{align*}
by the results of parts \ref{enu:ZZ} and \ref{enu:zeros}. Next,
note
\[
\sum_{t=1}^{T}z_{t-1,T}u_{t}=\sum_{t=1}^{T}\begin{bmatrix}T^{-1/2}\\
T^{-1}y_{t-1}\\
T^{-1/2}\Delta\vec y_{t-1}
\end{bmatrix}u_{t}.
\]
By Theorem~3.2 in \citet{BD22} and Theorem~2.1 in \citet{ito_convergence},
the second component converges weakly to $\sigma\int_{0}^{1}Y(\tau)\diff W(\tau)$,
and so is $O_{p}(1)$. The remaining components form a martingale
with variance matrix
\[
\sigma_{0}^{2}\begin{bmatrix}1 & 0\\
0 & \expect(\Delta\vec y_{t-1})(\Delta\vec y_{t-1}^{\trans})
\end{bmatrix}
\]
which is uniformly bounded by Lemma~\ref{lem:bd2022}\ref{enu:mombnd};
hence these components are also $O_{p}(1)$.

\textbf{\ref{enu:fourthz}.} Note that there exists a $C<\infty$
(depending only on $n$) such that for all $x=(x_{1},\ldots,x_{n})^{\trans}\in\reals^{n}$,
$\smlnorm x^{4}=(\sum_{i=1}^{n}x_{i}^{2})^{2}\leq C\sum_{i=1}^{n}x_{i}^{4}$.
Therefore, in particular
\begin{align*}
\sum_{t=1}^{T}\smlnorm{z_{t-1,T}}^{4} & \leq C\sum_{t=1}^{T}\left[T^{-2}+T^{-4}\smlabs{y_{t-1}}^{4}+T^{-2}\smlnorm{\Delta\vec y_{t-1}}^{4}\right]\\
 & =CT^{-1}\left[1+T^{-1}\sum_{t=1}^{T}\abs{\frac{y_{t-1}}{T^{1/2}}}^{4}+T^{-1}\sum_{t=1}^{T}\smlnorm{\Delta\vec y_{t-1}}^{4}\right]\\
 & =O_{p}(T^{-1})
\end{align*}
since $\expect\smlnorm{\Delta\vec y_{t-1}}^{4}$ is uniformly bounded
by Lemma~\ref{lem:bd2022}\ref{enu:mombnd}, and
\[
T^{-1}\sum_{t=1}^{T}\abs{\frac{y_{t-1}}{T^{1/2}}}^{4}\indist\int_{0}^{1}Y^{4}(\tau)\diff\tau
\]
by Theorem~3.2 in \citet{BD22} and the CMT.

\textbf{\ref{enu:vtzero}.} Note that we always have
\[
v_{t}=u_{t}-y_{t}^{-}\geq u_{t}
\]
since $-y_{t}^{-}\geq0$, while if $y_{t}=0$, then by (\ref{eq:v})
\[
0=y_{t}=\alpha_{0,T}+\beta_{0,T}y_{t-1}+\vec{\phi}_{0}^{\trans}\Delta\vec y_{t-1}+v_{t}\geq\alpha_{0,T}+\vec{\phi}_{0}^{\trans}\Delta\vec y_{t-1}+v_{t}.
\]
Thus $u_{t}\leq v_{t}\leq-\alpha_{0,T}-\vec{\phi}_{0}^{\trans}\Delta\vec y_{t-1}$
when $y_{t}=0$, whence there exists a $C<\infty$ (depending only
on $m$, $a_{0}$ and $\smlnorm{\vec{\phi}_{0}}$) such that
\begin{align*}
\expect\indic\{y_{t}=0\}\smlabs{v_{t}}^{m+\delta} & \leq\expect(\smlabs{u_{t}}+\smlabs{\alpha_{0,T}}+\smlnorm{\vec{\phi}_{0}}\smlnorm{\Delta\vec y_{t-1}})^{m+\delta}\\
 & \leq C(1+\expect\smlabs{u_{t}}^{m+\delta}+\expect\smlnorm{\Delta\vec y_{t-1}}^{m+\delta}).
\end{align*}
In view of Lemma~\ref{lem:bd2022}\ref{enu:mombnd}, we may take
$\delta\in(0,\delta_{u}-2)$ such that the r.h.s.\ is uniformly bounded,
for all $m\in[1,4]$. Hence, by H{\"o}lder's inequality,
\begin{align*}
\sum_{t=1}^{T}\indic\{y_{t}=0\}\smlabs{v_{t}}^{m} & \leq\left(\sum_{t=1}^{T}\indic\{y_{t}=0\}\right)^{\delta/(m+\delta)}\left(\sum_{t=1}^{T}\indic\{y_{t}=0\}\smlabs{v_{t}}^{m+\delta}\right)^{m/(m+\delta)}\\
 & =o_{p}(T^{\delta/(m+\delta)})\cdot O_{p}(T^{m/(m+\delta)})=o_{p}(T)
\end{align*}
by the preceding and the result of part~\ref{enu:zeros}.
\end{proof}
\begin{proof}[Proof of Lemma~\ref{lem:mills}]
 Both bounds are an immediate consequence of the bound given in (3)
in \citet{Sam53AMS}.
\end{proof}

\section{Maximum likelihood}
\label{sec:mlproof}

To economise on notation, in this section we shall generally write $\hat{\alpha}_{T}^{\mle}$ as merely $\hat{\alpha}_{T}$. Define
\begin{equation*}\begin{split}
\rho&\defeq(\alpha,\beta,\vec{\phi}^{\trans})^{\trans},\qquad  \vartheta=\sigma^{-1},\qquad
\rho_{T,0}\defeq(\alpha_{T,0},\beta_{T,0},\vec{\phi}_{0}^{\trans})^{\trans}, \qquad  \vartheta_0=\sigma_0^{-1},\\
\pi&\defeq(\rho^{\trans},\vartheta)^{\trans},\qquad \pi_{T,0}\defeq(\rho_{T,0}^{\trans},\vartheta_0)^{\trans},
\end{split}\end{equation*}
where the `$0$' subscript denotes the true parameter values.

It will be convenient to rewrite the
estimation problem in terms of local parameters, defined as:
\begin{align*}
\lcof & \defeq\begin{bmatrix}a\\
c\\
\vec f
\end{bmatrix}=\begin{bmatrix}T^{1/2} & 0 & 0\\
0 & T & 0\\
0 & 0 & T^{1/2}
\end{bmatrix}\begin{bmatrix}\alpha-\alpha_{T,0}\\
\beta-\beta_{T,0}\\
\vec{\phi}-\vec{\phi}_{0}
\end{bmatrix}\eqdef D_{\lcof,T}(\rho-\rho_{T,0}), & h & \defeq T^{1/2}(\vartheta-\vartheta_{0}),
\end{align*}
which we collect together as
\begin{equation}
\lpar\defeq\begin{bmatrix}\lcof\\
h
\end{bmatrix}=\begin{bmatrix}D_{\lcof,T} & 0\\
0 & T^{1/2}
\end{bmatrix}\begin{bmatrix}\rho-\rho_{T,0}\\
\vartheta-\vartheta_{0}
\end{bmatrix}\eqdef D_{\lpar,T}(\pi-\pi_{T,0}).\label{eq:p-to-pi}
\end{equation}
For each $T\in\naturals$, there is a bijective mapping between
$\lpar$ and $\pi$. To reparametrise the model in terms of the former,
let
\begin{align*}
\pi_{T}(\lpar) & =\pi_{T,0}+D_{\lpar,T}^{-1}\lpar & \rho_{T}(\lcof) & =\rho_{T,0}+D_{\lcof,T}^{-1}\lcof & \vartheta_{T}(h) & =\vartheta_{0}+T^{-1/2}h,
\end{align*}
denote the model parameters corresponding to given values of
$\lpar=(\lcof^{\trans},h)^{\trans}$.
Since the only constraint on the original parameter space is that
$\vartheta>0$, the parameter space for $\lpar$ is given by $\parset_{T}\defeq\reals^{k+1}\times\{h\in\reals\mid h>-T^{1/2}\vartheta_{0}\}$;
note that for any compact $K\subset\reals^{k+2}$, we will have
$K\subset\parset_{T}$ for all $T$ sufficiently large.

Recalling \eqref{eq:v}--\eqref{eq:ztm1} above, we thus have
\begin{align}
y_{t}-\alpha-\beta y_{t-1}-\vec{\phi}^{\trans}\Delta\vec y_{t-1} & =v_{t}+(\alpha_{0}-\alpha)+(\beta_{0}-\beta)y_{t-1}+(\vec{\phi}_{0}-\vec{\phi})^{\trans}\Delta\vec y_{t-1}\nonumber \\
 & =v_{t}-(\rho-\rho_{0})^{\trans}D_{\lcof,T}z_{t-1,T}=v_{t}-z_{t-1,T}^{\trans}\lcof,\label{eq:rewrite1}
\end{align}
for $v_{t}=u_{t}-y_{t}^{-}$, and
\begin{align}
\alpha+\beta y_{t-1}+\vec{\phi}^{\trans}\Delta\vec y_{t-1} & =[\alpha_{0}+\beta_{0}y_{t-1}+\vec{\phi}_{0}^{\trans}\Delta\vec y_{t-1}]\nonumber \\
 & \qquad\qquad+[(\alpha-\alpha_{0})+(\beta-\beta_{0})y_{t-1}+(\vec{\phi}-\vec{\phi}_{0})^{\trans}\Delta\vec y_{t-1}]\nonumber \\
 & =x_{t-1}+z_{t-1,T}^{\trans}\lcof.\label{eq:rewrite2}
\end{align}

\subsection{Proof of Theorem~\ref{thm:tobitmle}}

We may now proceed to analyse the asymptotic behaviour of the loglikelihood. To that end, let
\begin{equation*}
g_{t}(p)\defeq\log f_{\pi_{T}(\lpar)}(\mathsf{y}_{t}\mid \mathsf{y}_{t-1}, \ldots, \mathsf{y}_{t-k})
\end{equation*}
denote the conditional density of $y_{t}$ given $(y_{t-1}, \ldots, y_{t-k})$,
when $\pi=\pi_{T}(\lpar)$. Then in view of (\ref{eq:conddens}),
(\ref{eq:rewrite1}) and (\ref{eq:rewrite2}), evaluating this at $\mathsf{y}_{t-i} = y_{t-i}$ yields
\begin{align*}
 g_{t}(p)
 & =\indic\{y_{t}>0\}\log\vartheta_{T}(h)\varphi[\vartheta_{T}(h)(y_{t}-\alpha_{T}(\lcof)-\beta_{T}(\lcof)y_{t-1}-\vec{\phi}_{T}(\lcof)^{\trans}\Delta\vec y_{t-1})]\\
 & \qquad\qquad+\indic\{y_{t}=0\}\log\{1-\Phi[\vartheta_{T}(h)(\alpha_{T}(\lcof)+\beta_{T}(\lcof)y_{t-1}+\vec{\phi}_{T}(\lcof)^{\trans}\Delta\vec y_{t-1})]\}\\
 & =\indic\{y_{t}>0\}\log\vartheta_{T}(h)\varphi[\vartheta_{T}(h)(v_{t}-z_{t-1,T}^{\trans}\lcof)]\\
 & \qquad\qquad+\indic\{y_{t}=0\}\log\{1-\Phi[\vartheta_{T}(h)(x_{t-1}+z_{t-1,T}^{\trans}\lcof)]\}\\
 & =[\log\vartheta_{T}(h)-\tfrac{1}{2}\log2\pi]\indic\{y_{t}>0\}-\tfrac{1}{2}\vartheta_{T}^{2}(h)(u_{t}-z_{t-1,T}^{\trans}\lcof)^{2}\indic\{y_{t}>0\}\\
 & \qquad\qquad+\indic\{y_{t}=0\}\log\{1-\Phi[\vartheta_{T}(h)(z_{t-1,T}^{\trans}\lcof-v_{t})]\},
\end{align*}
since when $y_{t}>0$, $y_{t}^{-}=0$, so that $v_t=u_t$, and when $y_{t}=0$,
\[
0=y_{t}=x_{t-1}+v_{t}
\]
so that $v_{t}=-x_{t-1}$. If $\hat{\pi}$ maximises ${\cal L}_{T}(\pi)$,
then $\hat{\lpar}=D_{\lpar,T}(\hat{\pi}-\pi_{T,0})$ maximises
\begin{align*}
\ell_{T}(\lpar)\defeq{\cal L}_{T}[\pi_{T}(\lpar)] & =\sum_{t=1}^{T}g_{t}(p)\\
 & =N_{T}[\log\vartheta_{T}(h)-\tfrac{1}{2}\log2\pi]-\tfrac{1}{2}\vartheta_{T}^{2}(h)\sum_{t=1}^{T}(u_{t}-z_{t-1,T}^{\trans}\lcof)^{2}\indic\{y_{t}>0\}\\
 & \qquad\qquad+\sum_{t=1}^{T}\indic\{y_{t}=0\}\log\{1-\Phi[\vartheta_{T}(h)(z_{t-1,T}^{\trans}\lcof-v_{t})]\}
\end{align*}
where $N_{T}\defeq\sum_{t=1}^{T}\indic\{y_{t}>0\}$.

Let ${\cal S}_{T}(\lpar)$ and ${\cal H}_{T}(\lpar)$ denote the score
(gradient) and Hessian of $\ell_{T}$ at $\lpar$. The proof of the
following is deferred to Appendix~\ref{app:scorehess}.
\begin{prop}
\label{prop:mlescorehess}Suppose \ref{ass:INIT}--\ref{ass:JSR}
and \ref{ass:TOB} hold. Then jointly with $\tfrac1{\sqrt{T}}\sum_{t=1}^{\smlfloor{\tau T}}u_{t}\indist \sigma_0 W(\tau)$,
\begin{enumerate}
\item \label{enu:score}${\cal S}_{T}(0)\indist{\cal S}$, and
\item \label{enu:hess}${\cal H}_{T}(\lpar)\indist{\cal H}$ on $\ucc(\reals^{k+2})$
\end{enumerate}
where
\begin{align*}
{\cal S} & \defeq\sigma_{0}^{-1}\begin{bmatrix}W(1)\\
\int_{0}^{1}Y(\tau)\diff W(\tau)\\
\Omega^{1/2}\xi_{(1)}\\
2^{1/2}\sigma_{0}^{2}\xi_{(2)}
\end{bmatrix} & {\cal H} & \defeq-\sigma_{0}^{-2}\begin{bmatrix}1 & \int_{0}^{1}Y(\tau)\diff\tau & 0 & 0\\
\int_{0}^{1}Y(\tau)\diff\tau & \int_{0}^{1}Y^{2}(\tau)\diff\tau & 0 & 0\\
0 & 0 & \Omega & 0\\
0 & 0 & 0 & 2\sigma_{0}^{4}
\end{bmatrix}
\end{align*}
with $\xi_{(1)}\sim N[0,I_{k-1}]$, $\xi_{(2)}\sim N[0,1]$ and $W(\cdot)$ being mutually independent.
\end{prop}
Since $\ell_{T}$ is twice continuously differentiable on $\parset_{T}$,
we have by Taylor's theorem that for every $\lpar\in\parset_{T}$,
there exists a $\delta\in[0,1]$ such that
\begin{align}
\ell_{T}(\lpar)-\ell_{T}(0) & ={\cal S}_{T}(0)^{\trans}\lpar+\tfrac{1}{2}\lpar^{\trans}{\cal H}_{T}(\delta\lpar)\lpar\nonumber \\
 & ={\cal S}_{T}(0)^{\trans}\lpar+\tfrac{1}{2}\lpar^{\trans}{\cal H}_{T}(0)\lpar+\lpar^{\trans}R_{T}(\delta\lpar)\lpar\label{eq:like-exp}
\end{align}
where $R_{T}(\lpar)\defeq{\cal H}_{T}(\lpar)-{\cal H}_{T}(0)$. By
Proposition~\ref{prop:mlescorehess}, $R_{T}(\lpar)\inprob0$ on
$\ucc(\reals^{k+2})$, and hence
\[
\ell_{T}(\lpar)-\ell_{T}(0)\indist{\cal S}^{\trans}\lpar+\tfrac{1}{2}\lpar^{\trans}{\cal H}\lpar
\]
on $\ucc(\reals^{k+2})$. Since ${\cal H}$ is a.s.\ negative definite,
the r.h.s.\ almost surely has a unique maximiser at $\lpar^{\ast}\defeq-{\cal H}^{-1}{\cal S}$.
With the aid of a convex reparametrisation of $\ell_{T}$, we then
have the following, whose proof appears in Appendix~\ref{app:proofmaximiser}
\begin{prop}
\label{prop:maximiser}If $\hat{\lpar}_{T}=(\hat{\lcof}_{T},\hat{h}_{T})\in\argmax_{\lpar\in\parset_{T}}\ell_{T}(\lpar)$
for all $T$, then $\hat{\lpar}_{T}\indist\lpar^{\ast}$.
\end{prop}
Since we may evidently take $\hat{\lpar}_{T}=D_{\lpar,T}(\hat{\pi}_{T}-\pi_{T,0})$
in the preceding, it follows that
\[
D_{\lpar,T}(\hat{\pi}_{T}-\pi_{T,0})\indist-{\cal H}^{-1}{\cal S}.\tag*{\qedsymbol}
\]

\subsection{Proof of Proposition~\ref{prop:mlescorehess}}

\label{app:scorehess}

Partition the first and second partial derivatives of $g_{t}(p)$
as
\begin{gather*}
\grad_{\lpar}g_{t}(p)=\begin{bmatrix}\grad_{\lcof}\\
\grad_{h}
\end{bmatrix}g_{t}(p)\eqdef\begin{bmatrix}S_{\lcof,t}(\lpar)\\
S_{h,t}(\lpar)
\end{bmatrix}\eqdef S_{t}(\lpar),\\
\grad_{\lpar\lpar}g_{t}(p)=\begin{bmatrix}\grad_{\lcof\lcof} & \grad_{h\lcof}\\
\grad_{\lcof h} & \grad_{hh}
\end{bmatrix}g_{t}(p)=\begin{bmatrix}H_{\lcof\lcof,t}(\lpar) & H_{h\lcof,t}(\lpar)\\
H_{\lcof h,t}(\lpar) & H_{hh,t}(\lpar)
\end{bmatrix}\eqdef H_{\lpar\lpar,t}(\lpar).
\end{gather*}
Recalling the definition of the inverse Mills ratio (see the text
preceding the statement of Lemma~\ref{lem:mills}), differentiation
and straightforward algebra then yields
\begin{align}
S_{\lcof,t}(\lpar) & =\vartheta_{T}\left[\vartheta_{T}(u_{t}-z_{t-1,T}^{\trans}\lcof)\indic\{y_{t}>0\}-\lambda[\vartheta_{T}(z_{t-1,T}^{\trans}\lcof-v_{t})]\indic\{y_{t}=0\}\right]z_{t-1,T},\label{eq:Sqt}
\end{align}
where we have suppressed the argument of $\vartheta_{T}(h)$ to reduce
notational clutter, and
\begin{align}
H_{\lcof\lcof,t}(\lpar) & =-\vartheta_{T}^{2}\left[\indic\{y_{t}>0\}+\lambda^{\prime}[\vartheta_{T}(z_{t-1,T}^{\trans}\lcof-v_{t})]\indic\{y_{t}=0\}\right]z_{t-1,T}z_{t-1,T}^{\trans}.\label{eq:Hqq}
\end{align}
Similarly, differentiating with respect to $h$ yields
\begin{align}
S_{h,t}(\lpar) & =T^{-1/2}\{[\vartheta_{T}^{-1}-\vartheta_{T}(u_{t}-z_{t-1,T}^{\trans}\lcof)^{2}]\indic\{y_{t}>0\}\label{eq:Sht}\\
 & \qquad\qquad\qquad\qquad-(z_{t-1,T}^{\trans}\lcof-v_{t})\lambda[\vartheta_{T}(z_{t-1,T}^{\trans}\lcof-v_{t})]\indic\{y_{t}=0\}\}\nonumber
\end{align}
and
\begin{align}
H_{hh,t}(\lpar) & =T^{-1}\{[-\vartheta_{T}^{-2}-(u_{t}-z_{t-1,T}^{\trans}\lcof)^{2}]\indic\{y_{t}>0\}\label{eq:Hhh}\\
 & \qquad\qquad\qquad\qquad-(z_{t-1,T}^{\trans}\lcof-v_{t})^{2}\lambda^{\prime}[\vartheta_{T}(z_{t-1,T}^{\trans}\lcof-v_{t})]\indic\{y_{t}=0\}\}.\nonumber
\end{align}
Finally,
\begin{align}
H_{\lcof h,t}(\lpar) & =T^{-1/2}\vartheta_{T}^{-1}S_{\lcof,t}(\lpar)\label{eq:Hqh}\\
 & \qquad+T^{-1/2}\vartheta_{T}z_{t-1,T}[(u_{t}-z_{t-1,T}^{\trans}\lcof)\indic\{y_{t}>0\}\nonumber \\
 & \qquad\qquad\qquad\qquad\qquad\qquad-(z_{t-1,T}^{\trans}\lcof-v_{t})\lambda^{\prime}[\vartheta_{T}(z_{t-1,T}^{\trans}\lcof-v_{t})]\indic\{y_{t}=0\}].\nonumber
\end{align}

\subsubsection{Proof of part~\ref{enu:score}}\label{sec:score_appendix}

By construction, the process
\[
{\cal S}_{T}(0)\defeq\begin{bmatrix}{\cal S}_{\lcof,T}(0)\\
{\cal S}_{h,T}(0)
\end{bmatrix}\defeq\sum_{t=1}^{T}\begin{bmatrix}S_{\lcof,t}(0)\\
S_{h,t}(0)
\end{bmatrix}
\]
is a martingale (see e.g., \citealp{HH80}, Ch.\ 6). It follows
from evaluating (\ref{eq:Sqt}) at $\lpar=0$ (so that $\vartheta=\vartheta_{0}$),
in particular, that
\[
{\cal S}_{\lcof,T}(0)=\sum_{t=1}^{T}z_{t-1,T}m_{t},
\]
where
\begin{align*}
m_{t} & \defeq\vartheta_{0}\left[\vartheta_{0}u_{t}\indic\{y_{t}>0\}-\lambda[\vartheta_{0}(-v_{t})]\indic\{y_{t}=0\}\right]\\
 & =\vartheta_{0}\left[\vartheta_{0}u_{t}\indic\{y_{t}>0\}-\lambda(\vartheta_{0}x_{t-1})\indic\{y_{t}=0\}\right]
\end{align*}
is a martingale difference sequence. We first show the following.
\begin{lem}
\label{lem:Sqt}Suppose \ref{ass:INIT}--\ref{ass:JSR} and \ref{ass:TOB}
hold. Then
\[
{\cal S}_{\lcof,T}(0)=\vartheta_{0}^{2}\sum_{t=1}^{T}z_{t-1,T}u_{t}+o_{p}(1)\eqdef{\cal M}_{\lcof,T}+o_{p}(1)
\]
\end{lem}
\begin{proof}
Consider
\begin{align}
\vartheta_{0}^{2}\expect_{t-1}u_{t}^{2}\indic\{y_{t}>0\} & =\vartheta_{0}^{2}\expect_{t-1}u_{t}^{2}\indic\{x_{t-1}+u_{t}>0\} =\expect_{t-1}(\vartheta_{0}u_{t})^{2}\indic\{\vartheta_{0}u_{t}>-\vartheta_{0}x_{t-1}\}\nonumber \\
 & =\int_{-\vartheta_{0}x_{t-1}}^{\infty}u^{2}\varphi(u)\diff u
  =-\vartheta_{0}x_{t-1}\varphi(\vartheta_{0}x_{t-1})+\Phi(\vartheta_{0}x_{t-1})\label{eq:cv1}
\end{align}
since $\vartheta_{0}u_{t}\sim N[0,1]$ under \ref{ass:TOB}, and because
\begin{multline*}
\int_{-z}^{\infty}u^{2}\varphi(u)\diff u=-\int_{-z}^{\infty}u\varphi^{\prime}(u)\diff u=-[u\varphi(u)]_{-z}^{\infty}+\int_{-z}^{\infty}\varphi(u)\diff u\\
=-z\varphi(-z)+[1-\Phi(-z)]=-z\varphi(z)+\Phi(z)
\end{multline*}
via integration by parts. Further, since $y_{t}=0$ if and only if
$x_{t-1}+u_{t}\leq0$,
\begin{align}
\expect_{t-1}\lambda^{2}(\vartheta_{0}x_{t-1})\indic\{y_{t}=0\} & =\lambda^{2}(\vartheta_{0}x_{t-1})\expect_{t-1}\indic\{\vartheta_{0}u_{t}\leq-\vartheta_{0}x_{t-1}\}\nonumber \\
 & =\lambda^{2}(\vartheta_{0}x_{t-1})\Phi(-\vartheta_{0}x_{t-1})
  =\varphi(\vartheta_{0}x_{t-1})\lambda(\vartheta_{0}x_{t-1}).\label{eq:cv2}
\end{align}
Similarly,
\begin{align}
\expect_{t-1}\vartheta_{0}u_{t}\lambda(\vartheta_{0}x_{t-1})\indic\{y_{t}=0\} & =\lambda(\vartheta_{0}x_{t-1})\expect_{t-1}\vartheta_{0}u_{t}\indic\{\vartheta_{0}u_{t}\leq-\vartheta_{0}x_{t-1}\}\nonumber \\
 & =\lambda(\vartheta_{0}x_{t-1})\int_{-\infty}^{-\vartheta_{0}x_{t-1}}u\varphi(u)\diff u
  =-\varphi(\vartheta_{0}x_{t-1})\lambda(\vartheta_{0}x_{t-1})\label{eq:cv3}
\end{align}
since
\[
\int_{-\infty}^{-z}u\varphi(u)\diff u=-\int_{-\infty}^{-z}\varphi^{\prime}(u)\diff u=-\varphi(-z)=-\varphi(z).
\]

Since the events $\{y_{t}>0\}$ and $\{y_{t}=0\}$ are mutually exclusive,
it follows from (\ref{eq:cv1}) and (\ref{eq:cv2}) that
\begin{align*}
\expect_{t-1}m_{t}^{2} & =\vartheta_{0}^{2}\left[\vartheta_{0}^{2}\expect_{t-1}u_{t}^{2}\indic\{y_{t}>0\}+\lambda^{2}(\vartheta_{0}x_{t-1})\indic\{y_{t}=0\}\right]\\
 & =\vartheta_{0}^{2}\left[-\vartheta_{0}x_{t-1}\varphi(\vartheta_{0}x_{t-1})+\Phi(\vartheta_{0}x_{t-1})+\varphi(\vartheta_{0}x_{t-1})\lambda(\vartheta_{0}x_{t-1})\right]\eqdef\vartheta_{0}^{2}\Psi_{0}(\vartheta_{0}x_{t-1})
\end{align*}
and from (\ref{eq:cv1}) and (\ref{eq:cv3}) that
\begin{align*}
\expect_{t-1}m_{t}u_{t} & =\vartheta_{0}^{2}\expect_{t-1}u_{t}^{2}\indic\{y_{t}>0\}-\vartheta_{0}\expect_{t-1}u_{t}\lambda(\vartheta_{0}x_{t-1})\indic\{y_{t}=0\}\\
 & =-\vartheta_{0}x_{t-1}\varphi(\vartheta_{0}x_{t-1})+\Phi(\vartheta_{0}x_{t-1})+\varphi(\vartheta_{0}x_{t-1})\lambda(\vartheta_{0}x_{t-1})=\Psi_{0}(\vartheta_{0}x_{t-1})
\end{align*}
where $\Psi_{0}$ is bounded, and $\Psi_{0}(z)\goesto1$ as $z\goesto\infty$.
Hence
\begin{align}
\expect_{t-1}(m_{t}-\vartheta_{0}^{2}u_{t})^{2} & =\expect_{t-1}m_{t}^{2}-2\vartheta_{0}^{2}\expect_{t-1}m_{t}u_{t}+\vartheta_{0}^{4}\expect_{t-1}u_{t}^{2}=\vartheta_{0}^{2}[1-\Psi_{0}(\vartheta_{0}x_{t-1})].\label{eq:cvmlessu}
\end{align}
Finally, we note additionally that for each $\delta>0$, there exists
a $C<\infty$ such that
\begin{align}
\expect_{t-1}\smlabs{m_{t}}^{2+\delta} & \leq\vartheta_{0}^{2+\delta}C\left[\vartheta_{0}^{2+\delta}\expect_{t-1}\smlabs{u_{t}}^{2+\delta}+\lambda^{2+\delta}(\vartheta_{0}x_{t-1})\expect_{t-1}\indic\{\vartheta_{0}u_{t}\leq-\vartheta_{0}x_{t-1}\}\right]\nonumber \\
 & \leq\vartheta_{0}^{2+\delta}C\left[\vartheta_{0}^{2+\delta}\expect\smlabs{u_{1}}^{2+\delta}+\lambda^{2+\delta}(\vartheta_{0}x_{t-1})[1-\Phi(\vartheta_{0}x_{t-1})]\right],\label{eq:lind}
\end{align}
which is bounded by a constant, since in particular
\[
\lambda^{2+\delta}(z)[1-\Phi(z)]=\varphi(z)\lambda^{1+\delta}(z)\goesto0
\]
as $z\goesto+\infty$, by Lemma~\ref{lem:mills}.

Deduce from (\ref{eq:cvmlessu}) and (\ref{eq:lind}) (with $\delta=2$)
that the martingale
\[
M_{T}\defeq\sum_{t=1}^{T}z_{t-1,T}(m_{t}-\vartheta_{0}^{2}u_{t})
\]
has conditional variance
\begin{align*}
\smlnorm{\smlcv{M_{T}}} & =\norm{\sum_{t=1}^{T}\expect_{t-1}(m_{t}-\vartheta_{0}^{2}u_{t})^{2}z_{t-1,T}z_{t-1,T}^{\trans}}\\
 & \leq\vartheta_{0}^{2}\sum_{t=1}^{T}[1-\Psi_{0}(\vartheta_{0}x_{t-1})]\smlnorm{z_{t-1,T}}^{2}\\
 & \leq\vartheta_{0}^{2}\left(\sum_{t=1}^{T}[1-\Psi_{0}(\vartheta_{0}x_{t-1})]^{2}\right)^{1/2}\left(\sum_{t=1}^{T}\smlnorm{z_{t-1,T}}^{4}\right)^{1/2}=o_{p}(1)
\end{align*}
by Lemmas~\ref{lem:integ}\ref{enu:integ} and \ref{lem:moments}\ref{enu:fourthz}.
Hence $M_{T}=o_{p}(1)$ by Corollary 3.1 in \citet{HH80}.
\end{proof}

Regarding the other component of ${\cal S}_{T}(0)$, we have from
evaluating (\ref{eq:Sht}) at $\lpar=0$ that
\begin{align*}
{\cal S}_{h,T}(0) & =T^{-1/2}\sum_{t=1}^{T}[(\vartheta_{0}^{-1}-\vartheta_{0}u_{t}^{2})\indic\{y_{t}>0\}+v_{t}\lambda(-\vartheta_{0}v_{t})\indic\{y_{t}=0\}]\\
 & =T^{-1/2}\sum_{t=1}^{T}[(\vartheta_{0}^{-1}-\vartheta_{0}u_{t}^{2})\indic\{y_{t}>0\}-x_{t-1}\lambda(\vartheta_{0}x_{t-1})\indic\{y_{t}=0\}].
\end{align*}

\begin{lem}
\label{lem:Sht}Suppose \ref{ass:INIT}--\ref{ass:JSR} and \ref{ass:TOB}
hold. Then
\[
{\cal S}_{h,T}(0)=T^{-1/2}\sum_{t=1}^{T}(\vartheta_{0}^{-1}-\vartheta_{0}u_{t}^{2})+o_{p}(1)\eqdef{\cal M}_{h,T}+o_{p}(1)
\]
\end{lem}
\begin{proof}
We have
\begin{align*}
{\cal S}_{h,T}(0) & =T^{-1/2}\sum_{t=1}^{T}(\vartheta_{0}^{-1}-\vartheta_{0}u_{t}^{2})-T^{-1/2}\sum_{t=1}^{T}\indic\{y_{t}=0\}[(\vartheta_{0}^{-1}-\vartheta_{0}u_{t}^{2})+x_{t-1}\lambda(\vartheta_{0}x_{t-1})]\\
 & \eqdef T^{-1/2}\sum_{t=1}^{T}(\vartheta_{0}^{-1}-\vartheta_{0}u_{t}^{2})+\Delta_{T}.
\end{align*}
Since $\expect_{t-1}(\vartheta_{0}^{-1}-\vartheta_{0}u_{t}^{2})=0$,
the first r.h.s.\ term is a martingale, and hence so too is $\Delta_{T}$.
We have the bound
\[
\smlcv{\Delta_{T}}\leq2T^{-1}\sum_{t=1}^{T}\expect_{t-1}\indic\{y_{t}=0\}[(\vartheta_{0}^{-1}-\vartheta_{0}u_{t}^{2})^{2}+x_{t-1}^{2}\lambda^{2}(\vartheta_{0}x_{t-1})].
\]
To show that the r.h.s.\ is $o_{p}(1)$ we first note that, similarly
to the argument that yielded (\ref{eq:cv2}) and (\ref{eq:lind})
above, for any $\delta\geq0$,
\begin{align}
\expect_{t-1}\indic\{y_{t}=0\}\smlabs{x_{t-1}}^{2+\delta}\lambda^{2+\delta}(\vartheta_{0}x_{t-1}) & =\smlabs{x_{t-1}}^{2+\delta}\lambda^{2+\delta}(\vartheta_{0}x_{t-1})\Phi(-\vartheta_{0}x_{t-1})\nonumber \\
 & =\vartheta_{0}^{-(2+\delta)}\smlabs{\vartheta_{0}x_{t-1}}^{2+\delta}\lambda^{1+\delta}(\vartheta_{0}x_{t-1})\varphi(\vartheta_{0}x_{t-1})\nonumber \\
 & \eqdef\vartheta_{0}^{-(2+\delta)}G_{\delta}(\vartheta_{0}x_{t-1}),\label{eq:Gdelta}
\end{align}
where $G_{\delta}$ is bounded, and $G_{\delta}(z)\goesto0$ as $z\goesto+\infty$,
by Lemma~\ref{lem:mills}. Further, by H{\"o}lder's inequality
and Lemma~\ref{lem:moments}\ref{enu:zeros}, for any $\delta>0$,
\begin{multline*}
T^{-1}\sum_{t=1}^{T}\indic\{y_{t}=0\}(\vartheta_{0}^{-1}-\vartheta_{0}u_{t}^{2})^{2}\\
\leq T^{-1}\left(\sum_{t=1}^{T}\left(\frac{1}{\vartheta_{0}}-\vartheta_{0}u_{t}^{2}\right)^{2(1+\delta)}\right)^{1/(1+\delta)}\left(\sum_{t=1}^{T}\indic\{y_{t}=0\}\right)^{\delta/(1+\delta)}\inprob0.
\end{multline*}
Since
\[
\norm{T^{-1}\sum_{t=1}^{T}\left(\frac{1}{\vartheta_{0}}-\vartheta_{0}u_{t}^{2}\right)^{2}\indic\{y_{t}=0\}}_{1+\delta}\leq\norm{\left(\frac{1}{\vartheta_{0}}-\vartheta_{0}u_{1}^{2}\right)^{2}}_{1+\delta}<\infty
\]
it follows that the preceding convergence also holds in $L^{1}$.
Hence by the above and Lemma~\ref{lem:integ}\ref{enu:integ},
\[
\smlcv{\Delta_{T}}\leq o_{p}(1)+2\vartheta_{0}^{-1}T^{-1}\sum_{t=1}^{T}G_{0}(\vartheta_{0}x_{t-1})=o_{p}(1).
\]
Deduce that $\Delta_{T}=o_{p}(1)$ by Corollary 3.1 in \citet{HH80}.
\end{proof}

In view of Lemmas~\ref{lem:Sqt} and \ref{lem:Sht}, it remains to
determine the (joint) limiting distribution of the martingales ${\cal M}_{\lcof,T}$
and ${\cal M}_{h,T}$, whose definitions appear in the statements
of those lemmas. Note
\[
{\cal M}_{\lcof,T}=\vartheta_{0}^{2}\sum_{t=1}^{T}\begin{bmatrix}T^{-1/2}\\
T^{-1}y_{t-1}\\
T^{-1/2}\Delta\vec y_{t-1}
\end{bmatrix}u_{t}\eqdef\begin{bmatrix}{\cal M}_{a,T}\\
{\cal M}_{c,T}\\
{\cal M}_{\vec f,T}
\end{bmatrix}.
\]
We first consider ${\cal M}_{\vec f,T}$. By Lemma~\ref{lem:bd2022}\ref{enu:deltayt},
this has conditional variance
\begin{align*}
\vartheta_{0}^{4}T^{-1}\sum_{t=1}^{T}\Delta\vec y_{t-1}\Delta\vec y_{t-1}^{\trans}\expect_{t-1}u_{t}^{2} & \inprob\sigma_{0}^{-2}\Omega,
\end{align*}
where the convergence holds in probability since the limit is degenerate.
Further, by that same result,
\[
\sum_{t=1}^{T}\smlnorm{z_{t-1,T}}^{4}\expect_{t-1}\smlabs{u_{t}}^{4}\leq C\sum_{t=1}^{T}\smlnorm{z_{t-1,T}}^{4}=O_{p}(T^{-1})
\]
so that ${\cal M}_{\vec f,T}$ satisfies a conditional Lyapunov condition
(and therefore also a conditional Lindeberg condition). Hence by Corollary
3.1 in \citet{HH80},
\begin{equation}
{\cal M}_{\vec f,T}\indist\sigma_{0}^{-1}\Omega^{1/2}\xi_{(1)},\label{eq:wkcMfT}
\end{equation}
where $\xi_{(1)}\sim N[0,I_{k-1}]$. Next, since ${\cal M}_{h,T}$ is
merely the (scaled) sum of i.i.d.\ random variables (with finite
variances), and
\[
\expect(\vartheta_{0}^{-1}-\vartheta_{0}u_{t}^{2})^{2}=\sigma_{0}^{2}\expect(1-(\sigma_{0}^{-1}u_{t})^{2})^{2}=\sigma_{0}^{2}(1-2+3)=2\sigma_{0}^{2},
\]
since $\sigma_{0}^{-1}u_{t}\sim N[0,1]$, it also follows by the central
limit theorem that
\begin{equation}
{\cal M}_{h,T}\indist2^{1/2}\sigma_{0}\xi_{(2)}\label{eq:wkcMhT}
\end{equation}
for $\xi_{(2)}\sim N[0,1]$.

Our next step is to establish that the marginal convergences (\ref{eq:wkcMfT})
and (\ref{eq:wkcMhT}) hold jointly with the weak convergence of $({\cal M}_{a,T},{\cal M}_{c,T})$,
to mutually independent limits. To that end, we first observe that
${\cal M}_{a,T}=\vartheta_{0}^{2}U_{T}(1)$, where $U_{T}(\tau)\defeq T^{-1/2}\sum_{t=1}^{\smlfloor{\tau T}}u_{t}$.
Now fix $\{\lambda_{i}\}_{i=1}^{m}\subset[0,1]$ with $\lambda_{i}<\lambda_{i+1}$,
and consider the martingale
\[
{\cal N}_{T}(\{\lambda_{i}\})\defeq T^{-1/2}\sum_{t=1}^{T}\begin{bmatrix}\indic\{\lambda_{1}<t/T\leq\lambda_{2}\}u_{t}\\
\vdots\\
\indic\{\lambda_{m-1}<t/T\leq\lambda_{m}\}u_{t}\\
\vartheta_{0}^{2}\Delta\vec y_{t-1}u_{t}\\
\vartheta_{0}^{-1}-\vartheta_{0}u_{t}^{2}
\end{bmatrix}=\begin{bmatrix}{\cal U}_{T}(\lambda_{1},\lambda_{2})\\
\vdots\\
{\cal U}_{T}(\lambda_{m-1},\lambda_{m})\\
{\cal M}_{\vec f,T}\\
{\cal M}_{h,T}
\end{bmatrix},
\]
where ${\cal U}_{T}(\lambda_{1},\lambda_{2})=U_{T}(\lambda_{2})-U_{T}(\lambda_{1})$.
Since each element of the above satisfies a conditional Lindeberg
condition, by Corollary 3.1 in \citet{HH80}, and the Cram{\'e}r--Wold
device, it suffices to show that all their conditional covariances
converge in probability to zero. We have
\begin{align*}
\smlcv{{\cal U}_{T}(\lambda_{i-1},\lambda_{i}),{\cal M}_{\vec f,T}} & =T^{-1}\sum_{t=1}^{T}\indic\{\lambda_{i-1}<t/T\leq\lambda_{i}\}\Delta\vec y_{t-1}\expect_{t-1}u_{t}^{2}\\
 & =\sigma_{0}^{2}T^{-1}\sum_{t=\smlfloor{\lambda_{i-1}T}+1}^{\smlfloor{\lambda_{i}T}}\Delta\vec y_{t-1}\\
 & =\sigma_{0}^{2}T^{-1}(\vec y_{\smlfloor{\lambda_{i}T}}-\vec y_{\smlfloor{\lambda_{i-1}T}})=O_{p}(T^{-1/2})
\end{align*}
by Theorem~3.2 in \citet{BD22},
\[
\smlcv{{\cal U}_{T}(\lambda_{i-1},\lambda_{i}),{\cal M}_{h,T}}=T^{-1}\sum_{t=1}^{T}\indic\{\lambda_{i-1}<t/T\leq\lambda_{i}\}\expect_{t-1}(\vartheta_{0}^{-1}-\vartheta_{0}u_{t}^{2})u_{t}=0,
\]
\[
\smlcv{{\cal M}_{\vec f,T},{\cal M}_{h,T}}=T^{-1}\sum_{t=1}^{T}\Delta\vec y_{t-1}\expect_{t-1}(\vartheta_{0}-\vartheta_{0}^3 u_{t}^{2})u_{t}=0,
\]
and for $i\neq j$,
\[
\smlcv{{\cal U}_{T}(\lambda_{i-1},\lambda_{i}),{\cal U}_{T}(\lambda_{j-1},\lambda_{j})}=0.
\]
It follows that ${\cal N}_{T}(\{\lambda_{i}\})$ weakly converges
to a normal random vector with variance matrix
\[
\diag\{\sigma_{0}^{2}(\lambda_{2}-\lambda_{1}),\ldots,\sigma_{0}^{2}(\lambda_{m}-\lambda_{m-1}),\sigma_{0}^{-2}\Omega,2\sigma_{0}^{2}\}.
\]
Deduce that (\ref{eq:wkcMfT}) and (\ref{eq:wkcMhT}) hold jointly with
$U_{T}(\cdot)\indist \sigma_0W(\cdot)$ (the latter, on $D[0,1]$, by the functional central
limit theorem), with $\xi_{(1)}$, $\xi_{(2)}$ and $W$ being mutually
independent.

Finally, we note that by Theorem~3.2 in \citet{BD22} and Theorem~2.1
in \citet{ito_convergence}, the convergences
\[
\begin{bmatrix}{\cal M}_{a,T}\\
{\cal M}_{c,T}
\end{bmatrix}=\vartheta_{0}^{2}\begin{bmatrix}U_{T}(1)\\
T^{-1}\sum_{t=1}^{T}y_{t-1}u_{t}
\end{bmatrix}\indist\sigma_{0}^{-2}\begin{bmatrix}\sigma_0W(1)\\
\sigma_0\int_{0}^{1}Y(\tau)\diff W(\tau)
\end{bmatrix}
=\sigma_0^{-1}\begin{bmatrix}W(1)\\
\int_{0}^{1}Y(\tau)\diff W(\tau)
\end{bmatrix}
\]
hold jointly with $Y_{T}\indist Y$ and $U_{T}\indist \sigma_0 W$ (on $D[0,1]$).
Since $Y$ is a function only of $W$, it follows that the preceding
holds jointly with (\ref{eq:wkcMfT}) and (\ref{eq:wkcMhT}), to mutually
independent limits.\hfill{}\qedsymbol{}

\subsubsection{Proof of part~\ref{enu:hess}}\label{sec:Hessian_appendix}

We partition the Hessian as
\[
{\cal H}_{T}(\lpar)=\begin{bmatrix}{\cal H}_{\lcof\lcof,T}(\lpar) & {\cal H}_{h\lcof,T}(\lpar)\\
{\cal H}_{\lcof h,T}(\lpar) & {\cal H}_{hh,T}(\lpar)
\end{bmatrix}\defeq\sum_{t=1}^{T}\begin{bmatrix}H_{\lcof\lcof,t}(\lpar) & H_{h\lcof,t}(\lpar)\\
H_{\lcof h,t}(\lpar) & H_{hh,t}(\lpar)
\end{bmatrix},
\]
and consider each of these components in turn.

\paragraph{$\boldsymbol{{\cal H}_{\protect\lcof\protect\lcof,T}}$.}

From (\ref{eq:Hqq}), we have
\begin{align}
{\cal H}_{\lcof\lcof,T}(\lpar) & =-\sum_{t=1}^{T}\vartheta_{T}^{2}(h)\left[\indic\{y_{t}>0\}+\lambda^{\prime}[\vartheta_{T}(h)(z_{t-1,T}^{\trans}\lcof-v_{t})]\indic\{y_{t}=0\}\right]z_{t-1,T}z_{t-1,T}^{\trans}\nonumber \\
 & =-\vartheta_{T}^{2}(h)\sum_{t=1}^{T}\indic\{y_{t}>0\}z_{t-1,T}z_{t-1,T}^{\trans}\nonumber \\
 & \qquad\qquad-\vartheta_{T}^{2}(h)\sum_{t=1}^{T}\indic\{y_{t}=0\}\lambda^{\prime}[\vartheta_{T}(h)(z_{t-1,T}^{\trans}\lcof-v_{t})]z_{t-1,T}z_{t-1,T}^{\trans}\nonumber \\
 & \eqdef{\cal H}_{\lcof\lcof,T}^{0}(\lpar)+{\cal R}_{\lcof\lcof,T}(\lpar),\label{eq:Hqqdecomp}
\end{align}
where ${\cal H}_{\lcof\lcof,T}^{0}(\lpar)$ depends on $\lpar$ only
through $h$. Further,
\begin{align*}
\sum_{t=1}^{T}\indic\{y_{t}>0\}z_{t-1,T}z_{t-1,T}^{\trans} & =\sum_{t=1}^{T}z_{t-1,T}z_{t-1,T}^{\trans}-\sum_{t=1}^{T}\indic\{y_{t}=0\}z_{t-1,T}z_{t-1,T}^{\trans}\indist Q_{ZZ}
\end{align*}
by Lemma~\ref{lem:moments}\ref{enu:ZZ} and the fact that, by the
Cauchy--Schwarz (CS) inequality,
\[
\norm{\sum_{t=1}^{T}\indic\{y_{t}=0\}z_{t-1,T}z_{t-1,T}^{\trans}}^{2}\leq\sum_{t=1}^{T}\indic\{y_{t}=0\}\sum_{t=1}^{T}\smlnorm{z_{t-1,T}}^{4}=o_{p}(T)\cdot O_{p}(T^{-1})=o_{p}(1)
\]
by parts \ref{enu:zeros} and \ref{enu:fourthz} of Lemma~\ref{lem:moments}.
Since $\vartheta_{T}^{2}(h)\inprob\vartheta_{0}^{2}$ uniformly on
compacta, it therefore follows that ${\cal H}_{\lcof\lcof,T}^{0}(\lpar)\indist-\vartheta_{0}^{2}Q_{ZZ}$
on $\ucc(\reals^{k+2})$.

Regarding the second r.h.s.\ term in (\ref{eq:Hqqdecomp}), we note
by Lemma~\ref{lem:mills} $\smlabs{\lambda^{\prime}(x)}$ is bounded,
so there exists a $C<\infty$ such that
\begin{align*}
\smlnorm{{\cal R}_{\lcof\lcof,T}(\lpar)} & \leq C\vartheta_{T}^{2}(h)\sum_{t=1}^{T}\indic\{y_{t}=0\}\smlnorm{z_{t-1,T}}^{2}
\end{align*}
The sum on the r.h.s.\ is $o_{p}(1)$, since parts \ref{enu:zeros}
and \ref{enu:fourthz} of Lemma~\ref{lem:moments}, and the CS inequality,
yield
\begin{equation}
\left(\sum_{t=1}^{T}\indic\{y_{t}=0\}\smlnorm{z_{t-1,T}}^{2}\right)^{2}\leq\sum_{t=1}^{T}\indic\{y_{t}=0\}\sum_{t=1}^{T}\smlnorm{z_{t-1,T}}^{4}=o_{p}(1).\label{eq:y0vz}
\end{equation}
Deduce that ${\cal R}_{\lcof\lcof,T}(\lpar)\inprob0$ uniformly on
compacta, whence ${\cal H}_{\lcof\lcof,T}(\lpar)\indist-\vartheta_{0}^{2}Q_{ZZ}$
on $\ucc(\reals^{k+2})$.

\paragraph{$\boldsymbol{{\cal H}_{hh,T}}$.}

From (\ref{eq:Hhh}), we have
\begin{align*}
{\cal H}_{hh,T}(\lpar) & =-T^{-1}\sum_{t=1}^{T}(\vartheta_{T}^{-2}(h)+u_{t}^{2})+T^{-1}\sum_{t=1}^{T}(\vartheta_{T}^{-2}(h)+u_{t}^{2})\indic\{y_{t}=0\}\\
 & \qquad\qquad+T^{-1}\sum_{t=1}^{T}[2(z_{t-1,T}^{\trans}\lcof)u_{t}-(z_{t-1,T}^{\trans}\lcof)^{2}]\indic\{y_{t}>0\}\\
 & \qquad\qquad-T^{-1}\sum_{t=1}^{T}(z_{t-1,T}^{\trans}\lcof-v_{t})^{2}\lambda^{\prime}[\vartheta_{T}(h)(z_{t-1,T}^{\trans}\lcof-v_{t})]\indic\{y_{t}=0\}\}\\
 & \eqdef{\cal H}_{hh,T}^{0}(\lpar)+{\cal R}_{hh,T}^{1}(\lpar)+{\cal R}_{hh,T}^{2}(\lpar)+{\cal R}_{hh,T}^{3}(\lpar),
\end{align*}
where, by the law of large numbers (LLN),
\[
{\cal H}_{hh,T}^{0}(\lpar)=-\vartheta_{T}^{-2}(h)-T^{-1}\sum_{t=1}^{T}u_{t}^{2}\inprob-2\sigma_{0}^{2}
\]
on $\ucc(\reals^{k+2})$. Regarding the ${\cal R}_{hh,T}^{i}$ terms,
we first note that by the CS inequality and Lemma~\ref{lem:moments}\ref{enu:zeros}
\begin{align*}
\smlabs{{\cal R}_{hh,T}^{1}(\lpar)} & \leq\vartheta_{T}^{-2}(h)T^{-1}\sum_{t=1}^{T}\indic\{y_{t}=0\}+T^{-1}\left(\sum_{t=1}^{T}\indic\{y_{t}=0\}\right)^{1/2}\left(\sum_{t=1}^{T}u_{t}^{4}\right)^{1/2}\inprob0
\end{align*}
uniformly on compacta. Moreover, by parts \ref{enu:ZZ} and \ref{enu:xprodpos}
of Lemma~\ref{lem:moments},
\begin{align*}
\smlabs{{\cal R}_{hh,T}^{2}(\lpar)} & \leq2\smlnorm{\lcof}\norm{T^{-1}\sum_{t=1}^{T}z_{t-1,T}u_{t}\indic\{y_{t}>0\}}+\smlnorm{\lcof}^{2}T^{-1}\sum_{t=1}^{T}\smlnorm{z_{t-1,T}}^{2}\\
 & =2\smlnorm{\lcof}o_{p}(T^{-1/2})+\smlnorm{\lcof}^{2}O_{p}(T^{-1})\inprob0
\end{align*}
uniformly on compacta. Finally, by Lemma~\ref{lem:mills} there exists
a $C<\infty$ such that
\begin{align*}
\smlabs{{\cal R}_{hh,T}^{3}(\lpar)} & \leq CT^{-1}\sum_{t=1}^{T}\indic\{y_{t}=0\}(\smlnorm{\lcof}^{2}\smlnorm{z_{t-1,T}}^{2}+\smlabs{v_{t}}^{2})\\
 & =C\left[\smlnorm{\lcof}^{2}T^{-1}\sum_{t=1}^{T}\indic\{y_{t}=0\}\smlnorm{z_{t-1,T}}^{2}+T^{-1}\sum_{t=1}^{T}\indic\{y_{t}=0\}\smlabs{v_{t}}^{2}\right].
\end{align*}
That both of the sums on the r.h.s.\ are $o_{p}(1)$ then follows
from (\ref{eq:y0vz}) above, and Lemma~\ref{lem:moments}\ref{enu:vtzero}.
Deduce that ${\cal R}_{hh,T}^{3}(\lpar)\inprob0$ uniformly on compacta.
It follows that ${\cal H}_{hh,T}(\lpar)\inprob-2\sigma_{0}^{2}$ on
$\ucc(\reals^{k+2})$.

\paragraph{$\boldsymbol{{\cal H}_{\protect\lcof h,T}}$.}

We have from (\ref{eq:Hqh}) above that
\begin{align*}
{\cal H}_{\lcof h,T}(\lpar) & =\vartheta_{T}^{-1}(h)T^{-1/2}{\cal S}_{\lcof,T}(\lpar)\\
 & \qquad\qquad+\vartheta_{T}(h)T^{-1/2}\sum_{t=1}^{T}\indic\{y_{t}>0\}z_{t-1,T}(u_{t}-z_{t-1,T}^{\trans}\lcof)\\
 & \qquad\qquad-\vartheta_{T}(h)T^{-1/2}\sum_{t=1}^{T}\indic\{y_{t}=0\}z_{t-1,T}(z_{t-1,T}^{\trans}\lcof-v_{t})\lambda^{\prime}[\vartheta_{T}(z_{t-1,T}^{\trans}\lcof-v_{t})]\\
 & \eqdef T^{-1/2}\vartheta_{T}^{-1}(h){\cal S}_{\lcof,T}(\lpar)+\vartheta_{T}(h){\cal R}_{\lcof,T}^{1}(\lpar)+\vartheta_{T}(h){\cal R}_{\lcof,T}^{2}(\lpar).
\end{align*}
It follows from parts \ref{enu:ZZ} and \ref{enu:xprodpos} of Lemma~\ref{lem:moments}
that
\[
\smlabs{{\cal R}_{\lcof,T}^{1}(\lpar)}\leq\norm{T^{-1/2}\sum_{t=1}^{T}\indic\{y_{t}>0\}z_{t-1,T}u_{t}}+\smlnorm{\lcof}^{2}T^{-1/2}\sum_{t=1}^{T}\smlnorm{z_{t-1,T}}^{2}\inprob0
\]
uniformly on compacta. Moreover, by Lemma~\ref{lem:mills} there
exists a $C<\infty$ such that
\begin{align*}
\smlabs{{\cal R}_{\lcof,T}^{2}(\lpar)} & \leq CT^{-1/2}\sum_{t=1}^{T}\indic\{y_{t}=0\}\smlnorm{z_{t-1,T}}(\smlnorm{\lcof}\smlnorm{z_{t-1,T}}+\smlabs{v_{t}})\\
 & \leq CT^{-1/2}\left[\smlnorm{\lcof}\sum_{t=1}^{T}\smlnorm{z_{t-1,T}^{\trans}}^{2}+\sum_{t=1}^{T}\indic\{y_{t}=0\}\smlnorm{z_{t-1,T}^{\trans}}\smlabs{v_{t}}\right]
\end{align*}
That each of the sums on the r.h.s.\ is $o_{p}(T^{1/2})$ follows
from parts \ref{enu:ZZ} and \ref{enu:vtzero} of Lemma~\ref{lem:moments},
noting in particular that
\begin{equation}
\sum_{t=1}^{T}\indic\{y_{t}=0\}\smlnorm{z_{t-1,T}^{\trans}}\smlabs{v_{t}}\leq\left(\sum_{t=1}^{T}\smlnorm{z_{t-1,T}^{\trans}}^{2}\right)^{1/2}\left(\sum_{t=1}^{T}\indic\{y_{t}=0\}\smlabs{v_{t}}^{2}\right)^{1/2}=o_{p}(T^{1/2})\label{eq:y0zv}
\end{equation}
by the CS inequality. Deduce that ${\cal R}_{\lcof,T}^{2}(\lpar)\inprob0$
uniformly on compacta. Finally, we have from (\ref{eq:Sqt}) that
\begin{multline*}
{\cal S}_{\lcof,T}(\lpar)=\vartheta_{T}^{2}(h)\sum_{t=1}^{T}z_{t-1,T}(u_{t}-z_{t-1,T}^{\trans}\lcof)\indic\{y_{t}>0\}\\
-\vartheta_{T}(h)\sum_{t=1}^{T}z_{t-1,T}\lambda[\vartheta_{T}(h)(z_{t-1,T}^{\trans}\lcof-v_{t})]\indic\{y_{t}=0\}.
\end{multline*}
It follows from Lemma~\ref{lem:mills} that there exists a $C<\infty$
such that
\[
\lambda[\vartheta_{T}(h)(z_{t-1,T}^{\trans}\lcof-v_{t})]\leq C[1+\vartheta_{T}(h)(\smlnorm{\lcof}\smlnorm{z_{t-1,T}}+\smlabs{v_{t}})].
\]
Hence
\begin{align*}
\smlabs{{\cal S}_{\lcof,T}(\lpar)} & \leq\vartheta_{T}^{2}(h)\norm{\sum_{t=1}^{T}z_{t-1,T}u_{t}\indic\{y_{t}>0\}}+\smlnorm{\lcof}\vartheta_{T}^{2}(h)\sum_{t=1}^{T}\smlnorm{z_{t-1,T}}^{2}\\
 & \qquad\qquad+C\vartheta_{T}(h)\Biggl[\sum_{t=1}^{T}\indic\{y_{t}=0\}\smlnorm{z_{t-1,T}}+\vartheta_{T}(h)\smlnorm{\lcof}\sum_{t=1}^{T}\smlnorm{z_{t-1,T}}^{2}\\
 & \qquad\qquad\qquad\qquad\qquad\qquad+\vartheta_{T}(h)\sum_{t=1}^{T}\indic\{y_{t}=0\}\smlnorm{z_{t-1,T}}\smlabs{v_{t}}\Biggr].
\end{align*}
That each of the sums on the r.h.s.\ are $o_{p}(T^{1/2})$ follows
from parts \ref{enu:ZZ} and \ref{enu:vtzero} of Lemma~\ref{lem:moments},
(\ref{eq:y0zv}) above, and the fact that
\[
\sum_{t=1}^{T}\indic\{y_{t}=0\}\smlnorm{z_{t-1,T}}\leq\left(\sum_{t=1}^{T}\indic\{y_{t}=0\}\right)^{1/2}\left(\sum_{t=1}^{T}\smlnorm{z_{t-1,T}}^{2}\right)^{1/2}=o_{p}(T^{1/2})
\]
by the CS inequality. Deduce that $T^{-1/2}{\cal S}_{\lcof,T}(\lpar)\inprob0$
uniformly on compacta. Hence ${\cal H}_{\lcof h,T}(\lpar)\inprob0$
uniformly on compacta.

\paragraph{Joint convergence of ${\cal S}_{T}(0)$ and ${\cal H}_{T}(\protect\lpar)$.}

It follows from the preceding that ${\cal H}_{T}(\lpar)\indist{\cal H}$
on $\ucc(\reals^{k+2})$. In view of Lemma~\ref{lem:moments}\ref{enu:ZZ},
this convergence holds jointly with $U_{T}\indist \sigma_0 W$, with ${\cal H}$
being $\sigma(W)$-measurable. We have shown in the proof of part~\ref{enu:score}
that the convergence ${\cal S}_{T}(0)\indist{\cal S}$ also holds jointly
with $U_{T}\indist \sigma_0 W$, with each element of ${\cal S}$ either being
$\sigma(W)$-measurable, or independent of $W$. Deduce that ${\cal S}_{T}(0)\indist{\cal S}$
jointly with ${\cal H}_{T}(\lpar)\indist{\cal H}$, as claimed.\hfill{}\qedsymbol{}

\subsection{Proof of Proposition~\ref{prop:maximiser}}

\label{app:proofmaximiser}

Given $\lcof\in\reals^{k+1}$, $h\in\reals$ and $T\in\naturals$,
define $\tilde{\lcof}\in\reals^{k+1}$ to be such that $\lcof=\vartheta_{T}^{-1}(h)\tilde{\lcof}$.
Similarly, to \citet{Ols78Ecta}, setting $\lpar_{T}(\tilde{\lcof},h)\defeq(\vartheta_{T}^{-1}(h)\tilde{\lcof}^{\trans},h)^{\trans}$,
we have
\begin{align*}
\tilde{\ell}_{T}(\tilde{\lcof},h) & \defeq\ell_{T}[\lpar_{T}(\tilde{\lcof},h)]\\
 & =N_{T}[\log\vartheta_{T}(h)-\tfrac{1}{2}\log2\pi]-\tfrac{1}{2}\vartheta_{T}^{2}(h)\sum_{t=1}^{T}(u_{t}-\vartheta_{T}^{-1}(h)z_{t-1,T}^{\trans}\tilde{\lcof})^{2}\indic\{y_{t}>0\}\\
 & \qquad\qquad+\sum_{t=1}^{T}\indic\{y_{t}=0\}\log\{1-\Phi[\vartheta_{T}(h)(\vartheta_{T}^{-1}(h)z_{t-1,T}^{\trans}\tilde{\lcof}-v_{t})]\}\\
 & =N_{T}[\log(\vartheta_{0}+T^{-1/2}h)-\tfrac{1}{2}\log2\pi]-\tfrac{1}{2}\sum_{t=1}^{T}(\vartheta_{0}u_{t}+hT^{-1/2}u_{t}-z_{t-1,T}^{\trans}\tilde{\lcof})^{2}\indic\{y_{t}>0\}\\
 & \qquad\qquad+\sum_{t=1}^{T}\indic\{y_{t}=0\}\log\{1-\Phi[(z_{t-1,T}^{\trans}\tilde{\lcof}-hT^{-1/2}v_{t}-\vartheta_{0}v_{t})]\}.
\end{align*}
The penultimate line in the preceding display is clearly concave in
$(\tilde{\lcof},h)$, while the concavity of the final line follows
from log-concavity of the normal distribution (e.g., \citealp{BB05ET}).
Hence $\tilde{\ell}_{T}$ is concave.

It follows from (\ref{eq:like-exp}) that for every $(\tilde{\lcof},h)\in\reals^{k+2}$
and $T$ sufficiently large, there exists a $\lambda\in[0,1]$ such
that
\begin{align*}
\tilde{\ell}_{T}(\tilde{\lcof},h)-\tilde{\ell}_{T}(0,0) & =\ell_{T}[\lpar_{T}(\tilde{\lcof},h)]-\ell_{T}(0)\\
 & ={\cal S}_{T}(0)^{\trans}\lpar_{T}(\tilde{\lcof},h)+\tfrac{1}{2}\lpar_{T}(\tilde{\lcof},h)^{\trans}{\cal H}_{T}(0)\lpar_{T}(\tilde{\lcof},h)\\
 & \qquad\qquad+\lpar_{T}(\tilde{\lcof},h)^{\trans}{\cal R}_{T}[\lambda\lpar_{T}(\tilde{\lcof},h)]\lpar_{T}(\tilde{\lcof},h).
\end{align*}
Since for each $(\tilde{\lcof},h)\in\reals^{k+2}$,
\[
\lpar_{T}(\tilde{\lcof},h)=\begin{bmatrix}\vartheta_{T}^{-1}(h)\tilde{\lcof}\\
h
\end{bmatrix}\goesto\begin{bmatrix}\vartheta_{0}^{-1}\tilde{\lcof}\\
h
\end{bmatrix}=\begin{bmatrix}\vartheta_{0}^{-1}I_{k+1} & 0\\
0 & 1
\end{bmatrix}\begin{bmatrix}\tilde{\lcof}\\
h
\end{bmatrix}\eqdef M\tilde{\lpar}
\]
and ${\cal R}_{T}(\lpar)\inprob0$ uniformly on compacta, it follows
that
\[
\tilde{\ell}_{T}(\tilde{\lpar})-\tilde{\ell}_{T}(0)\indist{\cal S}^{\trans}(M\tilde{\lpar})+\tfrac{1}{2}(M\tilde{\lpar})^{\trans}{\cal H}(M\tilde{\lpar})
\]
in the sense of finite-dimensional convergence. Since both the left
and right hand sides of the preceding display are concave, the l.h.s.\ is
maximised at $(\vartheta_{T}(\hat{h}_{T})\hat{\lcof}_{T},\hat{h}_{T})$
and the r.h.s.\ almost surely has a unique maximum at $\tilde{\lpar}^{\ast}\defeq-M^{-1}{\cal H}^{-1}{\cal S}$,
it follows by Lemma~A in \citet{Kni89CJS} that
\[
M^{-1}\hat{\lpar}_{T}=\begin{bmatrix}\vartheta_{T}(\hat{h}_{T})\hat{\lcof}_{T}\\
\hat{h}_{T}
\end{bmatrix}+o_{p}(1)\indist\tilde{\lpar}^{\ast}=-M^{-1}{\cal H}^{-1}{\cal S}.
\]
Finally, applying the delta method to $\sigma^2=\vartheta^{-2}$, we get the variance of $\sqrt{T}(\sigma^2-\sigma_0^2)$ from the variance of $\sqrt{T}(\vartheta-\vartheta_0)$.\hfill{}\qedsymbol{}

\subsection{Proof of Theorem~\ref{thm:tobitqmle}}

Retain the local parametrisation $\lpar=(\lcof^{\trans},h)^{\trans}$ and the
Gaussian criterion $\ell_T(\lpar)$ introduced in the Proof of Theorem~\ref{thm:tobitmle}.
Although $\ell_T$ is now a pseudo-loglikelihood, the derivative formulae
\eqref{eq:Sqt}--\eqref{eq:Hqh} remain unchanged. Let
$I_t^0\defeq\indic\{y_t=0\}$ and $N_T^0\defeq\sum_{t=1}^T I_t^0$.
On $\{y_t=0\}$, $v_t=-x_{t-1}$. Hence, by Lemma~\ref{lem:mills},
$$
\lambda(\vartheta_0x_{t-1})I_t^0\leq \tilde{C}(1+\vartheta_0|v_t|)I_t^0\leq C(1+|v_t|)I_t^0.
$$
Using Lemma~\ref{lem:moments}\ref{enu:xprodpos} and Cauchy--Schwarz,
$$
\|\mathcal S_{\lcof,T}(0)\|
\leq C\left\|\sum_{t=1}^Tz_{t-1,T}u_t\indic\{y_t>0\}\right\|
   +C\sum_{t=1}^TI_t^0\|z_{t-1,T}\|(1+|v_t|)=o_p(T^{1/2}),
$$
where
\[
\sum_{t=1}^TI_t^0\|z_{t-1,T}\|
 \leq (N_T^0)^{1/2}\Bigl(\sum_{t=1}^T\|z_{t-1,T}\|^2\Bigr)^{1/2}
 =o_p(T^{1/2})
\]
and
\[
\sum_{t=1}^TI_t^0\|z_{t-1,T}\||v_t|
 \leq\Bigl(\sum_{t=1}^T\|z_{t-1,T}\|^2\Bigr)^{1/2}
      \Bigl(\sum_{t=1}^TI_t^0v_t^2\Bigr)^{1/2}
 =o_p(T^{1/2})
\]
by Lemma~\ref{lem:moments}\ref{enu:ZZ}, \ref{enu:zeros} and \ref{enu:vtzero}.

For the scale score, write $d_t\defeq\vartheta_0^{-1}-\vartheta_0u_t^2$.
Equation~\eqref{eq:Sht} gives
$$
\mathcal S_{h,T}(0)=T^{-1/2}\sum_{t=1}^Td_t-T^{-1/2}\sum_{t=1}^TI_t^0\{d_t+x_{t-1}\lambda(\vartheta_0x_{t-1})\}.
$$
The first term is $O_p(1)$ because $\mathbb E d_t=0$ and
$\mathbb E d_t^2<\infty$.  Moreover, H\"older's inequality and
Lemma~\ref{lem:moments}\ref{enu:zeros} imply
\[
T^{-1}\sum_{t=1}^TI_t^0u_t^2
\leq
\left(T^{-1}\sum_{t=1}^T|u_t|^{2+\delta_u}\right)^{2/(2+\delta_u)}
\left(T^{-1}N_T^0\right)^{\delta_u/(2+\delta_u)}
=o_p(1).
\]
Since $|x|\lambda(\vartheta_0x)\leq C(|x|+x^2)$, Lemma~\ref{lem:moments}\ref{enu:vtzero}
gives
\[
\sum_{t=1}^TI_t^0
\bigl|d_t+x_{t-1}\lambda(\vartheta_0x_{t-1})\bigr|=o_p(T).
\]
Consequently,
\begin{equation}
\|\mathcal S_T(0)\|=o_p(T^{1/2}).
\label{eq:qmle-score-rate}
\end{equation}

We next record the required Hessian approximation.  The preceding bounds and
Lemma~\ref{lem:moments} also yield
\begin{align}
\sum_{t=1}^TI_t^0\|z_{t-1,T}\|^2&=o_p(1),
&\sum_{t=1}^TI_t^0\|z_{t-1,T}\|&=o_p(T^{1/2}),
&\sum_{t=1}^TI_t^0\|z_{t-1,T}\||v_t|&=o_p(T^{1/2}).
\label{eq:qmle-boundary-orders}
\end{align}
Indeed, the first assertion follows from
\[
\left(\sum_{t=1}^TI_t^0\|z_{t-1,T}\|^2\right)^2
\leq N_T^0\sum_{t=1}^T\|z_{t-1,T}\|^4=o_p(1),
\]
and the remaining two follow by Cauchy--Schwarz as above.

Moreover, since
\[
\vartheta_T(h)=\vartheta_0+T^{-1/2}h
\]
and $|h|\leq \|\lpar\|\leq b_T$, for every deterministic
$b_T=o(T^{1/2})$,
\[
\sup_{\|\lpar\|\leq b_T}
|\vartheta_T(h)-\vartheta_0|
\leq T^{-1/2}b_T=o(1).
\]
Since $\vartheta_0>0$, it follows that, for all sufficiently large $T$,
\[
\frac{\vartheta_0}{2}
\leq \vartheta_T(h)
\leq \frac{3\vartheta_0}{2}
\qquad\text{uniformly over }\|\lpar\|\leq b_T.
\]
Hence $\vartheta_T(h)$ and $\vartheta_T(h)^{-1}$ are uniformly
bounded on this set. Direct substitution into
\eqref{eq:Hqq}--\eqref{eq:Hqh}, using the boundedness of $\lambda'$, as in the proof of Proposition~\ref{prop:mlescorehess}\ref{enu:hess}, implies, uniformly over $\|\lpar\|\leq b_T$,
\begin{align*}
\mathcal H_{\lcof\lcof,T}(\lpar)
  &=-\vartheta_0^2\sum_{t=1}^Tz_{t-1,T}z_{t-1,T}^{\trans}+o_p(1),\\
\mathcal H_{hh,T}(\lpar)
  &=-\vartheta_0^{-2}-T^{-1}\sum_{t=1}^Tu_t^2+o_p(1),\\
\mathcal H_{\lcof h,T}(\lpar)&=o_p(1)
\end{align*}
and
\begin{equation}
\sup_{\|\lpar\|\leq b_T}\|\mathcal H_T(\lpar)-\mathcal H_T(0)\|\inprob0,
\qquad
\mathcal H_T(0)\indist
\mathcal H\defeq-
\begin{bmatrix}
\vartheta_0^2Q_{ZZ}&0\\
0&2\sigma_0^2
\end{bmatrix},
\label{eq:qmle-hessian}
\end{equation}
where the second convergence follows from the LLN and
Lemma~\ref{lem:moments}\ref{enu:ZZ}.  The limit $\mathcal H$ is a.s.
negative definite.

Given \eqref{eq:qmle-score-rate}, one can choose a deterministic sequence
$\mathfrak{a}_T\to\infty$ such that $\mathfrak{a}_T=o(T^{1/2})$ and $\mathfrak{a}_T^{-1}\mathcal S_T(0)\inprob0$. Taylor's theorem and
\eqref{eq:qmle-hessian} then imply, uniformly on compact subsets of $\reals^{k+2}$,
\begin{equation}
\frac{\ell_T(\mathfrak{a}_T\lpar)-\ell_T(0)}{\mathfrak{a}_T^2}
=\frac{\mathcal S_T(0)^{\trans}\lpar}{\mathfrak{a}_T}
 +\frac12\lpar^{\trans}\mathcal H_T(\delta_{T,\lpar}\mathfrak{a}_T\lpar)\lpar
\indist \frac12\lpar^{\trans}\mathcal H\lpar,
\label{eq:qmle-scaled-criterion}
\end{equation}
for some $\delta_{T,\lpar}\in[0,1]$.

Finally, use the reparametrisation from the proof of
Proposition~\ref{prop:maximiser}: $\lcof=\vartheta_T^{-1}(h)\tilde{\lcof}$.
By \citet{Ols78Ecta}, the resulting Gaussian criterion
$\tilde\ell_T(\tilde{\lcof},h)$ is concave, irrespective of the true
innovation distribution.  Since $\mathfrak{a}_T/T^{1/2}\to0$, $\vartheta_T(\mathfrak{a}_Th)=\vartheta_0+\frac{\mathfrak{a}_T h}{\sqrt{T}}\to\vartheta_0$ and
\[
\mathfrak{a}_T^{-1}\lpar_T(\mathfrak{a}_T\tilde{\lcof},\mathfrak{a}_T h)=\left(\vartheta_T^{-1}(\mathfrak{a}_T h)\tilde{\lcof}^{\trans},h\right)^{\trans}
\to
M\left(\tilde{\lcof}^{\trans},h\right)^{\trans},
\qquad
M\defeq\begin{bmatrix}\vartheta_0^{-1}I_{k+1}&0\\0&1\end{bmatrix},
\]
uniformly on compacta.  Hence \eqref{eq:qmle-scaled-criterion} implies
\[
\frac{\tilde\ell_T(\mathfrak{a}_T\tilde{\lcof},\mathfrak{a}_T h)-\tilde\ell_T(0)}{\mathfrak{a}_T^2}
\indist \frac12 \left(\tilde{\lcof}^{\trans},h\right)M^{\trans}\mathcal H M \left(\tilde{\lcof}^{\trans},h\right)^{\trans}
\]
in finite-dimensional distributions.  The limit is strictly concave and uniquely maximised at zero.  If
\[
\hat{\tilde\lpar}_T\defeq
\bigl(\vartheta_T(\hat h_T)\hat\lcof_T^{\trans},\hat h_T\bigr)^{\trans},
\]
then $\hat{\tilde\lpar}_T$ maximises $\tilde\ell_T$.  Lemma~A of
\citet{Kni89CJS} therefore yields
$a_T^{-1}\hat{\tilde\lpar}_T\inprob0$.  Since $a_T/T^{1/2}\to0$,
$\vartheta_T(\hat h_T)\inprob\vartheta_0$, and hence
$a_T^{-1}\hat\lpar_T\inprob0$.  Thus $\hat\lpar_T=o_p(T^{1/2})$.
Recalling \eqref{eq:p-to-pi},
\[
\hat\alpha_T^{\mle}-\alpha_{T,0}=o_p(1),\quad
\hat\beta_T^{\mle}-\beta_{T,0}=o_p(1),\quad
\hat{\vec\phi}{}_{T}^{\mle}-\vec\phi_0=o_p(1),\quad
\hat\vartheta_T^{\mle}-\vartheta_0=o_p(1).
\]
The final assertion follows from $\hat\sigma_T^{\mle}=(\hat\vartheta_T^{\mle})^{-1}$.\hfill{}\qedsymbol{}

\subsection{QMLE asymptotics example}\label{sec:QMLE_example}

This subsection derives the asymptotic distribution of the Gaussian
QMLE for the Laplace errors with $k=1$, $y_0=0$ $\sigma_0=1$, and $a_0=c_0=0$. That is,
$$
 f_L(u)=\frac1{\sqrt2}e^{-\sqrt2|u|},\qquad u\in\mathbb R,
$$
and
$$
 \mathbb Eu_t=0,\qquad \mathbb Eu_t^2=1,\qquad
 \mathbb Eu_t^3=0,\qquad \mathbb Eu_t^4=6.
$$
Since $a_0=c_0=0$,
$$
 \alpha_{T,0}=0,\qquad \beta_{T,0}=1,
 \qquad y_t=[y_{t-1}+u_t]_+.
$$
Let
$$\Lambda(\tau)\defeq\sup_{\tau^{\prime}\leq \tau}\left[-K(\tau^{\prime})\right]_{+}, \quad K(\tau)\defeq b_0 +\sigma_0 W(\tau)\equiv W(\tau).$$
We start the analysis with the following useful lemma.
\begin{lem}\label{lem:lap00-boundary}
Suppose that $g:[0,\infty)\to\reals$ is continuous and, for some $C,m,\eta>0$,
\begin{equation}\label{eq:lap00-envelope}
 |g(x)|\leq C(1+x^m)e^{-\eta x},\qquad x\geq0.
\end{equation}
Then, in $D[0,1]$, jointly with $\tfrac1{\sqrt{T}}\sum_{t=1}^{\smlfloor{\tau T}}u_{t}\indist \sigma_0 W(\tau)$,
\begin{equation}\label{eq:lap00-boundary-limit}
 \frac1{\sqrt T}\sum_{t=1}^{\lfloor T\tau\rfloor}g(y_{t-1})
 \indist 2\nu_L(g)\Lambda(\tau),
\end{equation}
where
\begin{equation}
 \nu_L(g)=\int_0^\infty g(x)\,dx+\frac1{\sqrt{2}}g(0).
 \label{eq:lap00-nu}
\end{equation}
The conclusion holds jointly for any finite collection of
functions satisfying \eqref{eq:lap00-envelope}.
\end{lem}
\begin{proof}
Consider a Markov chain $X_t=[X_{t-1}+u_t]_+$, $X_0=0$, which has an atom at zero. For $x\geq0$, the probability of jumping to zero is $\prob(u_t\leq-x)$ and
\begin{align*}
 P(x,\{0\})&=\frac12e^{-\sqrt{2} x},\\
 P(x,dy)&=\frac1{\sqrt{2}}e^{-\sqrt{2}|y-x|}\,dy,\qquad y>0.
\end{align*}
The measure
\begin{equation*}
 \nu_L(dy)=dy+\frac1{\sqrt{2}}\delta_0(dy)
\end{equation*}
is invariant for the chain. Indeed, the mass transported to zero is
\begin{equation*}
 \frac1{\sqrt{2}}P(0,\{0\})+
 \int_0^\infty P(x,\{0\})\,dx
 =\frac1{2\sqrt{2}}+\frac1{2\sqrt{2}}=\frac1{\sqrt{2}},
\end{equation*}
and, for $y>0$, the transported density equals
\begin{equation*}
 \frac12e^{-\sqrt{2} y}
 +\int_0^\infty\frac1{\sqrt{2}}e^{-\sqrt{2}|y-x|}\,dx
 =\frac12e^{-\sqrt{2} y}
   +\frac12(1-e^{-\sqrt{2} y})+\frac12=1.
\end{equation*}
The chain is Harris recurrent and null recurrent; 
zero is a regeneration atom. The ratio ergodic theorem \citep[Theorem 17.2.1]{meyn2009markov} therefore gives, for $\nu_L$-integrable $g$ and a nonnegative $q$ satisfying $0<\nu_L(q)<\infty$,
\begin{equation} \label{eq:lap00-ratio}
 \frac{\sum_{t=1}^n g(X_{t-1})}
      {\sum_{t=1}^n q(X_{t-1})}
 \xrightarrow[n\to\infty]{a.s.} \frac{\nu_L(g)}{\nu_L(q)}.
\end{equation}

For $x\geq0$ let
$$
 q(x)=\expect[-x-u_t]_+
 =\int_{-\infty}^{-x}(-x-u)\frac1{\sqrt{2}}e^{\sqrt{2} u}\,du =\frac1{2\sqrt{2}}e^{-\sqrt{2} x}.
$$
Using \eqref{eq:lap00-nu},
$$
 \nu_L(q) =\int_0^\infty\frac1{2\sqrt{2}}e^{-\sqrt{2} x}\,dx
   +\frac1{\sqrt{2}}\cdot\frac1{2\sqrt{2}}=\frac12.
$$
Hence the ratio in \eqref{eq:lap00-ratio} equals $2\nu_L(g)$ and we are left with analyzing the denominator $\sum_{t=1}^n q(X_{t-1})$ in \eqref{eq:lap00-ratio}.

\citet[Theorem 3.1]{BD22} implies that
\begin{align*}
\frac1{\sqrt{T}}\sum\limits_{t=1}^{\lfloor T\tau\rfloor} [-y_{t-1}-u_t]_+
&=\frac1{\sqrt{T}}\sum\limits_{t=1}^{\lfloor T\tau\rfloor} (y_t-y_{t-1}-u_t)
=\frac{y_{\lfloor T\tau\rfloor}-y_0}{\sqrt{T}}-\frac1{\sqrt{T}}\sum\limits_{t=1}^{\lfloor T\tau\rfloor}u_t\\
&\indist \sup_{\tau^{\prime}\leq \tau}\left[-K(\tau^{\prime})\right]_{+}\equiv\Lambda(\tau).
\end{align*}
Thus,
$$\sup\limits_{0\leq\tau\leq1}\frac1{\sqrt{T}}\sum\limits_{t=1}^{\lfloor T\tau\rfloor}[-y_{t-1}-u_t]_+=O_p(1).$$
Write
$$
\frac1{\sqrt T}\sum_{t=1}^{\lfloor T\tau\rfloor}q(y_{t-1})
 =\frac1{\sqrt{T}}\sum\limits_{t=1}^{\lfloor T\tau\rfloor} [-y_{t-1}-u_t]_+ +
\underbrace{ \frac1{\sqrt T}\sum_{t=1}^{\lfloor T\tau\rfloor}\left(q(y_{t-1})-[-y_{t-1}-u_t]_+\right)}_{\mathfrak{M}_T(\tau)}.
$$
We want to show that $\sup\limits_{\tau}|\mathfrak{M}_T(\tau)|=o_p(1)$. To do this, note that $\mathfrak{M}_T$ is a martingale and
\begin{equation}\label{eq:mathfracM_bound}
  \expect_{t-1}[-y_{t-1}-u_t]_+^2=\int_{-\infty}^{-y_{t-1}}(-y_{t-1}-u)^2\frac1{\sqrt{2}}e^{\sqrt{2} u}\,du=\sqrt{2}q(y_{t-1})
\end{equation}
For fixed $M>0$, consider a stopping time:
\begin{equation*}
 \tau_{T,M}=
 \inf\left\{\tau\in[0,1]:\frac1{\sqrt{T}}\sum\limits_{t=1}^{\lfloor T\tau\rfloor} q(y_{t-1})>M\right\}\wedge1.
\end{equation*}
Using $q(x)\leq q(0)$ and \eqref{eq:mathfracM_bound}, the predictable quadratic variation satisfies
\begin{align*}
 \left\langle \mathfrak{M}_T(\cdot\wedge\tau_{T,M})\right\rangle(1)
 \leq\frac{\sqrt{2}}{T}
 \sum_{t\leq T\tau_{T,M}}q(y_{t-1})
\leq\frac{\sqrt{2}}{\sqrt T}
 \left(M+\frac{q(0)}{\sqrt T}\right).
\end{align*}
Lenglart's inequality therefore implies, for each fixed $M$,
\begin{equation*}
 \sup_{0\leq\tau\leq\tau_{T,M}}|\mathfrak{M}_T(\tau)|=o_p(1).
\end{equation*}
On the event $\{\tau_{T,M}<1\}$, we get
$$
\frac1{\sqrt{T}}\sum\limits_{t=1}^{\lfloor T\tau_{T,M}\rfloor} [-y_{t-1}-u_t]_+>M-|\mathfrak{M}_T(\tau_{T,M})|.
$$
Consequently,
$$
 \Prob(\tau_{T,M}<1)
 \leq \Prob\left\{\frac1{\sqrt{T}}\sum\limits_{t=1}^{T} [-y_{t-1}-u_t]_+>M/2\right\}
 +\Prob\left\{
 \sup_{\tau\leq\tau_{T,M}}|\mathfrak{M}_T(\tau)|\geq M/2
 \right\}.
$$
The second probability converges to zero for fixed $M$, while the first is
made arbitrarily small by taking $M\to\infty$, because the argument is $O_p(1)$. Hence
$$
 \sup_{0\leq\tau\leq1}|\mathfrak{M}_T(\tau)|=o_p(1),
$$
Thus,
\begin{equation}\label{eq:q_sum_indist_limit}
 \frac1{\sqrt T}\sum_{t=1}^{\lfloor T\tau\rfloor}q(y_{t-1})
 =\frac1{\sqrt{T}}\sum\limits_{t=1}^{\lfloor T\tau\rfloor}[-y_{t-1}-u_t]_+ +o_p(1) \indist\Lambda(\tau).
\end{equation}

To justify the final step uniformly in $\tau$, define
\[
 R_n:=\sum_{t=1}^n
 \bigl[g(y_{t-1})-2\nu_L(g)q(y_{t-1})\bigr],
 \qquad R_0:=0.
\]
Since $X_t=y_t$ under the present initialisation,
\eqref{eq:lap00-ratio} implies that, for every $\epsilon>0$,
there is an almost surely finite random integer $N_\epsilon$
such that
\[
 |R_n|\leq\epsilon\sum_{t=1}^n q(y_{t-1}),
 \qquad n\geq N_\epsilon.
\]
As $q$ is nonnegative, it follows that
\[
 \frac1{\sqrt T}\max_{0\leq n\leq T}|R_n|
 \leq
 \frac1{\sqrt T}\max_{0\leq n<N_\epsilon}|R_n|
 +\frac{\epsilon}{\sqrt T}\sum_{t=1}^T q(y_{t-1}).
\]
For each fixed $\epsilon$, the first term converges almost surely
to zero, while $T^{-1/2}\sum_{t=1}^Tq(y_{t-1})=O_p(1)$ by the
preceding argument. Letting first $T\to\infty$ and then
$\epsilon\to0$ therefore gives
\[
 \sup_{0\leq\tau\leq1}
 \left|
 \frac1{\sqrt T}\sum_{t=1}^{\lfloor T\tau\rfloor}g(y_{t-1})
 -\frac{2\nu_L(g)}{\sqrt T}
  \sum_{t=1}^{\lfloor T\tau\rfloor}q(y_{t-1})
 \right|=o_p(1).
\]
Together with \eqref{eq:q_sum_indist_limit}, this proves \eqref{eq:lap00-boundary-limit}.
Applying the same argument to each function in a finite
collection gives the joint convergence.
\end{proof}

To analyze the score, let $I_t^0=\indic\{y_t=0\}$ and $I_t^+=1-I_t^0$. At
$(\alpha,\beta,\vartheta)=(0,1,1)$, the score formulas from
\eqref{eq:Sqt} and \eqref{eq:Sht} reduce to
\begin{align*}
 \mathcal S_{r,T}(0)
 &=\sum_{t=1}^T
 \begin{bmatrix}T^{-1/2}\\T^{-1}y_{t-1}\end{bmatrix}
 \{u_tI_t^+-\lambda(y_{t-1})I_t^0\},\\
 \mathcal S_{h,T}(0)
 &=\frac1{\sqrt T}\sum_{t=1}^T
 \{(1-u_t^2)I_t^+-y_{t-1}\lambda(y_{t-1})I_t^0\},
\end{align*}
where $\lambda(x)=\phi(x)/(1-\Phi(x))$ is the inverse Mills ratio.
Since $I_t^0=\indic\{u_t\leq-y_{t-1}\}$, for $x\geq0$,
\begin{align*}
 \Pr(u_t\leq-x)&=\frac12e^{-\sqrt{2} x},\\
 \mathbb E[u_t\indic\{u_t\leq-x\}]
 &=-\frac12e^{-\sqrt{2} x}\left(x+\frac1{\sqrt{2}}\right),\\
 \mathbb E[u_t^2\indic\{u_t\leq-x\}]
 &=\frac12e^{-\sqrt{2} x}(x^2+\sqrt{2}x+1).
\end{align*}

Define the conditional misspecification drifts
\begin{align*}
 b_\rho(x)
 &=-\mathbb E[(u_t+\lambda(x))\mathbf1\{u_t\leq-x\}]
 =\frac12e^{-\sqrt{2} x}
   \left(x+\frac1{\sqrt{2}}-\lambda(x)\right),\\
 b_\sigma(x)
 &=-\mathbb E[(1-u_t^2+x\lambda(x))
              \mathbf1\{u_t\leq-x\}]
 =\frac12xe^{-\sqrt{2} x}\left(x+\sqrt{2}-\lambda(x)\right).
\end{align*}
Both satisfy \eqref{eq:lap00-envelope}, as do their polynomial multiples.
Set
\begin{equation*}
 A_L=\nu_L(b_\rho),\qquad D_L=\nu_L(b_\sigma).
\end{equation*}
Since $\lambda(0)=\sqrt{2/\pi}$,
\begin{align}
 A_L
 &=\frac34-\frac1{2\sqrt\pi}
   -\frac12\int_0^\infty e^{-\sqrt2x}\lambda(x)\,dx\approx-0.0003,
 \label{eq:lap00-AL}\\
 D_L
 &=\frac1{\sqrt2}
   -\frac12\int_0^\infty xe^{-\sqrt2x}\lambda(x)\,dx\approx0.23.
 \label{eq:lap00-DL}
\end{align}
Thus,
\begin{align*}
 \mathcal S_{r,T}(0)
 &=\sum_{t=1}^T
 \begin{bmatrix}T^{-1/2}\\T^{-1}y_{t-1}\end{bmatrix}u_t
 +\sum_{t=1}^T
 \begin{bmatrix}T^{-1/2}\\T^{-1}y_{t-1}\end{bmatrix}
 b_\rho(y_{t-1})+R_{\rho,T}.
\end{align*}
The conditional variance of $u_t+\lambda(y_{t-1}))I_t^0+b_\rho(y_{t-1})$ is bounded by
$C(1+y_{t-1}^2)e^{-\gamma y_{t-1}}$. Lemma~\ref{lem:lap00-boundary} then implies
\begin{align*}
 \operatorname{Var}_{T-1}(R_{\rho,T}^{(1)})
 &\leq \frac CT\sum_{t=1}^T
 (1+y_{t-1}^2)e^{-\gamma y_{t-1}}=O_p(T^{-1/2}),\\
 \operatorname{Var}_{T-1}(R_{\rho,T}^{(2)})
 &\leq \frac C{T^2}\sum_{t=1}^T
 y_{t-1}^2(1+y_{t-1}^2)e^{-\gamma y_{t-1}}
 =O_p(T^{-3/2}).
\end{align*}
Hence $R_{\rho,T}=o_p(1)$. The same lemma yields
\begin{align*}
 \frac1{\sqrt T}\sum_{t=1}^Tb_\rho(y_{t-1})
 &\indist 2A_L\sup_{\tau\leq 1}\left[-K(\tau)\right]_{+},\\
 \frac1T\sum_{t=1}^Ty_{t-1}b_\rho(y_{t-1})
 &=\frac1{\sqrt T}
   \left\{\frac1{\sqrt T}\sum_{t=1}^T
   y_{t-1}b_\rho(y_{t-1})\right\}=o_p(1).
\end{align*}
Thus, the boundary drift enters only the intercept score.

Similarly,
\begin{equation*}
 \mathcal S_{h,T}(0)
 =-\frac1{\sqrt T}\sum_{t=1}^T(u_t^2-1)
  +\frac1{\sqrt T}\sum_{t=1}^Tb_\sigma(y_{t-1})+o_p(1).
\end{equation*}
By Lemma~\ref{lem:lap00-boundary}, the second term converges to
$2D_L\sup_{\tau\leq 1}\left[-K(\tau)\right]_{+}$.

The analysis of the remaining score terms follows the derivations of limiting distribution of the martingales ${\cal M}_{\lcof,T}$
and ${\cal M}_{h,T}$ at the end of Section \ref{sec:score_appendix} with the only difference that
\[
\expect(\vartheta_{0}^{-1}-\vartheta_{0}u_{t}^{2})^{2}=\sigma_{0}^{2}\expect(1-(\sigma_{0}^{-1}u_{t})^{2})^{2}=1-2+6=5\neq2\sigma_0^2.
\]
Thus,
\begin{equation*}
 \mathcal S_T(0)\Rightarrow
 \begin{bmatrix}
 W(1)+2A_L\sup_{\tau\leq 1}\left[-K(\tau)\right]_{+}\\
 \int_{0}^{1}Y(\tau)\diff W(\tau)\\
 -\sqrt5\xi+2D_L\sup_{\tau\leq 1}\left[-K(\tau)\right]_{+}
 \end{bmatrix}.
\end{equation*}

The Hessian asymptotics from \ref{sec:Hessian_appendix} remains unchanged. Therefore,
$$
\begin{bmatrix}T^{1/2}(\hat{\alpha}_{T}^{\mle}-\alpha_{0})\\
T(\hat{\beta}_{T}^{\mle}-\beta_{0})
\end{bmatrix}\xrightarrow[T\to\infty]{d}\sigma_0\begin{bmatrix}1 & \int_{0}^{1}Y(\tau)\diff\tau\\
\int_{0}^{1}Y(\tau)\diff\tau & \int_{0}^{1}Y^{2}(\tau)\diff\tau
\end{bmatrix}^{-1}\begin{bmatrix}W(1)+2A_L\sup_{\tau\leq 1}\left[-K(\tau)\right]_{+}\\
\int_{0}^{1}Y(\tau)\diff W(\tau)
\end{bmatrix}
$$
jointly with
$$
T^{1/2}\left((\hat{\sigma}_{T}^{\mle})^{2}-1\right)\xrightarrow[T\to\infty]{d}
\sqrt{5}\xi-2D_L\sup_{\tau\leq 1}\left[-K(\tau)\right]_{+}
$$
where $\xi\sim \mathcal{N}(0,1)$ is independent of $W$ (and therefore also $Y$), and $A_L,D_L$ are defined in
\eqref{eq:lap00-AL}--\eqref{eq:lap00-DL}.

\section{Censored least absolute deviations}

\subsection{Proof of Theorem~\ref{thm:tobitLAD}}

The CLAD estimators $\hat{\rho}_{T}^{\lad}\defeq(\hat{\alpha}_{T}^{\lad},\hat{\beta}_{T}^{\lad},\hat{\vec{\phi}}{}_{T}^{\lad})$
are minimisers of
\begin{equation}
S_{T}(\rho)\defeq S_{T}(\alpha,\beta,\vec{\phi})\defeq\sum_{t=1}^{T}\smlabs{y_{t}-[\alpha+\beta y_{t-1}+\vec{\phi}^{\trans}\Delta\vec y_{t-1}]_{+}}\label{eq:ladcriterion}
\end{equation}
over $\Pi$. To economise on notation, in this section we shall generally
write $\hat{\alpha}_{T}^{\lad}$ as merely $\hat{\alpha}_{T}$, etc.
Our first result establishes that $\hat{\beta}_{T}$ concentrates
in a $O_{p}(T^{-1/2})$ neighbourhood of $\beta_{T,0}$; its proof
appears in Appendix~\ref{subsec:propproof} below.
\begin{prop}
\label{prop:consprelim}Suppose \ref{ass:INIT}--\ref{ass:JSR} and
\ref{ass:LAD} hold. Then $T^{1/2}(\hat{\beta}_{T}-\beta_{T,0})=O_{p}(1)$.
\end{prop}

In view of the preceding, it will be convenient to work with the following
reparametrisation of the model
\[
\varpi\defeq\begin{bmatrix}\alpha-\alpha_{T,0}\\
\delta\\
\vec{\phi}-\vec{\phi}_{0}
\end{bmatrix}\defeq\begin{bmatrix}\alpha-\alpha_{T,0}\\
T^{1/2}(\beta-\beta_{T,0})\\
\vec{\phi}-\vec{\phi}_{0}
\end{bmatrix}.
\]
Since Proposition~\ref{prop:consprelim} implies that for any $\epsilon>0$
we may choose $M<\infty$ such that $\liminf_{T\goesto\infty}\Prob\{\smlabs{\hat{\delta}_{T}}\leq M\}>1-\epsilon/2$,
we shall henceforth treat the parameter space $\bar{\Pi}$ for $\varpi$
as though it were compact. (Recall that the parameter spaces for $\alpha$
and $\vec{\phi}$ are compact under Assumption~\ref{ass:LAD}.) Similarly
to Appendix~\ref{sec:mlproof}, we define the local parameters
\begin{align*}
\lcof\defeq\begin{bmatrix}a\\
c\\
\vec f
\end{bmatrix} & =\begin{bmatrix}T^{1/2} & 0 & 0\\
0 & T & 0\\
0 & 0 & T^{1/2}
\end{bmatrix}\begin{bmatrix}\alpha-\alpha_{T,0}\\
\beta-\beta_{T,0}\\
\vec{\phi}-\vec{\phi}_{0}
\end{bmatrix}=T^{1/2}\begin{bmatrix}\alpha-\alpha_{T,0}\\
\delta\\
\vec{\phi}-\vec{\phi}_{0}
\end{bmatrix}=T^{1/2}\varpi,
\end{align*}
and denote the systematic part of the r.h.s.\ of the model (\ref{eq:tobitark}),
evaluated at the true parameters, as
\begin{equation}\label{eq:xtm1again}
x_{t-1}\defeq\alpha_{T,0}+\beta_{T,0}y_{t-1}+\vec{\phi}_{0}^{\trans}\Delta\vec y_{t-1}.
\end{equation}
We may then rewrite the part of the CLAD criterion function (\ref{eq:ladcriterion})
that depends on $\rho=(\alpha,\beta,\vec{\phi}^{\trans})^{\trans}$
in terms of the local parameters $\varpi$ or $\lcof$ as
\[
\alpha+\beta y_{t-1}+\vec{\phi}^{\trans}\Delta\vec y_{t-1}=x_{t-1}+\varpi^{\trans}\Z_{t-1,T}=x_{t-1}+\lcof^{\trans}z_{t-1,T},
\]
where, recalling \eqref{eq:ztm1} above,
\begin{align*}
z_{t-1,T} & =\begin{bmatrix}T^{-1/2}\\
T^{-1}y_{t-1}\\
T^{-1/2}\Delta\vec y_{t-1}
\end{bmatrix}, & \mathcal{Z}_{t-1,T} & \defeq\begin{bmatrix}1\\
T^{-1/2}y_{t-1}\\
\Delta\vec y_{t-1}
\end{bmatrix}=T^{1/2}z_{t-1,T}.
\end{align*}

To establish the rate of convergence and thence the limiting distribution
of the CLAD estimator, we will need to consider (appropriately scaled)
counterparts of the CLAD criterion in terms of both $\varpi$ and $\lcof$,
defined here (making a slight abuse of notation in reusing $S_{T}$)
as
\begin{align*}
S_{T}(\varpi) & \defeq\frac{1}{T}\sum_{t=1}^{T}\smlabs{y_{t}-[x_{t-1}+\varpi^{\trans}\Z_{t-1,T}]_{+}}, & \S_{T}(\lcof) & \defeq\sum_{t=1}^{T}\smlabs{y_{t}-[x_{t-1}+\lcof^{\trans}z_{t-1,T}]_{+}}.
\end{align*}
The first of these will be used to establish the rate of convergence, while the second will help us to derive the limiting distribution.
Because $\{y_{t}\}$ is nonstationary, the appropriate centring for
each of these functions is given not by their unconditional expectations,
but instead by
\begin{align*}
\bar{S}_{T}(\varpi) & \defeq\frac{1}{T}\sum_{t=1}^{T}\expect_{t-1}\smlabs{y_{t}-[x_{t-1}+\varpi^{\trans}\Z_{t-1,T}]_{+}}\\
\bar{\S}_{T}(\lcof) & \defeq\sum_{t=1}^{T}\expect_{t-1}\smlabs{y_{t}-[x_{t-1}+\lcof^{\trans}z_{t-1,T}]_{+}}.
\end{align*}
Define the associated recentred criterion functions
\begin{align*}
U_{T}(\varpi) & \defeq S_{T}(\varpi)-\bar{S}_{T}(\varpi), & \U_{T}(\lcof) & \defeq\S_{T}(\lcof)-\bar{\S}_{T}(\lcof),
\end{align*}
and let $\{e_{t}\}$ denote the (bounded) martingale difference sequence
defined by
\begin{equation}
e_{t}\defeq\sgn([x_{t-1}+u_{t}]_{+}-[x_{t-1}]_{+})-\expect_{t-1}\sgn([x_{t-1}+u_{t}]_{+}-[x_{t-1}]_{+}).\label{eq:emds}
\end{equation}

Our main results on the asymptotics of these functions are the following.
\begin{prop}
\label{prop:centring}Suppose \ref{ass:INIT}--\ref{ass:JSR} and
\ref{ass:LAD} hold.
\begin{enumerate}
\item \label{enu:consistency} For every $\epsilon,\delta>0$, there exists
an $\eta>0$ such that
\[
\liminf_{T\goesto\infty}\Prob\left\{ \inf_{\smlnorm{\varpi}\geq\delta}[\bar{S}_{T}(\varpi)-\bar{S}_{T}(0)]>\eta\right\} \geq1-\epsilon.
\]
\item \label{enu:centrerate}For every $\epsilon>0$, there exist $\delta,\eta>0$
such that
\begin{equation}
\liminf_{T\goesto\infty}\Prob\{\bar{S}_{T}(\varpi)-\bar{S}_{T}(0)\geq\eta\smlnorm{\varpi}^{2}\sep\forall\smlnorm{\varpi}\leq\delta\}\geq1-\epsilon.\label{eq:minorise}
\end{equation}
\item \label{enu:centrelim}Uniformly on compact subsets of $\reals^{k+1}$, for $Q_{ZZ}$ defined in Lemma~\ref{lem:moments}\ref{enu:ZZ},
\[
\bar{\S}_{T}(\lcof)-\bar{\S}_{T}(0)\indist f_{u}(0)\lcof^{\trans}Q_{ZZ}\lcof.
\]
\end{enumerate}
\end{prop}

\begin{prop}
\label{prop:mg} Suppose \ref{ass:INIT}--\ref{ass:JSR} and \ref{ass:LAD}
hold. Define $Z_{t,T}\defeq\sum_{s=0}^{t}\smlnorm{z_{s,T}}^{2}$.
\begin{enumerate}
\item \label{enu:serate}For every $\epsilon>0$, there exist $C_{\epsilon},K_{\epsilon}<\infty$
such that $\Prob\{Z_{T-1,T}>K_{\epsilon}\}<\epsilon$, and
\begin{equation}
\Prob\left\{ \sup_{\smlnorm{\varpi}\leq\delta}T^{1/2}\smlabs{U_{T}(\varpi)-U_{T}(0)}\geq\kappa\sep Z_{T-1,T}\leq K_{\epsilon}\right\} <\frac{C_{\epsilon}\delta}{\kappa}\label{eq:ratebnd}
\end{equation}
for all $\delta,\kappa>0$, for all $T$ sufficiently large.
\item \label{enu:selim}$\{\U_{T}(\lcof)\}$ is stochastically equicontinuous
on $\reals^{k+1}$.
\item \label{enu:mglimit} For each $\lcof=(a,c,\vec f^{\trans})^{\trans}\in\reals^{k+1}$,
\begin{align*}
\U_{T}(\lcof)-\U_{T}(0) & =-\lcof^{\trans}\sum_{t=1}^{T}z_{t-1,T}e_{t}+o_{p}(1)\\
 & \indist-\lcof^{\trans}\zeta=-\int_{0}^{1}[a+cY(\tau)]\diff \widetilde{W}(\tau)-\vec f^{\trans}\xi
\end{align*}
where
\[
\zeta\defeq\begin{bmatrix}\widetilde{W}(1) & \int_{0}^{1}Y(\tau)\diff \widetilde{W}(\tau) & \xi^{\trans}\end{bmatrix}^{\trans}
\]
for $(\sigma_0 W,\widetilde{W})$ a bivariate Brownian motion with covariance matrix
\[
\mathbb{V}\begin{bmatrix}\sigma_0 W(1)\\
\widetilde{W}(1)
\end{bmatrix}=\begin{bmatrix}\expect u_{t}^{2} & \expect\smlabs{u_{t}}\\
\expect\smlabs{u_{t}} & 1
\end{bmatrix},
\]
and $\xi\distrib\mathcal{N}[0,\Omega]$ independent of $(W,\widetilde{W})$.
\end{enumerate}
\end{prop}

We may now finally proceed with the proof.
\begin{proof}[Proof of Theorem~\ref{thm:tobitLAD}]
 We want to derive the limiting distribution of
\[
D_{\lcof,T}(\hat{\rho}_{T}-\rho_{T,0})=T^{1/2}\hat{\varpi}_{T}=\hat{\lcof}_{T}.
\]
By definition, $\hat{\lcof}_{T}$ minimises
\begin{equation}
\S_{T}(\lcof)-\S_{T}(0)=[\bar{\S}_{T}(\lcof)-\bar{\S}_{T}(0)]+[\U_{T}(\lcof)-\U_{T}(0)]\indist f_{u}(0)\lcof^{\trans}Q_{ZZ}\lcof-\lcof^{\trans}\zeta\label{eq:critcvg}
\end{equation}
where the convergence holds on $\ucc(\reals^{k+1})$, by Propositions~\ref{prop:centring}\ref{enu:centrelim},
\ref{prop:mg}\ref{enu:selim} and \ref{prop:mg}\ref{enu:mglimit}.
The r.h.s.\ defines a continuous function of $\lcof$, which since
$Q_{ZZ}$ is a.s.\ positive definite (by Lemma~\ref{lem:moments}\ref{enu:ZZ})
also has a unique minimum a.s. Therefore, once we have shown below
that
\begin{equation}
\hat{\lcof}_{T}=O_{p}(1)\label{eq:rtight}
\end{equation}
it will follow by the argmax continuous mapping theorem (Theorem~3.2.2
in \citealp{VVW96}) that $\hat{\lcof}_{T}$ converges in distribution
to the minimiser of the r.h.s.\ of (\ref{eq:critcvg}), i.e.
\[
\hat{\lcof}_{T}\indist\frac{1}{2f_{u}(0)}Q_{ZZ}^{-1}\zeta.
\]

It remains to show (\ref{eq:rtight}): equivalently, that $T^{1/2}\hat{\varpi}_{T}=O_{p}(1)$.
Although we cannot apply their result directly, the argument used
here closely follows the proof of Theorem~3.2.5 in \citet{VVW96}.
Let $\epsilon>0$, and note that, for $Z_{T-1,T}$ as appears in Proposition~\ref{prop:mg}\ref{enu:serate},
we may choose $K$ such that for every $L>0$,
\begin{align}
\Prob\{T^{1/2}\smlnorm{\hat{\varpi}_{T}}>L\} & \leq\Prob\{T^{1/2}\smlnorm{\hat{\varpi}_{T}}>L\sep Z_{T-1,T}\leq K\}+\Prob\{Z_{T-1,T}> K\}\nonumber \\
 & \leq\Prob\{T^{1/2}\smlnorm{\hat{\varpi}_{T}}>L\sep Z_{T-1,T}\leq K\}+\epsilon\label{eq:ratesplit}
\end{align}
for all $T$ sufficiently large. Further, by defining the `shells'
\[
\bar{\Pi}_{j,T}\defeq\{\varpi\in\bar{\Pi}\mid2^{j}<T^{1/2}\smlnorm{\varpi}\leq2^{j+1}\}
\]
for $j\in\naturals$, we obtain for any $M\in\naturals$ and $\gamma>0$
that
\begin{multline}
\Prob\{T^{1/2}\smlnorm{\hat{\varpi}_{T}}>2^{M}\sep Z_{T-1,T}\leq K\}\\
\leq\Prob\left(\Union_{\substack{j\geq M\\
2^{j}\leq T^{1/2}\gamma
}
}\{\hat{\varpi}_{T}\in\bar{\Pi}_{j,T}\sep Z_{T-1,T}\leq K\}\right)+\Prob\{\smlnorm{\hat{\varpi}_{T}}\geq\gamma\sep Z_{T-1,T}\leq K\}.\label{eq:shellbound}
\end{multline}
We will now show that the r.h.s.\ can be made arbitrarily small,
for all $T$ sufficiently large, by choice of $M$.

Regarding the final r.h.s.\ probability in (\ref{eq:shellbound}),
we note that if $\smlnorm{\hat{\varpi}_{T}}\geq\gamma$, then we must
have the first inequality in
\[
0\geq\inf_{\smlnorm{\varpi}\geq\gamma}[S_{T}(\varpi)-S_{T}(0)]\geq\inf_{\smlnorm{\varpi}\geq\gamma}[\bar{S}_{T}(\varpi)-\bar{S}_{T}(0)]+\inf_{\smlnorm{\varpi}\geq\gamma}[U_{T}(\varpi)-U_{T}(0)]
\]
where the second inequality follows from the decomposition $S_{T}=\bar{S}_{T}+U_{T}$.
This further implies the first inequality in
\[
\inf_{\smlnorm{\varpi}\geq\gamma}[\bar{S}_{T}(\varpi)-\bar{S}_{T}(0)]\leq-\inf_{\smlnorm{\varpi}\geq\gamma}[U_{T}(\varpi)-U_{T}(0)]\leq\sup_{\varpi\in\Pi}\smlabs{U_{T}(\varpi)-U_{T}(0)}
\]
whence
\begin{multline*}
\Prob\{\smlnorm{\hat{\varpi}_{T}}\geq\gamma\sep Z_{T-1,T}\leq K\}\\
\leq\Prob\left\{ \sup_{\varpi\in\Pi}\smlabs{U_{T}(\varpi)-U_{T}(0)}\geq\inf_{\smlnorm{\varpi}\geq\gamma}[\bar{S}_{T}(\varpi)-\bar{S}_{T}(0)]\sep Z_{T-1,T}\leq K\right\} .
\end{multline*}
In view of Proposition~\ref{prop:centring}\ref{enu:consistency},
we may choose $\eta>0$ such that
\[
\Prob\left\{ \inf_{\smlnorm{\varpi}\geq\gamma}[\bar{S}_{T}(\varpi)-\bar{S}_{T}(0)]>\eta\right\} >1-\epsilon/2
\]
for all $T$ sufficiently large. Thus
\begin{align*}
\Prob\{\smlnorm{\hat{\varpi}_{T}}\geq\gamma\sep Z_{T-1,T}\leq K\} & \leq\Prob\left\{ \sup_{\varpi\in\Pi}\smlabs{U_{T}(\varpi)-U_{T}(0)}\geq\eta\sep Z_{T-1,T}\leq K\right\} \\
 & \qquad\qquad\qquad\qquad\qquad+\Prob\left\{ \inf_{\smlnorm{\varpi}\geq\gamma}[\bar{S}_{T}(\varpi)-\bar{S}_{T}(0)]\leq\eta\right\} \\
 & \leq\Prob\left\{ \sup_{\varpi\in\Pi}\smlabs{U_{T}(\varpi)-U_{T}(0)}\geq\eta\sep Z_{T-1,T}\leq K\right\} +\epsilon/2
\end{align*}
for all $T$ sufficiently large. Finally, taking $\kappa=T^{1/2}\eta$
in Proposition~\ref{prop:mg}\ref{enu:serate}, and $\delta$ sufficiently
large that the $\delta$-ball centred at zero contains the whole of
$\bar{\Pi}$, we obtain that for some $C$ depending only on $\epsilon$,
\[
\Prob\left\{ \sup_{\varpi\in\Pi}\smlabs{U_{T}(\varpi)-U_{T}(0)}\geq\eta\sep Z_{T-1,T}\leq K\right\} <\frac{C\delta}{T^{1/2}\eta}\goesto0
\]
as $T\goesto\infty$. Deduce that
\begin{equation}
\limsup_{T\goesto\infty}\Prob\{\smlnorm{\hat{\varpi}_{T}}\geq\gamma\sep Z_{T-1,T}\leq K\}<\epsilon.\label{eq:consist}
\end{equation}

We next consider the events in the union on the r.h.s.\ of (\ref{eq:shellbound}).
Observe that the largest value that $\smlnorm{\varpi}$ can take if
$\varpi$ lies in $\Union_{\{j\in\naturals\mid2^{j}\leq T^{1/2}\gamma\}}\bar{\Pi}_{j,T}$
is bounded by $2\gamma$. Without disturbing the conclusion of (\ref{eq:consist}),
which holds for arbitrary $\gamma>0$, we may take $\gamma$ sufficiently
small that Proposition~\ref{prop:centring}\ref{enu:centrerate}
applies for $\delta=2\gamma$ and the $\epsilon>0$ given above; i.e.\ there
exists an $\eta>0$ (in general, different from that given previously)
such that for the event
\begin{equation}
A_{T}\defeq\{\bar{S}_{T}(\varpi)-\bar{S}_{T}(0)\geq\eta\smlnorm{\varpi}^{2}\sep\forall\smlnorm{\varpi}\leq2\gamma\}\label{eq:minorset}
\end{equation}
we have
\begin{equation}
\liminf_{T\goesto\infty}\Prob (A_{T})\geq1-\epsilon.\label{eq:minorrate}
\end{equation}
This ensures that $\bar{S}_{T}(\varpi)-\bar{S}_{T}(0)$ is minorised,
with high probability, by $\eta\smlnorm{\varpi}^{2}$ simultaneously
on each of the $\bar{\Pi}_{j,T}$ appearing on the r.h.s.\ of (\ref{eq:shellbound}).
We can accordingly bound
\begin{multline}
\Prob\left(\Union_{\substack{j\geq M\\
2^{j}\leq T^{1/2}\gamma
}
}\{\hat{\varpi}_{T}\in\bar{\Pi}_{j,T}\sep Z_{T-1,T}\leq K\}\right)\\
\leq\sum_{\substack{j\geq M\\
2^{j}\leq T^{1/2}\gamma
}
}\Prob\left(\{\hat{\varpi}_{T}\in\bar{\Pi}_{j,T}\sep Z_{T-1,T}\leq K\}\intsect A_{T}\right)+\Prob (A_{T}^{c})\label{eq:unionsplit}
\end{multline}
such that the final term is eventually bounded by $\epsilon$, while
the minorisation of $\bar{S}_{T}(\varpi)-\bar{S}_{T}(0)$ holds for
each of the events entering the sum.

Now $\hat{\varpi}_{T}\in\bar{\Pi}_{j,T}$ implies the first inequality
in
\[
0\geq\inf_{\varpi\in\bar{\Pi}_{j,T}}[S_{T}(\varpi)-S_{T}(0)]\geq\inf_{\varpi\in\bar{\Pi}_{j,T}}[\bar{S}_{T}(\varpi)-\bar{S}_{T}(0)]+\inf_{\varpi\in\bar{\Pi}_{j,T}}[U_{T}(\varpi)-U_{T}(0)]
\]
whence, in view of $\bar{\Pi}_{j,T}=\{\varpi\in\bar{\Pi}\mid2^{j}<T^{1/2}\smlnorm{\varpi}\leq2^{j+1}\}$
\[
\inf_{\varpi\in\bar{\Pi}_{j,T}}[\bar{S}_{T}(\varpi)-\bar{S}_{T}(0)]\leq\sup_{\smlnorm{\varpi}\leq T^{-1/2}2^{j+1}}\smlabs{U_{T}(\varpi)-U_{T}(0)}.
\]
Hence, in view of (\ref{eq:minorset}), and the fact that $\smlnorm{\varpi}$
is bounded below by $T^{-1/2}2^{j}$ for $\varpi\in\bar{\Pi}_{j,T}$,
\begin{align*}
 & \Prob\left(\{\hat{\varpi}_{T}\in\bar{\Pi}_{j,T}\sep Z_{T-1,T}\leq K\}\intsect A_{T}\right)\\
 & \qquad\leq\Prob\left(\left\{ \sup_{\smlnorm{\varpi}\leq T^{-1/2}2^{j+1}}\smlabs{U_{T}(\varpi)-U_{T}(0)}\geq\inf_{\varpi\in\bar{\Pi}_{j,T}}[\bar{S}_{T}(\varpi)-\bar{S}_{T}(0)]\sep Z_{T-1,T}\leq K\right\} \intsect A_{T}\right)\\
 & \qquad\leq\Prob\left\{ \sup_{\smlnorm{\varpi}\leq T^{-1/2}2^{j+1}}\smlabs{U_{T}(\varpi)-U_{T}(0)}\geq\eta T^{-1}2^{2j}\sep Z_{T-1,T}\leq K\right\}
\end{align*}
for all $j\geq M$ such that $2^{j}\leq T^{1/2}\gamma$, for all $T$
sufficiently large. Finally, we have by Proposition~\ref{prop:mg}\ref{enu:serate}
with $\kappa=\eta T^{-1/2}2^{2j}$ and $\delta=T^{-1/2}2^{j+1}$ that
for $C$ depending only on $\epsilon$,
\begin{align}
\Prob\left\{ \sup_{\smlnorm{\varpi}\leq T^{-1/2}2^{j+1}}\smlabs{U_{T}(\varpi)-U_{T}(0)}\geq\eta T^{-1}2^{2j}\sep Z_{T-1,T}\leq K\right\}  & \leq\frac{2C}{\eta}2^{-j}\label{eq:shellbnd}
\end{align}
for all $j\geq M$ such that $2^{j}\leq T^{1/2}\gamma$, for all $T$
sufficiently large. That is, the preceding bound applies simultaneously
to every summand on the r.h.s.\ of (\ref{eq:unionsplit}), for $T$
sufficiently large.

Thus it follows that, for all $T$ sufficiently large,
\begin{align*}
\Prob\{T^{1/2}\smlnorm{\hat{\varpi}_{T}}>2^{M}\} & \leq_{(1)}\Prob\{T^{1/2}\smlnorm{\hat{\varpi}_{T}}>2^{M}\sep Z_{T-1,T}\leq K\}+\epsilon\\
 & \leq_{(2)}\sum_{\substack{j\geq M\\
2^{j}\leq T^{1/2}\gamma
}
}\Prob\left(\{\hat{\varpi}_{T}\in\bar{\Pi}_{j,T}\sep Z_{T-1,T}\leq K\}\intsect A_{T}\right)+\Prob (A_{T}^{c})\\
 & \qquad\qquad\qquad\qquad+\Prob\{\smlnorm{\hat{\varpi}_{T}}\geq\gamma\sep Z_{T-1,T}\leq K\}+\epsilon\\
 & \leq_{(3)}\sum_{\substack{j\geq M\\
2^{j}\leq T^{1/2}\gamma
}
}\frac{2C}{\eta}2^{-j}+3\epsilon
\leq2^{-M}\frac{4C}{\eta}+3\epsilon
\end{align*}
where $\leq_{(1)}$ holds by (\ref{eq:ratesplit}) with $L=2^{M}$,
$\leq_{(2)}$ by (\ref{eq:shellbound}) and (\ref{eq:unionsplit}),
and $\leq_{(3)}$ by (\ref{eq:consist}), (\ref{eq:minorrate}) and
(\ref{eq:shellbnd}). Since $C$ depends on $\epsilon$ but not $M$,
the r.h.s.\ can be made arbitrarily small by suitable choice of $\epsilon>0$
and then $M\in\naturals$. Thus $T^{1/2}\hat{\varpi}_{T}=O_{p}(1)$.
\end{proof}

\subsection{Proofs of Propositions \ref{prop:consprelim}--\ref{prop:mg}}

\label{subsec:propproof}

\subsubsection{Preliminaries}

For future use, we note the following identity (see also \citealp{Herce96ET},
p.\ 150):
\begin{multline}
\smlabs{A-B}-\smlabs A=-\sgn(A)B+\indic\{A=0\}\smlabs B\\
+2(B-A)[\indic\{B>A>0\}-\indic\{B<A<0\}].\label{eq:herce}
\end{multline}
In proving Propositions~\ref{prop:centring} and \ref{prop:mg},
we will make repeated use of the following, the proof of which appears
immediately below. 
\begin{lem}
\label{lem:ladsummands}Suppose \ref{ass:INIT}--\ref{ass:JSR} and
\ref{ass:LAD} hold. Let $w_{t-1}$ be $\filt_{t-1}$-measurable,
and define
\begin{align*}
\ell_{t-1}(w_{t-1}) & \defeq[x_{t-1}+w_{t-1}]_{+}-[x_{t-1}]_{+}\\
G(b,a) & \defeq\int_{a}^{b}[F_{u}(u)-F_{u}(a)]\diff u,
\end{align*}
with the convention that $\int_{a}^{b}v(x)\diff x=-\int_{b}^{a}v(x)\diff x$
when $b<a$. Then
\begin{align}
 & \expect_{t-1}\{\smlabs{y_{t}-[x_{t-1}+w_{t-1}]_{+}}-\smlabs{y_{t}-[x_{t-1}]_{+}}\}\nonumber \\
 & \qquad=\int_{-\infty}^{\infty}\{\smlabs{u-\ell_{t-1}(w_{t-1})}-\smlabs u\}f_{u}(u-[x_{t-1}]_{-})\diff u\label{eq:expladsnd}\\
 & \qquad=2\ell_{t-1}(w_{t-1})[F_{u}(-[x_{t-1}]_{-})-F_{u}(0)]+2G\{\ell_{t-1}(w_{t-1})-[x_{t-1}]_{-},-[x_{t-1}]_{-}\}.\label{eq:laddecmp}
\end{align}
Moreover, there exists a $C<\infty$ such that
\begin{equation}
\smlabs{\ell_{t-1}(w_{t-1})[F_{u}(-[x_{t-1}]_{-})-F_{u}(0)]}\leq C\smlabs{w_{t-1}}^{2}\indic\{x_{t-1}<0\}\label{eq:firstbnd}
\end{equation}
and
\begin{equation}
\smlabs{G\{\ell_{t-1}(w_{t-1})-[x_{t-1}]_{-},-[x_{t-1}]_{-}\}-\tfrac{1}{2}f_{u}(-[x_{t-1}]_{-})\ell_{t-1}^{2}(w_{t-1})}\leq C\smlabs{w_{t-1}}^{3}.\label{eq:secondbnd}
\end{equation}
\end{lem}
\begin{proof}
Letting $(\mathrm{I})_{t-1}$ denote the l.h.s.\ of (\ref{eq:expladsnd}),
we have
\begin{align*}
(\mathrm{I})_{t-1} & =\int_{-\infty}^{\infty}\{\smlabs{[x_{t-1}+u]_{+}-[x_{t-1}+w_{t-1}]_{+}}-\smlabs{[x_{t-1}+u]_{+}-[x_{t-1}]_{+}}\}f_{u}(u)\diff u\\
 & =\int_{-x_{t-1}}^{\infty}\{\smlabs{x_{t-1}+u-[x_{t-1}+w_{t-1}]_{+}}-\smlabs{x_{t-1}+u-[x_{t-1}]_{+}}\}f_{u}(u)\diff u\\
 & \qquad\qquad+\int_{-\infty}^{-x_{t-1}}\{[x_{t-1}+w_{t-1}]_{+}-[x_{t-1}]_{+}\}f_{u}(u)\diff u\\
 & \eqdef(\mathrm{II})_{t-1}+(\mathrm{III})_{t-1}
\end{align*}
where
\begin{align*}
(\mathrm{II})_{t-1} & =\int_{-\infty}^{\infty}\{\smlabs{x_{t-1}+u-[x_{t-1}+w_{t-1}]_{+}}-\smlabs{x_{t-1}+u-[x_{t-1}]_{+}}\}f_{u}(u)\diff u\\
 & \qquad\qquad-\int_{-\infty}^{-x_{t-1}}\{\smlabs{x_{t-1}+u-[x_{t-1}+w_{t-1}]_{+}}-\smlabs{x_{t-1}+u-[x_{t-1}]_{+}}\}f_{u}(u)\diff u\\
 & \eqdef(\mathrm{IV})_{t-1}+(\mathrm{V})_{t-1}.
\end{align*}
Since $x_{t-1}+u\leq0$ for all $u$ in the range of integration in
$(\mathrm{V})_{t-1}$, we have
\begin{align*}
(\mathrm{V})_{t-1} & =-\int_{-\infty}^{-x_{t-1}}\{([x_{t-1}+w_{t-1}]_{+}-(x_{t-1}+u))-([x_{t-1}]_{+}-(x_{t-1}+u)\}f_{u}(u)\diff u\\
 & =-\int_{-\infty}^{-x_{t-1}}\{[x_{t-1}+w_{t-1}]_{+}-[x_{t-1}]_{+}\}f_{u}(u)\diff u\\
 & =-(\mathrm{III})_{t-1}.
\end{align*}
Deduce
\begin{align}
(\mathrm{I})_{t-1} & =(\mathrm{IV})_{t-1}+(\mathrm{V})_{t-1}+(\mathrm{III})_{t-1}\nonumber \\
 & =\int_{-\infty}^{\infty}\{\smlabs{x_{t-1}+u-[x_{t-1}+w_{t-1}]_{+}}-\smlabs{x_{t-1}+u-[x_{t-1}]_{+}}\}f_{u}(u)\diff u\nonumber \\
 & =\int_{-\infty}^{\infty}\{\smlabs{u+[x_{t-1}]_{-}-([x_{t-1}+w_{t-1}]_{+}-[x_{t-1}]_{+})}-\smlabs{u+[x_{t-1}]_{-}}\}f_{u}(u)\diff u\nonumber \\
 & =\int_{-\infty}^{\infty}\{\smlabs{u-\ell_{t-1}(w_{t-1})}-\smlabs u\}f_{u}(u-[x_{t-1}]_{-})\diff u\label{eq:ladpartresult}
\end{align}
as per (\ref{eq:expladsnd}), where the final equality follows by
a change of variables.

To prove (\ref{eq:laddecmp}), we first apply the identity (\ref{eq:herce})
with $A=u$ and $B=\ell_{t-1}(w_{t-1})$ to the term inside the braces
in the preceding display. Since $u=0$ has Lebesgue measure zero,
this yields
\begin{align*}
(\mathrm{I})_{t-1} & =-\int_{-\infty}^{\infty}\sgn(u)\ell_{t-1}(w_{t-1})f_{u}(u-[x_{t-1}]_{-})\diff u\\
 & \qquad+2\int_{-\infty}^{\infty}[\ell_{t-1}(w_{t-1})-u]\\
 & \qquad\qquad\quad\cdot[\indic\{\ell_{t-1}(w_{t-1})>u>0\}-\indic\{\ell_{t-1}(w_{t-1})<u<0\}]f_{u}(u-[x_{t-1}]_{-})\diff u\\
 & \eqdef(\mathrm{VI})_{t-1}+2(\mathrm{VII})_{t-1}.
\end{align*}
For $(\mathrm{VI})_{t-1}$, we have
\begin{align*}
(\mathrm{VI})_{t-1} & =-\ell_{t-1}(w_{t-1})\left[\int_{0}^{\infty}f_{u}(u-[x_{t-1}]_{-})\diff u-\int_{-\infty}^{0}f_{u}(u-[x_{t-1}]_{-})\diff u\right]\\
 & =\ell_{t-1}(w_{t-1})\{2F_{u}(-[x_{t-1}]_{-})-1\}
  =2\ell_{t-1}(w_{t-1})\{F_{u}(-[x_{t-1}]_{-})-F_{u}(0)\}
\end{align*}
since $F_{u}(0)=\tfrac{1}{2}$ by Assumption~\ref{ass:LAD}; this
gives the first r.h.s.\ term in (\ref{eq:laddecmp}). Regarding $(\mathrm{VII})_{t-1}$,
following the convention that $\int_{a}^{b}v(x)\diff x=-\int_{b}^{a}v(x)\diff x$
when $b<a$, and then making a change of variables,
\begin{align*}
(\mathrm{VII})_{t-1} & =\indic\{\ell_{t-1}(w_{t-1})>0\}\int_{0}^{\ell_{t-1}(w_{t-1})}[\ell_{t-1}(w_{t-1})-u]f_{u}(u-[x_{t-1}]_{-})\diff u\\
 & \qquad\qquad-\indic\{\ell_{t-1}(w_{t-1})<0\}\int_{\ell_{t-1}(w_{t-1})}^{0}[\ell_{t-1}(w_{t-1})-u]f_{u}(u-[x_{t-1}]_{-})\diff u\\
 & =\int_{0}^{\ell_{t-1}(w_{t-1})}[\ell_{t-1}(w_{t-1})-u]f_{u}(u-[x_{t-1}]_{-})\diff u\\
 & =\int_{-[x_{t-1}]_{-}}^{\ell_{t-1}(w_{t-1})-[x_{t-1}]_{-}}[\ell_{t-1}(w_{t-1})-[x_{t-1}]_{-}-u]f_{u}(u)\diff u.
\end{align*}
To put this into the required form, note that by integration by parts
\[
\int_{a}^{b}(b-u)f_{u}(u)\diff u=\int_{a}^{b}[F_{u}(u)-F_{u}(a)]\diff u=G(b,a).
\]
Hence
\[
2(\mathrm{VII})_{t-1}=2G\{\ell_{t-1}(w_{t-1})-[x_{t-1}]_{-},-[x_{t-1}]_{-}\}
\]
corresponds to the second r.h.s.\ term in (\ref{eq:laddecmp}).

It remains to prove (\ref{eq:firstbnd}) and (\ref{eq:secondbnd}).
Regarding (\ref{eq:firstbnd}), observe that
\[
\ell_{t-1}(w_{t-1})\{F_{u}(-[x_{t-1}]_{-})-F_{u}(0)\}=\{[x_{t-1}+w_{t-1}]_{+}-[x_{t-1}]_{+}\}\{F_{u}(-[x_{t-1}]_{-})-F_{u}(0)\}
\]
is zero if $x_{t-1}\geq0$, or if both $x_{t-1}<0$ and $x_{t-1}+w_{t-1}<0$.
Suppose therefore that $x_{t-1}<0$ and $x_{t-1}+w_{t-1}\geq0$, in
which case the preceding is equal to
\[
(x_{t-1}+w_{t-1})\{F_{u}(-x_{t-1})-F_{u}(0)\}=(x_{t-1}+w_{t-1})f_{u}(-\tilde{x}_{t-1})(-x_{t-1}),
\]
where the equality holds by the mean value theorem, for some $\tilde{x}_{t-1}\in[x_{t-1},0]$.
From $w_{t-1}\geq-x_{t-1}>0$, it follows that $\smlabs{x_{t-1}}\leq\smlabs{w_{t-1}}$.
Since $f_{u}$ is bounded by Assumption~\ref{ass:LAD}, it follows
that the preceding is bounded by $C\smlabs{w_{t-1}}^{2}$, for some
$C<\infty$. Hence
\begin{align*}
\smlabs{\ell_{t-1}(w_{t-1})\{F_{u}(-[x_{t-1}]_{-})-F_{u}(0)\}} & \leq C\smlabs{w_{t-1}}^{2}\indic\{x_{t-1}<0\sep x_{t-1}+w_{t-1}\geq0\}\\
 & \leq C\smlabs{w_{t-1}}^{2}\indic\{x_{t-1}<0\}.
\end{align*}
Regarding (\ref{eq:secondbnd}), observe that a Taylor expansion of
$b\elmap G(b,a)$ around $b=a$ yields
\[
G(b,a)=\tfrac{1}{2}f_{u}(a)(b-a)^{2}+\tfrac{1}{3!}f_{u}^{\prime}(\tilde{b})(b-a)^{3}
\]
where $\tilde{b}$ lies between $a$ and $b$; and since $f_{u}^{\prime}$
is bounded by Assumption~\ref{ass:LAD}, there exists a $C<\infty$
such that
\[
\smlabs{G(b,a)-\tfrac{1}{2}f_{u}(a)(b-a)^{2}}\leq C|b-a|^{3}.
\]
Thus (\ref{eq:secondbnd}) follows by taking $b=\ell_{t-1}(w_{t-1})-[x_{t-1}]_{-}$, $a=-[x_{t-1}]_{-}$ and noting that $|\ell_{t-1}(w_{t-1})|\leq|w_{t-1}|$.
\end{proof}

\subsubsection{Proof of Proposition~\ref{prop:consprelim}}
\label{app:consprelim}


Define
\begin{equation*}
x_{t-1}(\rho)\defeq\alpha+\beta y_{t-1}+\vec{\phi}^{\trans}\Delta\vec y_{t-1}
\end{equation*}
so that $x_{t-1}(\rho_{T,0}) = x_{t-1}$, where the latter is as in \eqref{eq:xtm1again}, and
\[
w_{t}\defeq y_{t}-[x_{t-1}(\rho_{T,0})]_{+}=[x_{t-1}(\rho_{T,0})+u_{t}]_{+}-[x_{t-1}(\rho_{T,0})]_{+},
\]
so that
\begin{align*}
y_{t}-[x_{t-1}(\rho)]_{+} & =(y_{t}-[x_{t-1}(\rho_{T,0})]_{+})-([x_{t-1}(\rho)]_{+}-[x_{t-1}(\rho_{T,0})]_{+})\\
 & =w_{t}-([x_{t-1}(\rho)]_{+}-[x_{t-1}(\rho_{T,0})]_{+})\\
 & \eqdef w_{t}-h_{t-1}(\rho,\rho_{T,0}).
\end{align*}
Then we can write the recentred and rescaled CLAD criterion function
as
\begin{align}
Q_{T}(\rho) & \defeq T^{-1}[S_{T}(\rho)-S_{T}(\rho_{0})]=T^{-1}\sum_{t=1}^{T}[\smlabs{w_{t}-h_{t-1}(\rho,\rho_{T,0})}-\smlabs{w_{t}}].\label{eq:QTlad}
\end{align}
Henceforth set $h_{t-1}(\rho)\defeq h_{t-1}(\rho,\rho_{T,0})$, for a
more compact notation.

Since $x\elmap[x]_{+}$ is Lipschitz continuous (with Lipschitz constant
equal to unity), $\smlabs{w_{t}}\leq\smlabs{u_{t}}$. Hence it is
always the case that
\[
\smlabs{w_{t}-h_{t-1}(\rho)}-\smlabs{w_{t}}\geq0-\smlabs{w_{t}}\ge-\smlabs{u_{t}},
\]
while if $\smlabs{h_{t-1}(\rho)}>3\smlabs{u_{t}}\geq3\smlabs{w_{t}}$,
then
\[
\smlabs{w_{t}-h_{t-1}(\rho)}-\smlabs{w_{t}}>2\smlabs{u_{t}}-\smlabs{w_{t}}\geq\smlabs{u_{t}}.
\]
Therefore for each summand in (\ref{eq:QTlad}),
\begin{align*}
\smlabs{w_{t}-h_{t-1}(\rho)}-\smlabs{w_{t}} & =(\smlabs{w_{t}-h_{t-1}(\rho)}-\smlabs{w_{t}})\indic\{\smlabs{h_{t-1}(\rho)}>3\smlabs{u_{t}}\}\\
 & \qquad\qquad+(\smlabs{w_{t}-h_{t-1}(\rho)}-\smlabs{w_{t}})\indic\{\smlabs{h_{t-1}(\rho)}\leq3\smlabs{u_{t}}\}\\
 & \geq\smlabs{u_{t}}\indic\{\smlabs{h_{t-1}(\rho)}>3\smlabs{u_{t}}\}-\smlabs{u_{t}}\indic\{\smlabs{h_{t-1}(\rho)}\leq3\smlabs{u_{t}}\}\\
 & =\smlabs{u_{t}}-2\smlabs{u_{t}}\indic\{\smlabs{h_{t-1}(\rho)}\leq3\smlabs{u_{t}}\}.
\end{align*}
Hence we have the lower bound
\begin{align*}
Q_{T}(\rho) & \geq T^{-1}\sum_{t=1}^{T}\smlabs{u_{t}}-2T^{-1}\sum_{t=1}^{T}\smlabs{u_{t}}\indic\{\smlabs{h_{t-1}(\rho)}\leq3\smlabs{u_{t}}\}\\
 & \geq T^{-1}\sum_{t=1}^{T}\smlabs{u_{t}}-2\left(T^{-1}\sum_{t=1}^{T}u_{t}^{2}\right)^{1/2}\left(1-T^{-1}\sum_{t=1}^{T}\indic\{\smlabs{h_{t-1}(\rho)}>3\smlabs{u_{t}}\}\right)^{1/2}\\
 & \eqdef T^{-1}\sum_{t=1}^{T}\smlabs{u_{t}}-2\left(T^{-1}\sum_{t=1}^{T}u_{t}^{2}\right)^{1/2}\left(1-K_{T}(\rho)\right)^{1/2}
\end{align*}
which depends on $\rho$ only through $K_{T}(\rho$).

We shall show below that for \emph{every} non-negative $\kappa_{T}\goesto\infty$,
\begin{equation}
\sup_{\{\rho\in\Pi\,\mid\,\smlabs{\beta-\beta_{T,0}}\geq T^{-1/2}\kappa_{T}\}}\smlabs{K_{T}(\rho)-1}\inprob0.\label{eq:indiclim}
\end{equation}
It then follows by the LLN that
\[
\Prob\left\{ \inf_{\{\rho\in\Pi\,\mid\,\smlabs{\beta-\beta_{0}}\geq T^{-1/2}\kappa_{T}\}}Q_{T}(\rho)\geq\tfrac{1}{2}\expect\smlabs{u_{1}}\right\} \goesto1.
\]
Since $Q_{T}(\rho_{T,0})=0$, it follows in particular that
\[
\Prob\{\smlabs{\hat{\beta}_{T}-\beta_{T,0}}\geq\kappa_{T}T^{-1/2}\}\leq\Prob\left\{ \inf_{\{\rho\in\Pi\,\mid\,\smlabs{\beta-\beta_{0}}\geq T^{-1/2}\kappa_{T}\}}Q_{T}(\rho)\leq0\right\} \goesto0.
\]
Deduce that $T^{1/2}(\hat{\beta}_{T}-\beta_{T,0})=o_{p}(\kappa_{T})$
for \emph{every} sequence $\kappa_{T}\goesto\infty$, whence $T^{1/2}(\hat{\beta}_{T}-\beta_{T,0})=O_{p}(1)$.

It remains to prove (\ref{eq:indiclim}). We first note that since
the parameter space $\Pi$ is compact, there exists a $C$ such that
\[
\smlabs{(\alpha-\alpha_{0})+(\vec{\phi}-\vec{\phi}_{0})^{\trans}\Delta\vec y_{t-1}}\leq C(1+\smlnorm{\Delta\vec y_{t-1}}).
\]
Hence we have the following lower bound,
\begin{align*}
\smlabs{h_{t-1}(\rho)} & =\smlabs{[x_{t-1}(\rho)]_{+}-[x_{t-1}(\rho_{T,0})]_{+}}\\
 & \geq\smlabs{x_{t-1}(\rho)-x_{t-1}(\rho_{T,0})}\indic\{x_{t-1}(\rho)>0\}\indic\{x_{t-1}(\rho_{T,0})>0\}\\
 & =\smlabs{(\beta-\beta_{T,0})y_{t-1}+(\alpha-\alpha_{0})+(\vec{\phi}-\vec{\phi}_{0})^{\trans}\Delta\vec y_{t-1}}\indic\{x_{t-1}(\rho)>0\}\indic\{x_{t-1}(\rho_{T,0})>0\}\\
 & \geq[\smlabs{\beta-\beta_{T,0}}y_{t-1}-C(1+\smlnorm{\Delta\vec y_{t-1}})]\indic\{x_{t-1}(\rho_{T,0})>0\}\\
 & \qquad\qquad\cdot\indic\{\beta y_{t-1}-C(1+\smlnorm{\Delta\vec y_{t-1}})>0\},
\end{align*}
since $y_{t-1}\geq0$, and so
\begin{align*}
\indic\{\smlabs{h_{t-1}(\rho)}>3\smlabs{u_{t}}\} & \geq\indic\{\smlabs{\beta-\beta_{T,0}}y_{t-1}-C(1+\smlnorm{\Delta\vec y_{t-1}})>3\smlabs{u_{t}}\}\\
 & \qquad\qquad\indic\{\beta y_{t-1}-C(1+\smlnorm{\Delta\vec y_{t-1}})>0\}\indic\{x_{t-1}(\rho_{0})>0\}.
\end{align*}
Observe that the r.h.s.\ is increasing in $\smlabs{\beta-\beta_{T,0}}$
and $\beta$. Suppose that $\beta\geq\beta_{T,0}+T^{-1/2}\kappa_{T}$.
Then for all such $\rho=(\alpha,\beta,\vec{\phi})$,
\begin{align}
K_{T}(\rho) & \geq T^{-1}\sum_{t=1}^{T}\indic\{T^{-1/2}y_{t-1}-\kappa_{T}^{-1}C(1+\smlnorm{\Delta\vec y_{t-1}})-3\kappa_{T}^{-1}\smlabs{u_{t}}>0\}\nonumber \\
 & \qquad\qquad\cdot\indic\{\beta_{T,0}T^{-1/2}y_{t-1}-T^{-1/2}C(1+\smlnorm{\Delta\vec y_{t-1}})>0\}\indic\{T^{-1/2}x_{t-1}(\rho_{0})>0\}\nonumber \\
 & \geq T^{-1}\sum_{t=1}^{T}g_{M}[T^{-1/2}y_{t-1}-\kappa_{T}^{-1}C(1+\smlnorm{\Delta\vec y_{t-1}})-3\kappa_{T}^{-1}\smlabs{u_{t}}]\nonumber \\
 & \qquad\qquad\cdot g_{M}[\beta_{T,0}T^{-1/2}y_{t-1}-T^{-1/2}C(1+\smlnorm{\Delta\vec y_{t-1}})]g_{M}[T^{-1/2}x_{t-1}(\rho_{0})]\nonumber \\
 & \eqdef L_{T}^{(1)}(M)\label{eq:KTbnd1}
\end{align}
where $g_{M}$ is a Lipschitz function such that $\indic\{x>0\}\geq g_{M}(x)\geq\indic\{x>M\}$.
Similarly, if $\beta\in[\tfrac{1}{2}\beta_{T,0},\beta_{T,0}-T^{-1/2}\kappa_{T}]$,
we have the lower bound
\begin{align}
K_{T}(\rho) & \geq T^{-1}\sum_{t=1}^{T}g_{M}[T^{-1/2}y_{t-1}-\kappa_{T}^{-1}C(1+\smlnorm{\Delta\vec y_{t-1}})-3\kappa_{T}^{-1}\smlabs{u_{t}}]\nonumber \\
 & \qquad\qquad\cdot g_{M}[\tfrac{1}{2}\beta_{T,0}T^{-1/2}y_{t-1}-T^{-1/2}C(1+\smlnorm{\Delta\vec y_{t-1}})]g_{M}[T^{-1/2}x_{t-1}(\rho_{0})]\nonumber \\
 & \eqdef L_{T}^{(2)}(M).\label{eq:KTbnd2}
\end{align}

Finally, suppose $\beta\leq\tfrac{1}{2}\beta_{T,0}$: we need a slightly
different lower bound for $\smlabs{h_{t-1}(\rho)}$ in this case. We
first note that if $a>0$ and $b<\tfrac{3}{4}a$, then $\smlabs{[a]_{+}-[b]_{+}}>\tfrac{1}{4}a.$
Hence
\begin{align*}
\smlabs{h_{t-1}(\rho)} & =\smlabs{[x_{t-1}(\rho_{T,0})]_{+}-[x_{t-1}(\rho)]_{+}}\\
 & \geq\tfrac{1}{4}x_{t-1}(\rho_{T,0})\indic\{x_{t-1}(\rho_{T,0})>0\}\indic\{x_{t-1}(\rho)<\tfrac{3}{4}x_{t-1}(\rho_{T,0})\}\\
 & =\tfrac{1}{4}[\alpha_{T,0}+\beta_{T,0}y_{t-1}+\vec{\phi}_{0}^{\trans}\Delta\vec y_{t-1}]\indic\{x_{t-1}(\rho_{T,0})>0\}\\
 & \qquad\qquad\indic\{(\tfrac{3}{4}\beta_{T,0}-\beta)y_{t-1}+(\tfrac{3}{4}\alpha_{T,0}-\alpha)+(\tfrac{3}{4}\vec{\phi}_{0}-\vec{\phi})^{\trans}\Delta\vec y_{t-1}>0\}\\
 & \geq\tfrac{1}{4}[\beta_{T,0}y_{t-1}+\alpha_{T,0}+\vec{\phi}_{0}^{\trans}\Delta\vec y_{t-1}]\indic\{x_{t-1}(\rho_{T,0})>0\}\\
 & \qquad\qquad\indic\{\tfrac{1}{4}\beta_{T,0}y_{t-1}-C_{1}(1+\smlnorm{\Delta\vec y_{t-1}})>0\}.
\end{align*}
where the final inequality holds since $y_{t-1}\geq0$, and $\tfrac{3}{4}\beta_{T,0}-\beta\geq\tfrac{1}{4}\beta_{T,0}$.
Therefore,
\begin{align*}
&\indic\{\smlabs{h_{t-1}(\rho)}>3\smlabs{u_{t}}\}  \geq\indic\{\beta_{T,0}T^{-1/2}y_{t-1}+T^{-1/2}\alpha_{T,0}+\vec{\phi}_{0}^{\trans}T^{-1/2}\Delta\vec y_{t-1}-12T^{-1/2}\smlabs{u_{t}}>0\}\\
 & \qquad\cdot\indic\{\tfrac{1}{4}\beta_{T,0}T^{-1/2}y_{t-1}-T^{-1/2}C_{1}(1+\smlnorm{\Delta\vec y_{t-1}})>0\}\indic\{T^{-1/2}x_{t-1}(\rho_{T,0})>0\},
\end{align*}
whence
\begin{align}
K_{T}(\rho) & \geq T^{-1}\sum_{t=1}^{T}g_{M}[\beta_{0}T^{-1/2}y_{t-1}+T^{-1/2}\alpha_{0}+\vec{\phi}_{0}^{\trans}T^{-1/2}\Delta\vec y_{t-1}-12T^{-1/2}\smlabs{u_{t}}]\nonumber \\
 & \qquad\qquad\qquad\cdot g_{M}[\tfrac{1}{4}\beta_{T,0}T^{-1/2}y_{t-1}-T^{-1/2}C_{1}(1+\smlnorm{\Delta\vec y_{t-1}})]g_{M}[T^{-1/2}x_{t-1}(\rho_{T,0})]\nonumber \\
 & \eqdef L_{T}^{(3)}(M).\label{eq:KTbnd3}
\end{align}

It follows from (\ref{eq:KTbnd1})--(\ref{eq:KTbnd3}) that
\[
1\geq K_{T}(\rho)\geq\min\{L_{T}^{(i)}(M)\}_{i=1}^{3}
\]
for all $\rho\in\Pi$ such that $\smlabs{\beta-\beta_{T,0}}\geq T^{-1/2}\kappa_{T}$.
Hence it remains to show that $L_{T}^{(i)}(M)\inprob1$ for each $i\in\{1,2,3\}$.
We give the proof only when $i=1$; the proof in the other cases is
analogous. To that end, we note that by the continuous mapping theorem
(CMT), Lemma~\ref{lem:bd2022}\ref{enu:mombnd}, and Theorem~3.2
in \citet{BD22}, the finite-dimensional distributions of the process
\begin{align*}
\xi_{T}(\tau) & \defeq g_{M}[T^{-1/2}y_{\smlfloor{\tau T}-1}-\kappa_{T}^{-1}C(1+\smlnorm{\Delta\vec y_{\smlfloor{\tau T}-1}})-3\kappa_{T}^{-1}\smlabs{u_{\smlfloor{\tau T}}}]\\
 & \qquad\qquad\cdot g_{M}[\beta_{T,0}T^{-1/2}y_{\smlfloor{\tau T}-1}-T^{-1/2}C(1+\smlnorm{\Delta\vec y_{\smlfloor{\tau T}-1}})]g_{M}[T^{-1/2}x_{\smlfloor{\tau T}-1}(\rho_{0})],
\end{align*}
for $\tau\in[0,1]$, converge weakly to those of
\[
{\cal \xi}(\tau)\defeq g_{M}^{3}[Y(\tau)],
\]
since $\beta_{T,0}\goesto1$. (This cannot be strengthened to weak
convergence in $D[0,1]$, since we may have $\kappa_{T}\goesto\infty$
arbitrarily slowly.) Since $g_{M}$ is bounded by unity, $\sup_{\tau\in[0,1]}\expect\smlabs{\xi_{T}(\tau)}\leq1$.
Finally, since $g_{M}$ is Lipschitz and bounded by unity, there exists
a constant $C_{M}$ depending only on $M$ such that
\begin{align*}
 & \sup_{\smlabs{\tau^{\prime}-\tau}\leq\epsilon}\expect\smlabs{\xi_{T}(\tau^{\prime})-\xi_{T}(\tau)}\\
 & \qquad\qquad\leq C_{M}\Biggl[\sup_{\smlabs{\tau^{\prime}-\tau}\leq\epsilon}\expect\min\{\smlabs{T^{-1/2}y_{\smlfloor{\tau^{\prime}T}-1}-T^{-1/2}y_{\smlfloor{\tau T}-1}},1\}\\
 & \qquad\qquad\qquad\qquad+\max\{\kappa_{T}^{-1},T^{-1/2}\}\max_{1\leq t\leq T}\expect\smlnorm{\Delta\vec y_{t}}+\kappa_{T}^{-1}\max_{1\leq t\leq T}\expect\smlabs{u_{t}}\Biggr]\goesto0
\end{align*}
as $T\goesto\infty$ and then $\epsilon\goesto0$, by Lemma~\ref{lem:bd2022}\ref{enu:mombnd}
and since $T^{-1/2}y_{\smlfloor{\tau T}}$is tight in $D[0,1]$ by
Theorem~3.2 in \citet{BD22}. It follows by Theorem~IX.7.1 in \citet{GS69}
that
\[
L_{T}(M)=\int_{0}^{1}\xi_{T}(\tau)\diff\tau\indist\int_{0}^{1}\xi(\tau)\diff\tau\geq\int_{0}^{1}\indic\{Y(\tau)>M\}\diff\tau\inprob1
\]
as $T\goesto\infty$ and then $M\goesto0$, by Lemma~\ref{lem:integ}\ref{enu:loctime}.\hfill\qedsymbol{}

\subsubsection{Proof of Proposition~\ref{prop:centring}}

\paragraph{\ref{enu:consistency}.}

Defining
\[
\nu(b,a)\defeq\int_{-\infty}^{\infty}\{\smlabs{u-b}-\smlabs u\}f_{u}(u-a)\diff u,
\]
we have by Lemma~\ref{lem:ladsummands}, in particular by (\ref{eq:expladsnd}),
that
\begin{align}
\bar{S}_{T}(\varpi)-\bar{S}_{T}(0) & =\frac{1}{T}\sum_{t=1}^{T}\nu(\ell_{t-1}(w_{t-1}),[x_{t-1}]_{-})\nonumber \\
 & =\frac{1}{T}\sum_{t=1}^{T}\nu(\ell_{t-1}(w_{t-1}),0)-\frac{1}{T}\sum_{t=1}^{T}\nu(\ell_{t-1}(w_{t-1}),0)\indic\{x_{t-1}<0\}\nonumber \\
 & \qquad\qquad\qquad+\frac{1}{T}\sum_{t=1}^{T}\nu(\ell_{t-1}(w_{t-1}),[x_{t-1}]_{-})\indic\{x_{t-1}<0\}\nonumber \\
 & \eqdef(\mathrm{I})_{T}-(\mathrm{II})_{T}+(\mathrm{III})_{T}\label{eq:sbarcons}
\end{align}
with $w_{t-1}=\varpi^{\trans}\Z_{t-1,T}$. Since $\smlabs{\nu(b,a)}\leq\smlabs b$
and $\smlabs{\ell_{t-1}(w_{t-1})}\leq\smlabs{w_{t-1}}$, it follows
that for all $\omega\in\bar{\Pi}$,
\begin{align}
\smlabs{(\mathrm{II})_{T}} & \leq\frac{1}{T}\sum_{t=1}^{T}\smlabs{\ell_{t-1}(\varpi^{\trans}\Z_{t-1,T})}\indic\{x_{t-1}<0\}
 \leq\frac{1}{T}\sum_{t=1}^{T}\smlabs{\varpi^{\trans}\Z_{t-1,T}}\indic\{x_{t-1}<0\}\nonumber \\
 & \leq_{(1)}C\frac{1}{T}\sum_{t=1}^{T}\smlnorm{\Z_{t-1,T}}\indic\{x_{t-1}<0\}\nonumber \\
 & \leq_{(2)}C\left(\frac{1}{T}\sum_{t=1}^{T}\smlnorm{\Z_{t-1,T}}^{2}\right)^{1/2}\left(\frac{1}{T}\sum_{t=1}^{T}\indic\{x_{t-1}<0\}\right)^{1/2}=_{(3)}o_{p}(1),\label{eq:sbar2}
\end{align}
where $\leq_{(1)}$ follows by the compactness of $\bar{\Pi}$, $\leq_{(2)}$
by the Cauchy--Schwarz inequality, and $=_{(3)}$ by Lemmas~\ref{lem:integ}\ref{enu:max}
and \ref{lem:moments}\ref{enu:ZZ}. An identical argument shows that
$(\mathrm{III})_{T}=o_{p}(1)$.

To handle $(\mathrm{I})_{T}$, we need to consider the map $b\elmap\nu(b,0)\eqdef\nu(b)$,
which is minimised at $b=\med(u_{t})=0$ by Assumption~\ref{ass:LAD}.
Since $b\elmap\smlabs{u-b}$ is differentiable at all $b\in\reals\backslash\{u\}$,
with bounded derivatives, the function $b\elmap\nu(b)$ has first
derivative
\begin{align*}
\nu^{\prime}(b) & =-\int_{-\infty}^{\infty}\sgn(u-b)f_{u}(u)\diff u=\Prob\{u_{t}<b\}-\Prob\{u_{t}>b\}\\
 & =F_{u}(b)-[1-F_{u}(b)]=2F_{u}(b)-1
\end{align*}
so that $\nu^{\prime}(b)$ is (weakly) increasing in $b$,
with $\nu^{\prime}(b)\geq0$ when $b\geq0$, and $\nu^{\prime}(b)\leq0$
when $b\leq0$. Since $f_{u}(0)>0$ by Assumption~\ref{ass:LAD},
these inequalities hold strictly for $b$ in a neighbourhood of zero,
and so it follows that for any given $\tau>0$,
\[
\nu(b)\geq\gamma_{\tau}\indic\{\smlabs b\geq\tau\}
\]
for all $b\in\reals$, with $\gamma_{\tau}\defeq\min\{\nu(\tau),\nu(-\tau)\}>0$.
Hence
\begin{equation}\label{eq:sbar1}
(\mathrm{I})_{T}=\frac{1}{T}\sum_{t=1}^{T}\nu[\ell_{t-1}(\varpi^{\trans}\Z_{t-1,T})]
\geq\gamma_{\tau}\frac{1}{T}\sum_{t=1}^{T}\indic\{\smlabs{\ell_{t-1}(\varpi^{\trans}\Z_{t-1,T})}\geq\tau\}.
\end{equation}
It remains to lower bound the sum on the r.h.s. Observe that, for
every $\tau>0$ and $w\in\reals$,
\[
\indic\{|\ell_{t-1}(w)|\geq\tau\}
\geq
\indic\{|w|\geq\tau\}-\indic\{x_{t-1}<\tau\}.
\]
Indeed, when $x_{t-1}\geq\tau$, either $x_{t-1}+w\geq0$,
in which case $\ell_{t-1}(w)=w$, or $x_{t-1}+w<0$,
in which case $|\ell_{t-1}(w)|=x_{t-1}\geq\tau$.
Consequently, uniformly over $\varpi\in\bar{\Pi}$,
\begin{equation*}
\frac1T\sum_{t=1}^T
 \indic\{|\ell_{t-1}(\varpi^{\trans}\Z_{t-1,T})|\geq\tau\}
\geq
 \frac1T\sum_{t=1}^T
 \indic\{|\varpi^{\trans}\Z_{t-1,T}|\geq\tau\}
 -B_T(\tau),
\end{equation*}
where
\[
B_T(\tau):=\frac1T\sum_{t=1}^T\indic\{x_{t-1}<\tau\}
=o_p(1)
\]
by Lemma~\ref{lem:integ}\ref{enu:integ}, applied to
$g(x)=\indic\{x<\tau\}$.

For every $\varpi\in\bar{\Pi}$ with
$\smlnorm{\varpi}\geq\delta$, H\"older's inequality gives
\begin{align*}
&\left(
 \frac1T\sum_{t=1}^T
 \indic\{|\varpi^{\trans}\Z_{t-1,T}|\geq\tau\}
 \right)^{\delta_u/(2+\delta_u)}
 \left(\frac1T\sum_{t=1}^T \smlnorm{\Z_{t-1,T}}^{2+\delta_u}\right)^{2/(2+\delta_u)}\\
&\geq
 \frac1T\sum_{t=1}^T
 \indic\{|\varpi^{\trans}\Z_{t-1,T}|\geq\tau\}
 \smlnorm{\Z_{t-1,T}}^2
 \geq
\|\varpi\|^{-2}\frac1T\sum_{t=1}^T
\indic\{|\varpi^\top\mathcal Z_{t-1,T}|\ge\tau\}
|\varpi^\top\mathcal Z_{t-1,T}|^2\\
&\geq
 \smlnorm{\varpi}^{-2}
 \left(\varpi^{\trans}Q_T\varpi-\tau^2\right)
\geq
 \lambda_{\min}(Q_T)-\delta^{-2}\tau^2,
\end{align*}
where $\lambda_{\min}(A)$ denotes the smallest eigenvalue of the
positive semi-definite matrix $A$, and
\begin{equation}\label{eq:QTlim}
Q_T:=\frac1T\sum_{t=1}^T
 \Z_{t-1,T}\Z_{t-1,T}^{\trans}
=\sum_{t=1}^T z_{t-1,T}z_{t-1,T}^{\trans}
\indist Q_{ZZ}
\end{equation}
by Lemma \ref{lem:moments}\ref{enu:ZZ}, with $Q_{ZZ}$
almost surely positive definite.

Hence
\[
\inf_{\substack{\varpi\in\bar{\Pi}\\
                 \smlnorm{\varpi}\geq\delta}}
 \frac1T\sum_{t=1}^T
 \indic\{|\varpi^{\trans}\Z_{t-1,T}|\geq\tau\}
\geq
 [\lambda_{\min}(Q_T)-\delta^{-2}\tau^2]_+^{(2+\delta_u)/\delta_u}
 \left(\frac1T\sum_{t=1}^T
 \smlnorm{\Z_{t-1,T}}^{2+\delta_u}\right)^{-2/\delta_u}.
\]
By Lemma \ref{lem:bd2022}\ref{enu:mombnd} and \citet[Theorem 3.2]{BD22},
$$\frac1T\sum_{t=1}^T \smlnorm{\Z_{t-1,T}}^{2+\delta_u}
 \leq C\left(1+T^{-1}\sum_{t=1}^{T}\smlabs{T^{-1/2}y_{t-1}}^{2+\delta_{u}}+T^{-1}\sum_{t=1}^{T}\smlnorm{\Delta\vec y_{t-1}}^{2+\delta_{u}}\right)
 =O_p(1).$$
Since $Q_{ZZ}$ is almost surely positive definite, for the given $\epsilon,\delta>0$ we may
choose $\tau,\kappa>0$ sufficiently small that
\[
\Prob\left\{
 \inf_{\substack{\varpi\in\bar{\Pi}\\
                 \smlnorm{\varpi}\geq\delta}}
 \frac1T\sum_{t=1}^T
 \indic\{|\varpi^{\trans}\Z_{t-1,T}|\geq\tau\}
 >2\kappa
\right\}>1-\epsilon/2
\]
for all sufficiently large $T$. Since $B_T(\tau)=o_p(1)$,
it follows that
\[
\liminf_{T\goesto\infty}\Prob\left\{
 \inf_{\substack{\varpi\in\bar{\Pi}\\
                 \smlnorm{\varpi}\geq\delta}}
 \frac1T\sum_{t=1}^T
 \indic\{|\ell_{t-1}(\varpi^{\trans}\Z_{t-1,T})|\geq\tau\}
 >\kappa
\right\}\geq1-\epsilon/2.
\]
Combining this with \eqref{eq:sbarcons}, \eqref{eq:sbar2}
and \eqref{eq:sbar1}, and using the uniform negligibility
of $(\mathrm{II})_T$ and $(\mathrm{III})_T$, yields
\[
\liminf_{T\goesto\infty}\Prob\left\{
 \inf_{\substack{\varpi\in\bar{\Pi}\\
                 \smlnorm{\varpi}\geq\delta}}
 [\bar S_T(\varpi)-\bar S_T(0)]
 >\tfrac12\kappa\gamma_\tau
\right\}\geq1-\epsilon.
\]
The result follows by taking
$\eta=\tfrac12\kappa\gamma_\tau>0$.

\paragraph{\ref{enu:centrerate}.}

For suitable choices of $w_{t-1}$, both $\bar{S}_{T}(\varpi)-\bar{S}_{T}(0)$
and $\bar{\S}_{T}(\lcof)-\bar{\S}_{T}(0)$ can be expressed in terms
of
\begin{align}
 & \sum_{t=1}^{T}\expect_{t-1}\{\smlabs{y_{t}-[x_{t-1}+w_{t-1}]_{+}}-\smlabs{y_{t}-[x_{t-1}]_{+}}\label{eq:ladsummands}\\
 & =2\sum_{t=1}^{T}\ell_{t-1}(w_{t-1})[F_{u}(-[x_{t-1}]_{-})-F_{u}(0)]\nonumber
 +2\sum_{t=1}^{T}G\{\ell_{t-1}(w_{t-1})-[x_{t-1}]_{-},-[x_{t-1}]_{-}\}\nonumber \\
 &\eqdef2(\mathrm{IV})_{T}+2(\mathrm{V})_{T},\label{eq:ladsumdecomp}
\end{align}
where the first equality follows by Lemma~\ref{lem:ladsummands}.
That same result also entails that
\begin{equation}
\smlabs{(\mathrm{IV})_{T}}\leq C\sum_{t=1}^{T}\smlabs{w_{t-1}}^{2}\indic\{x_{t-1}<0\}\label{eq:ladbnd1}
\end{equation}
and
\begin{equation}
\abs{(\mathrm{V})_{T}-\frac{1}{2}\sum_{t=1}^{T}f_{u}(-[x_{t-1}]_{-})\ell_{t-1}^{2}(w_{t-1})}\leq C\sum_{t=1}^{T}\smlabs{w_{t-1}}^{3}.\label{eq:ladbnd2}
\end{equation}
We now use (\ref{eq:ladsumdecomp})--(\ref{eq:ladbnd2}) to prove
parts~\ref{enu:centrerate} and \ref{enu:centrelim} of the proposition.

Noting that $\bar{S}_{T}(\varpi)-\bar{S}_{T}(0)$ is equal to (\ref{eq:ladsummands})
times $T^{-1}$, with $w_{t-1}=\varpi^{\trans}\Z_{t-1,T}$, we obtain
from (\ref{eq:ladsumdecomp})--(\ref{eq:ladbnd2}) that
\begin{align}
\bar{S}_{T}(\varpi)-\bar{S}_{T}(0) & \geq\frac{1}{T}\sum_{t=1}^{T}f_{u}(-[x_{t-1}]_{-})\ell_{t-1}^{2}(\varpi^{\trans}\Z_{t-1,T})\nonumber \\
 & \qquad\qquad-\smlnorm{\varpi}^{2}\frac{2C}{T}\sum_{t=1}^{T}\smlnorm{\Z_{t-1,T}}^{2}\indic\{x_{t-1}<0\}-\smlnorm{\varpi}^{3}\frac{2C}{T}\sum_{t=1}^{T}\smlnorm{\Z_{t-1,T}}^{3}\nonumber \\
 & \eqdef(\mathrm{VI})_{T}-(\mathrm{VII})_{T}-(\mathrm{VII     I})_{T}.\label{eq:Sbar}
\end{align}
To determine the limit of $(\mathrm{VI})_{T}$, note $\max_{1\leq t\leq T}\smlnorm{\Z_{t-1,T}}=O_{p}(T^{1/(2+\delta_{u})})$
by Lemma~\ref{lem:moments}\ref{enu:max}, and so for $\kappa\defeq1/(2+\delta_{u}/2)<1/2$
is bounded by $T^{\kappa}$, w.p.a.1. Define
\begin{align*}
\indic_{t-1} & \defeq\indic\{x_{t-1}>T^{\kappa}\sep\smlnorm{\Z_{t-1,T}}\leq T^{\kappa}\}, & \indic_{t-1}^{c} & \defeq1-\indic_{t-1}
\end{align*}
and observe that when $\indic_{t-1}=1$, we must have $x_{t-1}>0$
and
\[
x_{t-1}+\varpi^{\trans}\Z_{t-1,T}\geq x_{t-1}-\smlnorm{\varpi}\smlnorm{\Z_{t-1,T}}>T^{\kappa}(1-\smlnorm{\varpi}),
\]
and so the preceding must be non-negative if $\smlnorm{\varpi}\leq1$.
Deduce that for all such $\varpi$,
\begin{align*}
f_{u}(-[x_{t-1}]_{-})\ell_{t-1}^{2}(\varpi^{\trans}\Z_{t-1,T})\indic_{t-1} & =f_{u}(-[x_{t-1}]_{-})\{[x_{t-1}+\varpi^{\trans}\Z_{t-1,T}]_{+}-[x_{t-1}]_{+}\}^{2}\indic_{t-1}\\
 & =f_{u}(0)\{(x_{t-1}+\varpi^{\trans}\Z_{t-1,T})-x_{t-1}\}^{2}\indic_{t-1}\\
 & =f_{u}(0)(\varpi^{\trans}\Z_{t-1,T})^{2}\indic_{t-1}.
\end{align*}
Hence
\begin{align*}
(\mathrm{VI})_{T} & =f_{u}(0)\frac{1}{T}\sum_{t=1}^{T}(\varpi^{\trans}\Z_{t-1,T})^{2}+\Biggl[-\frac{1}{T}\sum_{t=1}^{T}f_{u}(0)(\varpi^{\trans}\Z_{t-1,T})^{2}\indic_{t-1}^{c}\\
 & \qquad\qquad+\frac{1}{T}\sum_{t=1}^{T}f_{u}(-[x_{t-1}]_{-})\ell_{t-1}^{2}(\varpi^{\trans}\Z_{t-1,T})\indic_{t-1}^{c}\Biggr] \eqdef(\mathrm{IX})_{T}+(\mathrm{X})_{T}.
\end{align*}
Since $f_{u}$ is bounded, and
\begin{align*}
\ell_{t-1}^{2}(\varpi^{\trans}\Z_{t-1,T}) & =\smlabs{[x_{t-1}+\varpi^{\trans}\Z_{t-1,T}]_{+}-[x_{t-1}]_{+}}^{2}
  \leq\smlnorm{\varpi}^{2}\smlnorm{\Z_{t-1,T}}^{2}=\smlnorm{\varpi}^{2}T\smlnorm{z_{t-1,T}}^{2}
\end{align*}
we may bound
\begin{align*}
\smlabs{(\mathrm{X})_{T}} & \leq C\smlnorm{\varpi}^{2}\sum_{t=1}^{T}\smlnorm{z_{t-1,T}}^{2}\indic_{t-1}^{c}\\
 & \leq C\smlnorm{\varpi}^{2}\left(\sum_{t=1}^{T}\smlnorm{z_{t-1,T}}^{4}\right)^{1/2}\left(\sum_{t=1}^{T}\indic_{t-1}^{c}\right)^{1/2}=\smlnorm{\varpi}^{2}o_{p}(1),
\end{align*}
by the CS inequality, where the final equality holds by Lemma~\ref{lem:moments}\ref{enu:fourthz},
and
\begin{align}
T^{-1}\sum_{t=1}^{T}\indic_{t-1}^{c} & \leq T^{-1}\sum_{t=1}^{T}\indic\{T^{-1/2}x_{t-1}\leq T^{\kappa-1/2}\}+T^{-1}\sum_{t=1}^{T}\indic\{\smlnorm{\Z_{t-1,T}}>T^{\kappa}\}\nonumber \\
 & \leq_{(1)}T^{-1}\sum_{t=1}^{T}\indic\{T^{-1/2}x_{t-1}\leq\gamma\}+\indic\left\{ \max_{1\leq t\leq T}\smlnorm{\Z_{t-1,T}}>T^{\kappa}\right\} \inprob_{(2)}0\label{eq:indcomp}
\end{align}
where $\leq_{(1)}$ holds for any fixed $\gamma>0$, for all $T$
sufficiently large, and $\inprob_{(2)}$ holds as $T\goesto\infty$
and then $\gamma\goesto0$, by Lemma~\ref{lem:integ}\ref{enu:indicbnd}
and the noted stochastic order of $\max_{1\leq t\leq T}\smlnorm{\Z_{t-1,T}}$.
Recalling the definition of $Q_{T}$ in \eqref{eq:QTlim} above, we
thus have
\begin{equation}
(\mathrm{VI})_{T}\geq(\mathrm{IX})_{T}-\smlnorm{\varpi}^{2}o_{p}(1)=f_{u}(0)\varpi^{\trans}[Q_{T}+o_{p}(1)]\varpi\label{eq:iiibnd}
\end{equation}

Returning now to (\ref{eq:Sbar}), we see that by the same argument
that was applied to $(\mathrm{X})_{T}$, with Lemma~\ref{lem:integ}\ref{enu:indicbnd}
playing the role of (\ref{eq:indcomp}), we obtain
\begin{equation}
(\mathrm{VII})_{T}=C\smlnorm{\varpi}^{2}o_{p}(1).\label{eq:ivlim}
\end{equation}
Regarding $(\mathrm{VIII})_{T}$, it follows from Lemma~\ref{lem:bd2022}\ref{enu:mombnd}
(since $2+\delta_{u}>3$ under Assumption~\ref{ass:LAD}), and Theorem~3.2
in \citet{BD22}, that
\begin{equation}
\smlabs{(\mathrm{VIII})_{T}}=\smlnorm{\varpi}^{3}\frac{C}{T}\sum_{t=1}^{T}\smlnorm{\Z_{t-1,T}}^{3}=\smlnorm{\varpi}^{3}O_{p}(1).\label{eq:vbnd}
\end{equation}
It follows from (\ref{eq:Sbar}), (\ref{eq:iiibnd}), (\ref{eq:ivlim})
and (\ref{eq:vbnd}) that there exist random sequences $\xi_{1,T}=o_{p}(1)$
and $\xi_{2,T}=O_{p}(1)$ such that
\begin{align*}
\bar{S}_{T}(\varpi)-\bar{S}_{T}(0) & \geq f_{u}(0)\varpi^{\trans}[Q_{T}+\xi_{1,T}]\varpi+\xi_{2,T}\smlnorm{\varpi}^{3}\\
 & \geq f_{u}(0)[\lambda_{\min}(Q_{T})+\xi_{1,T}]\smlnorm{\varpi}^{2}+\xi_{2,T}\smlnorm{\varpi}^{3}
\end{align*}
for all $\varpi\in\bar{\Pi}$. Now let $\epsilon>0$ be as in the
statement of part~\ref{enu:centrerate} of the proposition. Since
$Q_{T}\indist Q_{ZZ}$, which is a.s.\ positive definite, we may
choose $\eta>0$ sufficiently small that
\[
\Prob\{f_{u}(0)[\lambda_{\min}(Q_{T})+\xi_{1,T}]<2\eta\}<\epsilon/2
\]
for all $T$ sufficiently large. Moreover, we may take $\delta>0$
such that
\[
\Prob\{\smlnorm{\varpi}^{3}\xi_{2,T}<-\eta\smlnorm{\varpi}^{2}\sep\forall\smlnorm{\varpi}\leq\delta\}\leq\Prob\{\smlabs{\xi_{2,T}}>\eta\delta^{-1}\}<\epsilon/2,
\]
for all $T$ sufficiently large, and therefore
\begin{align*}
 & \Prob\{\bar{S}_{T}(\varpi)-\bar{S}_{T}(0)\geq\eta\smlnorm{\varpi}^{2}\sep\forall\smlnorm{\varpi}\leq\delta\}\\
 & \qquad\qquad\qquad\geq\Prob\{f_{u}(0)[\lambda_{\min}(Q_{T})+\xi_{1,T}]\smlnorm{\varpi}^{2}\geq2\eta\smlnorm{\varpi}^{2}\sep\forall\smlnorm{\varpi}\leq\delta\}\\
 & \qquad\qquad\qquad\qquad\qquad\qquad\qquad\qquad\qquad-\Prob\{\smlnorm{\varpi}^{3}\xi_{2,T}<-\eta\smlnorm{\varpi}^{2}\sep\forall\smlnorm{\varpi}\leq\delta\}\\
 & \qquad\qquad\qquad\geq\Prob\{f_{u}(0)[\lambda_{\min}(Q_{T})+\xi_{1,T}]\geq2\eta\}-\Prob\{\smlabs{\xi_{2,T}}>\eta\delta^{-1}\}>1-\epsilon
\end{align*}
for all $T$ sufficiently large.

\paragraph{\ref{enu:centrelim}.}

Noting in this case that $\bar{\S}_{T}(\lcof)-\bar{\S}_{T}(0)$ equals
(\ref{eq:ladsummands}) with $w_{t-1}=\lcof^{\trans}z_{t-1,T}$, and
recalling that
\begin{equation}
\smlabs{\ell_{t-1}(\lcof^{\trans}z_{t-1,T})}=\smlabs{[x_{t-1}+\lcof^{\trans}z_{t-1,T}]_{+}-[x_{t-1}]_{+}}\leq\smlabs{\lcof^{\trans}z_{t-1,T}}\label{eq:lipbnd}
\end{equation}
we obtain from (\ref{eq:ladsumdecomp})--(\ref{eq:ladbnd2}) that
for all $\lcof\in\reals^{k+1}$,
\begin{align}
 & \abs{[\bar{\S}_{T}(\lcof)-\bar{\S}_{T}(0)]-\sum_{t=1}^{T}f_{u}(-[x_{t-1}]_{-})\ell_{t-1}^{2}(\lcof^{\trans}z_{t-1,T})}\label{eq:Sdiff}\\
 & \leq C\sum_{t=1}^{T}\smlabs{\lcof^{\trans}z_{t-1,T}}^{2}\indic\{x_{t-1}<0\}+C\sum_{t=1}^{T}\smlabs{\lcof^{\trans}z_{t-1,T}}^{3}\nonumber \\
 & \leq C\smlnorm{\lcof}^{2}\left(\sum_{t=1}^{T}\smlnorm{z_{t-1,T}}^{4}\right)^{1/2}\left(\sum_{t=1}^{T}\indic\{x_{t-1}<0\}\right)^{1/2}
 +C\smlnorm{\lcof}^{3}\max_{1\leq s\leq T}\smlnorm{z_{s-1,T}}\sum_{t=1}^{T}\smlnorm{z_{t-1,T}}^{2}\nonumber \\
 &=o_{p}(1)\nonumber
\end{align}
uniformly on compact subsets of $\reals^{k+1}$, by the CS inequality,
Theorem~3.2 of \citet{BD22} and Lemmas~\ref{lem:integ}\ref{enu:indicbnd}, \ref{lem:moments}\ref{enu:ZZ}, \ref{lem:moments}\ref{enu:max}
and \ref{lem:moments}\ref{enu:fourthz}.

It remains to consider the l.h.s.\ of (\ref{eq:Sdiff}). Let
\begin{align}
\indic_{t-1}(\lcof) & \defeq\indic\{x_{t-1}\geq0\sep x_{t-1}+\lcof^{\trans}z_{t-1,T}\geq0\}, & \indic_{t-1}^{c}(\lcof) & \defeq1-\indic_{t-1}(\lcof).\label{eq:bothpos}
\end{align}
We note that
\begin{align}
\indic_{t-1}^{c}(\lcof) & =\indic\{x_{t-1}<0\}+\indic\{x_{t-1}\geq0\}\indic\{x_{t-1}+\lcof^{\trans}z_{t-1,T}<0\}\nonumber \\
 & =\indic\{x_{t-1}<0\}+\indic\{\lcof^{\trans}z_{t-1,T}<-x_{t-1}\leq0\}\nonumber \\
 & \leq\indic\{x_{t-1}<0\}+\indic\{\smlabs{x_{t-1}}\leq\smlnorm{\lcof^{\trans}z_{t-1,T}}\}\nonumber \\
 & \leq\indic\{x_{t-1}<0\}+\indic\{\smlabs{x_{t-1}}\leq1\}+\indic\{1\leq\smlnorm{\lcof^{\trans}z_{t-1,T}}\}\nonumber \\
 & \leq2\cdot\indic\{x_{t-1}\leq1\}+\indic\left\{ 1\leq\smlnorm{\lcof}\sup_{1\leq s\leq T}\smlnorm{z_{s-1,T}}\right\} \eqdef2\cdot\indic_{1,t-1}^{c}+\indic_{2}^{c}(\lcof),\label{eq:compldecomp}
\end{align}
where
\begin{equation}
\sum_{t=1}^{T}\indic_{1,t-1}^{c}=\sum_{t=1}^{T}\indic\{x_{t-1}\leq1\}=o_{p}(T)\label{eq:1cnegl}
\end{equation}
and, since $\sup_{1\leq s\leq T}\smlnorm{z_{s-1,T}}=o_{p}(1)$, $\indic_{2}^{c}(\lcof)=o_{p}(1)$
uniformly on compact subsets of $\reals^{k+1}$. Now decompose
\begin{align}
 & \sum_{t=1}^{T}f_{u}(-[x_{t-1}]_{-})\ell_{t-1}^{2}(\lcof^{\trans}z_{t-1,T})\nonumber \\
 & =\sum_{t=1}^{T}f_{u}(0)\ell_{t-1}^{2}(\lcof^{\trans}z_{t-1,T})\indic_{t-1}(\lcof)
 +\sum_{t=1}^{T}f_{u}(-[x_{t-1}]_{-})\ell_{t-1}^{2}(\lcof^{\trans}z_{t-1,T})\indic_{t-1}^{c}(\lcof)\nonumber \\
 & =\sum_{t=1}^{T}f_{u}(0)(\lcof^{\trans}z_{t-1,T})^{2}-\sum_{t=1}^{T}f_{u}(0)(\lcof^{\trans}z_{t-1,T})^{2}\indic_{t-1}^{c}(\lcof)\nonumber \\
 & \qquad\qquad\qquad\qquad\qquad\qquad+\sum_{t=1}^{T}f_{u}(-[x_{t-1}]_{-})\ell_{t-1}^{2}(\lcof^{\trans}z_{t-1,T})\indic_{t-1}^{c}(\lcof)\nonumber \\
 & \eqdef(\mathrm{XI})_{T}+(\mathrm{XII})_{T}+(\mathrm{XIII})_{T}\label{eq:quaddecmp}
\end{align}
where in obtaining the second equality, we have used the fact that
\begin{align*}
\ell_{t-1}^{2}(\lcof^{\trans}z_{t-1,T})\indic_{t-1}(\lcof) & =\{[x_{t-1}+\lcof^{\trans}z_{t-1,T}]_{+}-[x_{t-1}]_{+}\}^{2}\indic\{x_{t-1}\geq0\sep x_{t-1}+\lcof^{\trans}z_{t-1,T}\geq0\}\\
 & =(\lcof^{\trans}z_{t-1,T})^{2}\indic_{t-1}(\lcof)
\end{align*}
and similarly $f_{u}(-[x_{t-1}]_{-})\indic_{t-1}(\lcof)=f_{u}(0)\indic_{t-1}(\lcof)$.

Regarding the final r.h.s.\ term in (\ref{eq:quaddecmp}), in view
of (\ref{eq:lipbnd}) and the boundedness of $f_{u}$, there exists
a $C<\infty$ such that
\begin{align}
\smlabs{(\mathrm{XIII})_{T}} & \leq C\smlnorm{\lcof}^{2}\sum_{t=1}^{T}\smlnorm{z_{t-1,T}}^{2}\{2\cdot\indic_{1,t-1}^{c}+\indic_{2}^{c}(\lcof)\}\nonumber \\
 & \leq_{(1)}C\smlnorm{\lcof}^{2}\left[2\left(\sum_{t=1}^{T}\smlnorm{z_{t-1,T}}^{4}\right)^{1/2}\left(\sum_{t=1}^{T}\indic_{1,t-1}^{c}\right)^{1/2}+\indic_{2}^{c}(\lcof)\sum_{t=1}^{T}\smlnorm{z_{t-1,T}}^{2}\right]\nonumber \\
 & =_{(2)}o_{p}(1)\label{eq:ixnegl}
\end{align}
uniformly on compact subsets of $\reals^{k+1}$, where $\leq_{(1)}$
holds by the CS inequality, and $=_{(2)}$ by Lemmas~\ref{lem:moments}\ref{enu:ZZ}
and \ref{lem:moments}\ref{enu:fourthz}, (\ref{eq:1cnegl}), and
the noted asymptotic negligibility (on compacta) of $\indic_{2}^{c}(\lcof)$.
An identical argument yields that $(\mathrm{XII})_{T}=o_{p}(1)$,
uniformly on compact subsets of $\reals^{k+1}$. This leaves the first
r.h.s.\ term in (\ref{eq:quaddecmp}), for which
\[
(\mathrm{XI})_{T}=f_{u}(0)\lcof^{\trans}\sum_{t=1}^{T}z_{t-1,T}z_{t-1,T}^{\trans}\lcof\indist f_{u}(0)\lcof^{\trans}Q_{ZZ}\lcof
\]
uniformly on compact subsets of $\reals^{k+1}$, by Lemma~\ref{lem:moments}\ref{enu:ZZ},
as required.\hfill\qedsymbol{}

\subsubsection{Proof of Proposition~\ref{prop:mg}}

Define
\begin{equation}
g_{t}(w_{t-1})\defeq\smlabs{y_{t}-[x_{t-1}+w_{t-1}]_{+}}-\smlabs{y_{t}-[x_{t-1}]_{+}},\label{eq:gt}
\end{equation}
for which
\begin{equation}
\smlabs{g_{t}(w_{t-1})-g_{t}(\tilde{w}_{t-1})}\leq\smlabs{w_{t-1}-\tilde{w}_{t-1}}.\label{eq:glip}
\end{equation}
For suitable choices of $w_{t-1}$, both $U_{T}(\varpi)-U_{T}(0)$
and $\U_{T}(\lcof)-\U_{T}(0)$ can be expressed in terms of
\[
\sum_{t=1}^{T}[g_{t}(w_{t-1})-\expect_{t-1}g_{t}(w_{t-1})].
\]
Since $Z_{T-1,T}=\sum_{s=1}^{T}\smlnorm{z_{s-1,T}}^{2}=O_{p}(1)$
by Lemma~\ref{lem:moments}\ref{enu:ZZ}, for $\epsilon>0$ as given
in the statement of part~\ref{enu:serate} of the lemma, we may choose
$K_{\epsilon}$ sufficiently large that
\[
\limsup_{T\goesto\infty}\Prob\{Z_{T-1,T}\geq K_{\epsilon}\}<\epsilon.
\]
For each $T\in\naturals$, define
\begin{equation}
\varsigma_{T}\defeq\inf\{t\in\naturals\mid Z_{t,T}>K_{\epsilon}\}.\label{eq:stopping}
\end{equation}
Then for each $t\in\naturals$,
\[
\{\varsigma_{T}>t\}=\{Z_{t,T}\leq K_{\epsilon}\}\in\filt_{t},
\]
and so $\varsigma_{T}$ is a stopping time with respect to $\{\filt_{t}\}$.
By construction
\begin{align}
\sum_{s=1}^{T\pmin\varsigma_{T}}\smlnorm{z_{s-1,T}}^{2}=Z_{T\pmin\varsigma_{T}-1,T} & =\indic\{\varsigma_{T}\geq T\}Z_{T-1,T}+\indic\{\varsigma_{T}<T\}Z_{\varsigma_{T}-1,T}\leq K_{\epsilon}.\label{eq:stopbbnd}
\end{align}

\paragraph{\ref{enu:serate}.}
Define
\begin{equation}
U_{t,T}(\varpi)\defeq\frac{1}{T}\sum_{s=1}^{t}\{\smlabs{y_{s}-[x_{s-1}+\varpi^{\trans}\Z_{s-1,T}]_{+}}-\expect_{s-1}\smlabs{y_{s}-[x_{s-1}+\varpi^{\trans}\Z_{s-1,T}]_{+}}\}\label{eq:WtT}
\end{equation}
so that $U_{T}(\varpi)=U_{T,T}(\varpi)$. Consider the martingale
array $\{M_{t,T}(\varpi),\filt_{t}\}_{t=1}^{T}$ defined by
\begin{align*}
M_{t,T}(\varpi) & \defeq T^{1/2}[U_{t,T}(\varpi)-U_{t,T}(0)] =\frac{1}{T^{1/2}}\sum_{s=1}^{t}[g_{s}(\varpi^{\trans}\Z_{s-1,T})-\expect_{s-1}g_{s}(\varpi^{\trans}\Z_{s-1,T})].
\end{align*}
Define the stopped process $N_{t,T}(\varpi)\defeq M_{t\pmin\varsigma_{T},T}(\varpi)$,
and the associated filtration $\filtg_{t}\defeq\filt_{t\pmin\varsigma_{T}}$,
for $t\in\{1,\ldots,T\}$. Observe that since
\[
Z_{T-1,T}\leq K_{\epsilon}\implies\varsigma_{T}\geq T\implies N_{T,T}(\varpi)=M_{T,T}(\varpi)\sep\forall\varpi\in\bar{\Pi},
\]
the event on the l.h.s.\ of (\ref{eq:ratebnd}) may be written as
\begin{align*}
\left\{ \sup_{\smlnorm{\varpi}\leq\delta}\smlabs{M_{T,T}(\varpi)}\geq\kappa\sep Z_{T-1,T}\leq K_{\epsilon}\right\}  & =\left\{ \sup_{\smlnorm{\varpi}\leq\delta}\smlabs{N_{T,T}(\varpi)}\geq\kappa\sep Z_{T-1,T}\leq K_{\epsilon}\right\} \\
 & \subseteq\left\{ \sup_{\smlnorm{\varpi}\leq\delta}\smlabs{N_{T,T}(\varpi)}\geq\kappa\right\}.
\end{align*}
It remains to control the probability of the r.h.s.

By the stopping time lemma (in particular, the corollary to Theorem~9.3.4
in \citealp{Chung01}), $\{N_{t,T}(\varpi),\filtg_{t}\}_{t=1}^{T}$
is a martingale array for every $\varpi\in\bar{\Pi}$, with increments
\begin{align}
\Delta N_{t,T}(\varpi) & \defeq N_{t,T}(\varpi)-N_{t-1,T}(\varpi)
 =M_{t\pmin\varsigma_{T},T}(\varpi)-M_{(t-1)\pmin\varsigma_{T},T}(\varpi)\nonumber \\
 & =T^{-1/2}[g_{t}(\varpi^{\trans}\Z_{t-1,T})-\expect_{t-1}g_{t}(\varpi^{\trans}\Z_{t-1,T})]\indic\{\varsigma_{T}>t-1\}.\label{eq:Nincr}
\end{align}
In view of (\ref{eq:glip}) and the $\filt_{t-1}$-measurability of
$\Z_{t-1,T}$, we have that the martingale sum of squares of $N_{t,T}(\varpi)-N_{t,T}(\tilde{\varpi})$,
at $t=T$, satisfies
\begin{align}
 & \sum_{t=1}^{T}[\Delta N_{t,T}(\varpi)-\Delta N_{t,T}(\tilde{\varpi})]^{2}\nonumber \\
 & \qquad\qquad=T^{-1}\sum_{t=1}^{T}\{[g_{t}(\varpi^{\trans}\Z_{t-1,T})-g_{t}(\tilde{\varpi}^{\trans}\Z_{t-1,T})]\nonumber \\
 & \qquad\qquad\qquad\qquad\qquad\qquad-\expect_{t-1}[g_{t}(\varpi^{\trans}\Z_{t-1,T})-g_{t}(\tilde{\varpi}^{\trans}\Z_{t-1,T})]\}^{2}\indic\{\varsigma_{T}>t-1\}\nonumber \\
 & \qquad\qquad\leq4\smlnorm{\varpi-\tilde{\varpi}}^{2}T^{-1}\sum_{t=1}^{T}\smlnorm{\Z_{t-1,T}}^{2}\indic\{\varsigma_{T}\geq t\}
 =4\smlnorm{\varpi-\tilde{\varpi}}^{2}Z_{T\pmin\varsigma_{T}-1,T}\nonumber \\
 & \qquad\qquad\leq4\smlnorm{\varpi-\tilde{\varpi}}^{2}K_{\epsilon},\label{eq:Nssq}
\end{align}
where the final equality holds by (\ref{eq:stopbbnd}).

Therefore by Burkholder's inequality (Theorem~2.10 in \citealp{HH80}),
we have for any $\varpi,\tilde{\varpi}\in\bar{\Pi}$ and $j\in\naturals$
that
\begin{align*}
\expect\smlabs{N_{T,T}(\varpi)-N_{T,T}(\tilde{\varpi})}^{2j} & \leq C\expect\abs{\sum_{t=1}^{T}[\Delta N_{t,T}(\varpi)-\Delta N_{t,T}(\tilde{\varpi})]^{2}}^{j}\\
 & \leq4^{j}CK_{\epsilon}^{j}\smlnorm{\varpi-\tilde{\varpi}}^{2j}
\end{align*}
whence
\[
\smlnorm{N_{T,T}(\varpi)-N_{T,T}(\tilde{\varpi})}_{2j}\leq(4C^{1/j}K_{\epsilon})^{1/2}\smlnorm{\varpi-\tilde{\varpi}}.
\]

Now fix $j\in\naturals$ such that $2j>k+1$. Observe that $N_{T,T}(0)=0$
by construction. By Corollary~2.2.5 of \citet{VVW96}, we therefore
have
\begin{align*}
\expect\sup_{\smlnorm{\varpi}\leq\delta}\smlabs{N_{T,T}(\varpi)}\leq\norm{\sup_{\smlnorm{\varpi}\leq\delta}\smlabs{N_{T,T}(\varpi)}}_{2j} & \leq K_{\epsilon}^{\prime}\int_{0}^{2\delta}(2\delta/u)^{(k+1)/2j}\diff u\\
 & =\delta\cdot2K_{\epsilon}^{\prime}\int_{0}^{1}u^{-(k+1)/2j}\diff u\eqdef C_{\epsilon}\delta
\end{align*}
where $K_{\epsilon}^{\prime}$ and $C_{\epsilon}$ depend only on
$K_{\epsilon}$ (with $k$ and $j$ fixed), noting that the integral
is convergent since $(k+1)/2j<1$ by choice of $j$. Hence by Chebyshev's
inequality
\[
\Prob\left\{ \sup_{\smlnorm{\varpi}\leq\delta}\smlabs{N_{T,T}(\varpi)}\geq\kappa\right\} \leq\frac{\expect\sup_{\smlnorm{\varpi}\leq\delta}\smlabs{N_{T,T}(\varpi)}}{\kappa}\leq\frac{C_{\epsilon}\delta}{\kappa}.
\]

\paragraph{\ref{enu:selim}.}

The argument here is similar to that used to prove part~\ref{enu:serate}.
Analogously to (\ref{eq:WtT}), define
\begin{equation}
\U_{t,T}(\lcof)\defeq\sum_{s=1}^{t}\{\smlabs{y_{s}-[x_{s-1}+\lcof^{\trans}z_{s-1,T}]_{+}}-\expect_{s-1}\smlabs{y_{s}-[x_{s-1}+\lcof^{\trans}z_{s-1,T}]_{+}}\}\label{eq:WrT}
\end{equation}
so that $\U_{T}(\lcof)=\U_{T,T}(\lcof)$. Define the stopped process
$\N_{t,T}(\lcof)\defeq\U_{t\pmin\varsigma_{T},T}(\lcof)$, so that
$\{\N_{t,T}(\lcof),\filtg_{t}\}_{t=1}^{T}$ is a martingale array,
for each $\lcof\in\reals^{k+1}$. Define
\[
\L_{t,T}(\lcof)\defeq\U_{t,T}(\lcof)-\N_{t,T}(\lcof)=\U_{t,T}(\lcof)-\U_{t\pmin\varsigma_{T},T}(\lcof)
\]
so that
\begin{multline}
\Prob\left\{ \sup_{\smlnorm{\lcof-\tilde{\lcof}}\leq\delta}\smlabs{\U_{T}(\lcof)-\U_{T}(\tilde{\lcof})}>2\epsilon\right\} \\
\leq\Prob\left\{ \sup_{\smlnorm{\lcof-\tilde{\lcof}}\leq\delta}\smlabs{\N_{T,T}(\lcof)-\N_{T,T}(\tilde{\lcof})}>\epsilon\right\} +\Prob\left\{ \sup_{\smlnorm{\lcof-\tilde{\lcof}}\leq\delta}\smlabs{\L_{T,T}(\lcof)-\L_{T,T}(\tilde{\lcof})}>\epsilon\right\} .\label{eq:Wse}
\end{multline}
Thus it suffices to show that, for a suitable choice of $\delta>0$,
both r.h.s.\ probabilities can eventually be bounded by $\epsilon$.
Regarding the second of these, the definition of $\L_{t,T}$ entails
that if $\varsigma_{T}\geq T$, then $\L_{T,T}(\lcof)=0$ for all
$\lcof\in\reals^{k+1}$. Hence, recalling the choice of $K_{\epsilon}$
made above,
\[
\Prob\left\{ \sup_{\smlnorm{\lcof-\tilde{\lcof}}\leq\delta}\smlabs{\L_{T,T}(\lcof)-\L_{T,T}(\tilde{\lcof})}>\epsilon\right\} \leq1-\Prob\{\varsigma_{T}\geq T\}=1-\Prob\{Z_{T-1,T}\leq K_{\epsilon}\}<\epsilon,
\]
where the final bound holds for all $T$ sufficiently large.

It remains to bound the first r.h.s.\ term in (\ref{eq:Wse}). Similarly
to (\ref{eq:Nincr}), we have that the increments of $\N_{t,T}(\lcof)$
are equal to
\begin{align*}
\Delta\N_{t,T}(\lcof) & =\U_{t\pmin\varsigma_{T},T}(\lcof)-\U_{(t-1)\pmin\varsigma_{T},T}(\lcof) =[g_{t}(\lcof^{\trans}z_{t-1,T})-\expect_{t-1}g_{t}(\lcof^{\trans}z_{t-1,T})]\indic\{\varsigma_{T}>t-1\}
\end{align*}
and so, similarly to (\ref{eq:Nssq}),
\[
\sum_{t=1}^{T}[\Delta\N_{t,T}(\lcof)-\Delta\N_{t,T}(\tilde{\lcof})]^{2}\leq4\smlnorm{\lcof-\tilde{\lcof}}^{2}\sum_{t=1}^{T}\smlnorm{z_{t-1,T}}^{2}\indic\{\varsigma_{T}\geq t\}\leq4K_{\epsilon}\smlnorm{\lcof-\tilde{\lcof}}^{2}.
\]
Hence we have again by Burkholder's inequality that
\[
\expect\smlabs{\N_{T,T}(\lcof)-\N_{T,T}(\tilde{\lcof})}^{2j}\leq4^{j}CK_{\epsilon}^{j}\smlnorm{\lcof-\tilde{\lcof}}^{2j}.
\]
Taking $j$ such that $2j>k+1$, it follows by Kolmogorov's continuity
criterion (Corollary~16.9 in \citealp{Kal01}) that there exists
a $\delta>0$ such that
\[
\limsup_{T\goesto\infty}\Prob\left\{ \sup_{\smlnorm{r-\tilde{r}}\leq\delta}\smlabs{\N_{T,T}(\lcof)-\N_{T,T}(\tilde{\lcof})}>\epsilon\right\} <\epsilon.
\]

\paragraph{\ref{enu:mglimit}.}

Fix an $\lcof\in\reals^{k+1}$. Recalling (\ref{eq:gt}) above, so
that
\begin{align*}
g_{t}(\lcof^{\trans}z_{t-1,T}) & =\smlabs{y_{t}-[x_{t-1}+\lcof^{\trans}z_{t-1,T}]_{+}}-\smlabs{y_{t}-[x_{t-1}]_{+}}\\
 & =\smlabs{(y_{t}-[x_{t-1}]_{+})-([x_{t-1}+\lcof^{\trans}z_{t-1,T}]_{+}-[x_{t-1}]_{+})}-\smlabs{y_{t}-[x_{t-1}]_{+}}
\end{align*}
and we may write
\begin{equation}
\U_{T}(\lcof)-\U_{T}(0)=\sum_{t=1}^{T}[g_{t}(\lcof^{\trans}z_{t-1,T})-\expect_{t-1}g_{t}(\lcof^{\trans}z_{t-1,T})].\label{eq:WTrep}
\end{equation}
Applying the identity (\ref{eq:herce}) to $g_{t}(\lcof^{\trans}z_{t-1,T})$
with
\begin{align*}
A_{t} & \defeq y_{t}-[x_{t-1}]_{+}=[x_{t-1}+u_{t}]_{+}-[x_{t-1}]_{+}\\
B_{t-1}(\lcof) & \defeq[x_{t-1}+\lcof^{\trans}z_{t-1,T}]_{+}-[x_{t-1}]_{+},
\end{align*}
yields
\begin{align}
g_{t}(\lcof^{\trans}z_{t-1,T}) & =-\sgn(A_{t})B_{t-1}(\lcof)+\indic\{A_{t}=0\}\smlabs{B_{t-1}(\lcof)}\nonumber \\
 & \qquad+2[B_{t-1}(\lcof)-A_{t}][\indic\{B_{t-1}(\lcof)>A_{t}>0\}-\indic\{B_{t-1}(\lcof)<A_{t}<0\}]\nonumber \\
 & \eqdef\upsilon_{1,t}+\upsilon_{2,t}+2\upsilon_{3,t}.\label{eq:gdecomp}
\end{align}

We will show that only $\upsilon_{1t}$ contributes non-negligibly
to (\ref{eq:WTrep}), as $T\goesto\infty$. Noting that
\begin{equation}
\smlabs{B_{t-1}(\lcof)}\leq\smlabs{\lcof^{\trans}z_{t-1,T}}\label{eq:Btbound}
\end{equation}
and that $A_{t}=0$ with nonzero probability only if $x_{t-1}\leq0$,
we obtain
\begin{align*}
\expect_{t-1}\upsilon_{2,t}^{2} & =\smlabs{B_{t-1}(\lcof)}^{2}\expect_{t-1}\indic\{A_{t}=0\}\leq\smlabs{\lcof^{\trans}z_{t-1,T}}^{2}\indic\{x_{t-1}\leq0\}.
\end{align*}
Further, since $B_{t-1}(\lcof)>A_{t}>0$ and $B_{t-1}(\lcof)<A_{t}<0$
both imply that
\[
\smlabs{B_{t-1}(\lcof)-A_{t}}\leq\smlabs{B_{t-1}(\lcof)}\text{ and }\smlabs{u_{t}}\leq\smlabs{\lcof^{\trans}z_{t-1,T}},
\]
we have
\begin{align*}
\expect_{t-1}\upsilon_{3,t}^{2} & \leq\smlabs{B_{t-1}(\lcof)}^{2}\indic\{\smlabs{u_{t}}\leq\smlabs{\lcof^{\trans}z_{t-1,T}}\}\\
 & =\smlabs{B_{t-1}(\lcof)}^{2}[F_{u}(\smlabs{\lcof^{\trans}z_{t-1,T}})-F_{u}(-\smlabs{\lcof^{\trans}z_{t-1,T}})]
  \leq C\smlabs{\lcof^{\trans}z_{t-1,T}}^{3}
\end{align*}
for some $C<\infty$ by the mean value theorem, since $f_{u}$ is
bounded. Deduce that
\begin{align*}
 & \sum_{t=1}^{T}\expect_{t-1}[v_{2,t}+2v_{3,t}-\expect_{t-1}(v_{2,t}+2v_{3,t})]^{2}
 \leq\sum_{t=1}^{T}\expect_{t-1}(v_{2,t}^{2}+4v_{3,t}^{2})\\
 & \leq\smlnorm{\lcof}^{2}\sum_{t=1}^{T}\smlnorm{z_{t-1,T}}^{2}\indic\{x_{t-1}\leq0\}+8C\smlnorm{\lcof}^{3}\sum_{t=1}^{T}\smlnorm{z_{t-1,T}}^{3}\\
 & \leq\smlnorm{\lcof}^{2}\left(\sum_{t=1}^{T}\smlnorm{z_{t-1,T}}^{4}\right)^{1/2}\left(\sum_{t=1}^{T}\indic\{x_{t-1}\leq0\}\right)^{1/2}
 +4C\smlnorm{\lcof}^{3}\max_{1\leq s\leq T}\smlnorm{z_{s-1,T}}\sum_{t=1}^{T}\smlnorm{z_{t-1,T}}^{2}\\
 &=o_{p}(1)
\end{align*}
by Lemmas~\ref{lem:integ}\ref{enu:indicbnd}, \ref{lem:moments}\ref{enu:ZZ},
\ref{lem:moments}\ref{enu:max}, and \ref{lem:moments}\ref{enu:fourthz}.
It follows by (\ref{eq:WTrep}), (\ref{eq:gdecomp}) and Corollary~3.1
of \citet{HH80} that
\begin{align*}
\U_{T}(\lcof)-\U_{T}(0) & =-\sum_{t=1}^{T}B_{t-1}(\lcof)[\sgn(A_{t})-\expect_{t-1}\sgn(A_{t})]+o_{p}(1)\\
 & =-\sum_{t=1}^{T}B_{t-1}(\lcof)e_{t}+o_{p}(1)
\end{align*}
for $e_{t}$ the (bounded) m.d.s.\ defined in (\ref{eq:emds}) above.

It remains to determine the weak limit of the r.h.s. To that end we
recall the argument given in the proof of Proposition~\ref{prop:centring}\ref{enu:centrelim},
that by defining
\begin{align*}
\indic_{t-1}(\lcof) & \defeq\indic\{x_{t-1}\geq0\sep x_{t-1}+\lcof^{\trans}z_{t-1,T}\geq0\} & \indic_{t-1}^{c}(\lcof) & \defeq1-\indic_{t-1}(\lcof)
\end{align*}
as in (\ref{eq:bothpos}) above, we obtain the bound
\begin{align}
\indic_{t-1}^{c}(\lcof) & \leq2\cdot\indic\{x_{t-1}\leq1\}+\indic\left\{ 1\leq\smlnorm{\lcof}\sup_{1\leq s\leq T}\smlnorm{z_{s-1,T}}\right\} \nonumber \\
 & \eqdef2\cdot\indic_{1,t-1}^{c}+\indic_{2}^{c}(\lcof),\label{eq:complagain}
\end{align}
per (\ref{eq:compldecomp}) above. Moreover,
\begin{align*}
B_{t-1}(\lcof)\indic_{t-1}(\lcof) & =\{[x_{t-1}+\lcof^{\trans}z_{t-1,T}]_{+}-[x_{t-1}]_{+}\}\indic_{t-1}(\lcof)=(\lcof^{\trans}z_{t-1,T})\indic_{t-1}(\lcof).
\end{align*}
Therefore
\begin{align*}
\U_{T}(\lcof)-\U_{T}(0) & =-\sum_{t=1}^{T}\indic_{t-1}(\lcof)B_{t-1}(\lcof)e_{t}-\sum_{t=1}^{T}\indic_{t-1}^{c}(\lcof)B_{t-1}(\lcof)e_{t}+o_{p}(1)\\
 & =-\sum_{t=1}^{T}(\lcof^{\trans}z_{t-1,T})e_{t}+\sum_{t=1}^{T}\indic_{t-1}^{c}(\lcof)[\lcof^{\trans}z_{t-1,T}-B_{t-1}(\lcof)]e_{t}+o_{p}(1).
\end{align*}
To see that the second term is asymptotically negligible, observe
that this is a martingale with conditional variance
\begin{align*}
\sum_{t=1}^{T}\indic_{t-1}^{c}(\lcof)[\lcof^{\trans}z_{t-1,T}-B_{t-1}(\lcof)]^{2}\expect_{t-1}e_{t}^{2} & \leq_{(1)}C\smlnorm{\lcof}^{2}\sum_{t=1}^{T}\indic_{t-1}^{c}(\lcof)\smlnorm{z_{t-1,T}}^{2}\\
 & \leq_{(2)}C\smlnorm{\lcof}^{2}\sum_{t=1}^{T}\smlnorm{z_{t-1,T}}^{2}\{2\cdot\indic_{1,t-1}^{c}+\indic_{2}^{c}(\lcof)\}\\
 & =_{(3)}o_{p}(1)
\end{align*}
where $\leq_{(1)}$ follows by (\ref{eq:Btbound}) and the boundedness
of $e_{t}$, $\leq_{(2)}$ by (\ref{eq:complagain}) and $=_{(3)}$
by the arguments that yielded (\ref{eq:ixnegl}) in the proof of Proposition~\ref{prop:centring}\ref{enu:centrelim}.
Hence by Corollary~3.1 of \citet{HH80},
\[
\U_{T}(\lcof)-\U_{T}(0)=-\lcof^{\trans}\sum_{t=1}^{T}z_{t-1,T}e_{t}+o_{p}(1).
\]

The martingale on the r.h.s.\ of the preceding is
\begin{equation}
\sum_{t=1}^{T}z_{t-1,T}e_{t}=T^{-1/2}\sum_{t=1}^{T}\begin{bmatrix}1\\
T^{-1/2}y_{t-1}\\
\Delta\vec y_{t-1}
\end{bmatrix}e_{t}.\label{eq:sumze}
\end{equation}
To determine its weak limit, consider the vector martingale process
\[
\mathbb{M}_{T}(\tau T)\defeq T^{-1/2}\sum_{t=1}^{\smlfloor{\tau T}}\begin{bmatrix}u_{t}\\
e_{t}\\
\Delta\vec y_{t-1}e_{t}
\end{bmatrix}
\]
for $\tau\in[0,1]$. Each of the elements of $\mathbb{M}_{T}$ satisfy a conditional
Lyapunov condition, for every $\tau\in[0,1]$: the first two elements
because $\expect_{t-1}\smlabs{u_{t}}^{2+\delta_{u}}$ and $e_{t}$
are respectively bounded, and the final $k-1$ elements because
\[
T^{-1-\delta_{u}/2}\sum_{t=1}^{T}\smlnorm{\Delta\vec y_{t-1}}^{2+\delta_{u}}\expect_{t-1}\smlabs{e_{t}}^{2+\delta_{u}}\leq CT^{-\delta_{u}/2}O_{p}(1)=o_{p}(1)
\]
by Lemma~\ref{lem:bd2022}\ref{enu:mombnd}. Therefore, to apply
a (functional) CLT to $\mathbb{M}_{T}$, it suffices to determine the probability
limit of the conditional variance process
\begin{align}
\smlcv{\mathbb{M}_{T}(\tau T)} & =\frac{1}{T}\sum_{t=1}^{\smlfloor{\tau T}}\expect_{t-1}\begin{bmatrix}u_{t}^{2} & e_{t}u_{t} & \Delta\vec y_{t-1}^{\trans}e_{t}u_{t}\\
e_{t}u_{t} & e_{t}^{2} & \Delta\vec y_{t-1}^{\trans}e_{t}^{2}\\
\Delta\vec y_{t-1}e_{t}u_{t} & \Delta\vec y_{t-1}e_{t}^{2} & \Delta\vec y_{t-1}\Delta\vec y_{t-1}^{\trans}e_{t}^{2}
\end{bmatrix},\label{eq:Mcv}
\end{align}
for each fixed $\tau\in[0,1]$.

To compute the limit of the r.h.s.\ of the preceding, note that
\[
\indic\{x_{t-1}>0\}\expect_{t-1}\sgn([x_{t-1}+u_{t}]_{+}-[x_{t-1}]_{+})=\indic\{x_{t-1}>0\}\expect_{t-1}\sgn(u_{t})=0
\]
and
\begin{align*}
\indic\{x_{t-1}\leq0\}\expect_{t-1}\sgn([x_{t-1}+u_{t}]_{+}-[x_{t-1}]_{+}) & =\indic\{x_{t-1}\leq0\}\Prob_{t-1}\{u_{t}>-x_{t-1}\}\\
 & =\indic\{x_{t-1}\leq0\}[1-F(-x_{t-1})],
\end{align*}
and hence
\begin{align}
\expect_{t-1}e_{t}^{2} & =\expect_{t-1}\sgn([x_{t-1}+u_{t}]_{+}-[x_{t-1}]_{+})^{2}-\{\expect_{t-1}\sgn([x_{t-1}+u_{t}]_{+}-[x_{t-1}]_{+})\}^{2}\nonumber \\
 & =1-\indic\{x_{t-1}\leq0\}[F(-x_{t-1})+(1-F(-x_{t-1}))^2]\nonumber \\
 & \eqdef1+H_{1}(x_{t-1})\label{eq:Ee2}
\end{align}
where $H_{1}(x)$ is bounded, and (trivially) $H_{1}(x)\goesto0$
as $x\goesto+\infty$. Similarly, since
\[
\expect_{t-1}e_{t}u_{t}=\expect_{t-1}\sgn([x_{t-1}+u_{t}]_{+}-[x_{t-1}]_{+})u_{t}
\]
we obtain
\begin{align*}
\indic\{x_{t-1}>0\}\expect_{t-1}e_{t}u_{t} & =\indic\{x_{t-1}>0\}\expect_{t-1}\sgn(u_{t})u_{t}=\indic\{x_{t-1}>0\}\expect\smlabs{u_{1}}
\end{align*}
and
\begin{align*}
\indic\{x_{t-1}\leq0\}\expect_{t-1}e_{t}u_{t} & =\indic\{x_{t-1}\leq0\}\expect_{t-1}u_{t}\indic\{u_{t}>-x_{t-1}\}\\
 & =\indic\{x_{t-1}\leq0\}\int_{-x_{t-1}}^{\infty}uf_{u}(u)\diff u
\end{align*}
whence
\begin{align}
\expect_{t-1}e_{t}u_{t} & =\indic\{x_{t-1}>0\}\expect\smlabs{u_{1}}+\indic\{x_{t-1}\leq0\}\int_{-x_{t-1}}^{\infty}uf_{u}(u)\diff u\nonumber \\
 & =\expect\smlabs{u_{1}}+\indic\{x_{t-1}\leq0\}\left[\int_{-x_{t-1}}^{\infty}uf_{u}(u)\diff u-\expect\smlabs{u_{1}}\right]\nonumber \\
 & \eqdef\expect\smlabs{u_{1}}+H_{2}(x_{t-1})\label{eq:Eeu}
\end{align}
where $H_{2}$ is bounded, and (trivially) $H_{2}(x)\goesto0$ as
$x\goesto+\infty$. Deduce that for each $\tau\in[0,1]$,
\begin{equation}
\frac{1}{T}\sum_{t=1}^{\smlfloor{\tau T}}\expect_{t-1}\begin{bmatrix}u_{t}^{2}\\
e_{t}u_{t}\\
e_{t}^{2}
\end{bmatrix}=\frac{1}{T}\sum_{t=1}^{\smlfloor{\tau T}}\begin{bmatrix}\sigma^{2}\\
\expect\smlabs{u_{1}}+H_{2}(x_{t-1})\\
1+H_{1}(x_{t-1})
\end{bmatrix}=\tau\begin{bmatrix}\sigma^{2}\\
\expect\smlabs{u_{1}}\\
1
\end{bmatrix}+o_{p}(1)\label{eq:cvp1}
\end{equation}
by Lemma~\ref{lem:integ}\ref{enu:integ}. This yields the limit
of the upper left $2\times2$ block of the r.h.s.\ of (\ref{eq:Mcv}).

Regarding the remaining elements of the r.h.s.\ of (\ref{eq:Mcv}),
we next observe that
\begin{align*}
\frac{1}{T}\sum_{t=1}^{\smlfloor{\tau T}}\expect_{t-1}\Delta\vec y_{t-1}\begin{bmatrix}e_{t}u_{t}\\
e_{t}^{2}
\end{bmatrix}^{\trans} & =\frac{1}{T}\sum_{t=1}^{\smlfloor{\tau T}}\Delta\vec y_{t-1}\expect_{t-1}\begin{bmatrix}e_{t}u_{t}\\
e_{t}^{2}
\end{bmatrix}^{\trans}\\
 & =\left(\frac{1}{T}\sum_{t=1}^{\smlfloor{\tau T}}\Delta\vec y_{t-1}\right)\begin{bmatrix}\expect\smlabs{u_{1}}\\
1
\end{bmatrix}^{\trans}+\frac{1}{T}\sum_{t=1}^{\smlfloor{\tau T}}\Delta\vec y_{t-1}\begin{bmatrix}H_{2}(x_{t-1})\\
H_{1}(x_{t-1})
\end{bmatrix}^{\trans},
\end{align*}
by (\ref{eq:Ee2}) and (\ref{eq:Eeu}). Now
\[
\frac{1}{T}\sum_{t=1}^{\smlfloor{\tau T}}\Delta\vec y_{t-1}=\frac{1}{T}(\vec y_{\smlfloor{\tau T}}-\vec y_{0})=O_{p}(T^{-1/2})
\]
by Theorem~3.2 in \citet{BD22}, while
\begin{multline*}
\abs{\frac{1}{T}\sum_{t=1}^{\smlfloor{\tau T}}\Delta\vec y_{t-1}\begin{bmatrix}H_{2}(x_{t-1})\\
H_{1}(x_{t-1})
\end{bmatrix}^{\trans}}\\
\leq\left(\frac{1}{T}\sum_{t=1}^{T}\smlnorm{\Delta\vec y_{t-1}}^{2}\right)^{1/2}\left(\frac{1}{T}\sum_{t=1}^{T}\norm{\begin{bmatrix}H_{2}(x_{t-1})\\
H_{1}(x_{t-1})
\end{bmatrix}}^{2}\right)^{1/2}=o_{p}(1)
\end{multline*}
by Lemmas~\ref{lem:bd2022}\ref{enu:deltayt} and \ref{lem:integ}\ref{enu:integ}.
Hence for each $\tau\in[0,1]$
\begin{equation}
\frac{1}{T}\sum_{t=1}^{\smlfloor{\tau T}}\expect_{t-1}\Delta\vec y_{t-1}\begin{bmatrix}e_{t}u_{t}\\
e_{t}^{2}
\end{bmatrix}^{\trans}=o_{p}(1).\label{eq:cvp2}
\end{equation}
Finally, since
\[
\frac{1}{T}\sum_{t=1}^{\smlfloor{\tau T}}\expect_{t-1}\Delta\vec y_{t-1}\Delta\vec y_{t-1}^{\trans}e_{t}^{2}=\frac{1}{T}\sum_{t=1}^{\smlfloor{\tau T}}\Delta\vec y_{t-1}\Delta\vec y_{t-1}^{\trans}+\frac{1}{T}\sum_{t=1}^{\smlfloor{\tau T}}\Delta\vec y_{t-1}\Delta\vec y_{t-1}^{\trans}H_{1}(x_{t-1})
\]
where by H\"{o}lder's inequality, the norm of the final term is bounded
by
\[
\left(\frac{1}{T}\sum_{t=1}^{T}\smlnorm{\Delta\vec y_{t-1}}^{2+\delta_{u}}\right)^{2/(2+\delta_{u})}\left(\frac{1}{T}\sum_{t=1}^{T}\smlabs{H_{1}(x_{t-1})}^{(2+\delta_{u})/\delta_{u}}\right)^{\delta_{u}/(2+\delta_{u})}=o_{p}(1)
\]
by Lemmas~\ref{lem:bd2022}\ref{enu:mombnd} and \ref{lem:integ}\ref{enu:integ}.
Hence, for each $\tau\in[0,1]$,
\begin{equation}
\frac{1}{T}\sum_{t=1}^{\smlfloor{\tau T}}\expect_{t-1}\Delta\vec y_{t-1}\Delta\vec y_{t-1}^{\trans}e_{t}^{2}\inprob\tau\Omega\label{eq:cvp3}
\end{equation}
by Lemma~\ref{lem:bd2022}\ref{enu:deltayt}.

It follows from (\ref{eq:cvp1})--(\ref{eq:cvp3}) that for each $\tau\in[0,1]$,
\[
\smlcv{\mathbb{M}_{T}(\tau T)}\inprob\tau\begin{bmatrix}\sigma^{2} & \expect\smlabs{u_{1}} & 0\\
\expect\smlabs{u_{1}} & 1 & 0\\
0 & 0 & \Omega
\end{bmatrix}\eqdef\tau\Omega_{M}
\]
Therefore Theorem~2.3 in \citet{DR78AP}, and the Cram\'{e}r--Wold device,
imply that
\begin{equation}
\mathbb{M}_{T}(\tau T)=T^{-1/2}\sum_{t=1}^{\smlfloor{\tau T}}\begin{bmatrix}u_{t}\\
e_{t}\\
\Delta\vec y_{t-1}e_{t}
\end{bmatrix}\indist\begin{bmatrix}\sigma_0 W(\tau)\\
\widetilde{W}(\tau)\\
\Xi(\tau)
\end{bmatrix}\eqdef \mathbb{M}(\tau)\label{eq:Mcvg}
\end{equation}
on $D[0,1]$, where $\mathbb{M}$ is a $(k+1)$-dimensional Brownian motion
with variance matrix $\Omega_{M}$: and so, in particular, $\Xi$
is independent of $(W,\widetilde{W})$.

Finally, returning to (\ref{eq:sumze}), i.e.\ to
\[
\sum_{t=1}^{T}z_{t-1,T}e_{t}=\sum_{t=1}^{T}\begin{bmatrix}T^{-1/2}e_{t}\\
T^{-1}y_{t-1}e_{t}\\
T^{-1/2}\Delta\vec y_{t-1}e_{t}
\end{bmatrix}=\begin{bmatrix}T^{-1/2}\sum_{t=1}^{T}e_{t}\\
T^{-1}\sum_{t=1}^{T}y_{t-1}e_{t}\\
\Xi_{T}(1)
\end{bmatrix},
\]
we note that by Theorem~3.2 of \citet{BD22} and Theorem~2.1 of
\citet{ito_convergence},
\[
T^{-1}\sum_{t=1}^{T}y_{t-1}e_{t}\indist\int_{0}^{1}Y(\tau)\diff \widetilde{W}(\tau)
\]
holds jointly with (\ref{eq:Mcvg}), where $Y$ is a function only
of $W$. Deduce
\[
\sum_{t=1}^{T}z_{t-1,T}e_{t}\indist\begin{bmatrix}\widetilde{W}(1)\\
\int_{0}^{1}Y(\tau)\diff \widetilde{W}(\tau)\\
\Xi(1)
\end{bmatrix},
\]
where $\Xi(1)\distrib\mathcal{N}[0,\Omega]$ independently of $(\widetilde{W},Y)$.\hfill\qedsymbol

\section{Model selection}
\begin{proof}[Proof of Theorem~\ref{thm:ttest_MLE}]
Let $k^*$ be the order selected by the iterative procedure. The $t$-statistic for $\phi_{k-1}$ in the Tobit of order $k$ is
\begin{equation}\label{eq_tstatphik}
  t_k(T):=\frac{\sqrt{T}\hat{\phi}^{\mle}_{k-1,T}}{\hat{\sigma}_T^{\mle}\sqrt{\left[\hat{\Omega}_T^{-1}\right]_{k-1,k-1}}},
\end{equation}
where $\left[\hat{\Omega}_T^{-1}\right]_{k-1,k-1}$ is the bottom-right element of $k-1\times k-1$ matrix $\hat{\Omega}_T^{-1}$. Then
\begin{equation}\label{eq_prob_overselection}
\Prob(k^*>k_0)=\Prob\left(\max_{k>k_0} |t_k(T)|>cv(T)\right)\leq\sum\limits_{k=k_0+1}^{\bar{k}}\Prob\left(|t_k(T)|>cv(T)\right)\xrightarrow[T\to\infty]{}0,
\end{equation}
since $cv(T)\to\infty$, but, given $\phi_{k-1}=0$ for any fixed $k\geq k_0$, Theorem \ref{thm:tobitmle} implies ${t_k(T)\xrightarrow[T\to\infty]{}\mathcal{N}(0,1)}$, i.e., $t_k(T)=O_p(1)$.

If $k_0=1$, then \eqref{eq_prob_overselection} implies $\Prob(k^*>1)\to0$, i.e., $\Prob(k^*=1)\to1$. Thus, the procedure is consistent. Suppose that $k_0>1$ so that $\phi_{k_0-1}\neq0$. Then
\begin{equation*}\begin{split}
\Prob\left(|t_{k_0}(T)|<cv(T)\right)&=\Prob\left(\frac{|t_{k_0}(T)|}{\sqrt{T}}<\frac{cv(T)}{\sqrt{T}}\right)
=\Prob\left(\frac{|\hat{\phi}^{\mle}_{k_0-1,T}-\phi_{k_0}+\phi_{k_0}|}{\hat{\sigma}_T^{\mle}\sqrt{\left[\hat{\Omega}_T^{-1}\right]_{k-1,k-1}}}<\frac{cv(T)}{\sqrt{T}}\right)\\
&\leq\Prob\left(\frac{|\phi_{k_0-1}|}{\hat{\sigma}_T^{\mle}\sqrt{\left[\hat{\Omega}_T^{-1}\right]_{k-1,k-1}}}
-\frac{|\hat{\phi}^{\mle}_{k_0-1,T}-\phi_{k_0}|}{\hat{\sigma}_T^{\mle}\sqrt{\left[\hat{\Omega}_T^{-1}\right]_{k-1,k-1}}}
<\frac{cv(T)}{\sqrt{T}}\right)
\xrightarrow[T\to\infty]{}0,
\end{split}\end{equation*}
where the convergence to $0$ follows from $cv(T)/\sqrt{T}\to0$, while $\hat{\phi}^{\mle}_{k_0-1,T}-\phi_{k_0}=O_p(T^{-1/2})$, $\hat{\sigma}_T^{\mle}=\xrightarrow[T\to\infty]{p}\sigma_0>0$, $\hat{\Omega}_T\xrightarrow[T\to\infty]{p}\Omega>0$ and by Cramer's rule
$$\left|\left[\hat{\Omega}_T^{-1}\right]_{k-1,k-1}\right|=\left|det\left(\left[\hat{\Omega}_T\right]_{1:k-1,1:k-1}\right)/\det\left(\hat{\Omega}_T\right)\right|>0.$$
Thus,
\begin{equation}
\Prob(k^*<k_0)\leq\Prob\left(|t_{k_0}(T)|<cv(T)\right)\xrightarrow[T\to\infty]{}0,
\end{equation}
and the procedure is consistent, $\Prob(k^*=k_0)\xrightarrow[T\to\infty]{}1$.
\end{proof}

\begin{proof}[Proof of Theorem~\ref{thm:ttest_LAD}]
The proof is identical to that for the MLE, except for the use of a proper consistent estimator of the standard deviation $\frac1{2\hat{f}_{u,k,T}}\sqrt{\left[\hat{\Omega}_T^{-1}\right]_{k-1,k-1}}$.
\end{proof}

\bibliographystyle{ecta}
\addcontentsline{toc}{section}{References}
\bibliography{tobit_lad_mle}

\end{document}

%% file: macros-v4.tex


\global\long\def\uwrite#1#2{\underset{#2}{\underbrace{#1}} }%

\global\long\def\blw#1{\ensuremath{\underline{#1}}}%

\global\long\def\abv#1{\ensuremath{\overline{#1}}}%

\global\long\def\vect#1{\mathbf{#1}}%


\global\long\def\smlseq#1{\{#1\} }%

\global\long\def\seq#1{\left\{  #1\right\}  }%

\global\long\def\smlsetof#1#2{\{#1\mid#2\} }%

\global\long\def\setof#1#2{\left\{  #1\mid#2\right\}  }%


\global\long\def\goesto{\ensuremath{\rightarrow}}%

\global\long\def\ngoesto{\ensuremath{\nrightarrow}}%

\global\long\def\uto{\ensuremath{\uparrow}}%

\global\long\def\dto{\ensuremath{\downarrow}}%

\global\long\def\uuto{\ensuremath{\upuparrows}}%

\global\long\def\ddto{\ensuremath{\downdownarrows}}%

\global\long\def\ulrto{\ensuremath{\nearrow}}%

\global\long\def\dlrto{\ensuremath{\searrow}}%


\global\long\def\setmap{\ensuremath{\rightarrow}}%

\global\long\def\elmap{\ensuremath{\mapsto}}%

\global\long\def\compose{\ensuremath{\circ}}%

\global\long\def\cont{C}%

\global\long\def\cadlag{D}%

\global\long\def\Ellp#1{\ensuremath{\mathcal{L}^{#1}}}%


\global\long\def\naturals{\ensuremath{\mathbb{N}}}%

\global\long\def\reals{\mathbb{R}}%

\global\long\def\complex{\mathbb{C}}%

\global\long\def\rationals{\mathbb{Q}}%

\global\long\def\integers{\mathbb{Z}}%


\global\long\def\abs#1{\ensuremath{\left|#1\right|}}%

\global\long\def\smlabs#1{\ensuremath{\lvert#1\rvert}}%
 
\global\long\def\bigabs#1{\ensuremath{\bigl|#1\bigr|}}%
 
\global\long\def\Bigabs#1{\ensuremath{\Bigl|#1\Bigr|}}%
 
\global\long\def\biggabs#1{\ensuremath{\biggl|#1\biggr|}}%

\global\long\def\norm#1{\ensuremath{\left\Vert #1\right\Vert }}%

\global\long\def\smlnorm#1{\ensuremath{\lVert#1\rVert}}%
 
\global\long\def\bignorm#1{\ensuremath{\bigl\|#1\bigr\|}}%
 
\global\long\def\Bignorm#1{\ensuremath{\Bigl\|#1\Bigr\|}}%
 
\global\long\def\biggnorm#1{\ensuremath{\biggl\|#1\biggr\|}}%

\global\long\def\floor#1{\left\lfloor #1\right\rfloor }%
\global\long\def\smlfloor#1{\lfloor#1\rfloor}%

\global\long\def\ceil#1{\left\lceil #1\right\rceil }%
\global\long\def\smlceil#1{\lceil#1\rceil}%


\global\long\def\Union{\ensuremath{\bigcup}}%

\global\long\def\Intsect{\ensuremath{\bigcap}}%

\global\long\def\union{\ensuremath{\cup}}%

\global\long\def\intsect{\ensuremath{\cap}}%

\global\long\def\pset{\ensuremath{\mathcal{P}}}%

\global\long\def\clsr#1{\ensuremath{\overline{#1}}}%

\global\long\def\symd{\ensuremath{\Delta}}%

\global\long\def\intr{\operatorname{int}}%

\global\long\def\cprod{\otimes}%

\global\long\def\Cprod{\bigotimes}%


\global\long\def\smlinprd#1#2{\ensuremath{\langle#1,#2\rangle}}%

\global\long\def\inprd#1#2{\ensuremath{\left\langle #1,#2\right\rangle }}%

\global\long\def\orthog{\ensuremath{\perp}}%

\global\long\def\dirsum{\ensuremath{\oplus}}%


\global\long\def\spn{\operatorname{sp}}%

\global\long\def\rank{\operatorname{rk}}%

\global\long\def\proj{\operatorname{proj}}%

\global\long\def\tr{\operatorname{tr}}%

\global\long\def\vek{\operatorname{vec}}%

\global\long\def\diag{\operatorname{diag}}%

\global\long\def\col{\operatorname{col}}%


\global\long\def\smpl{\ensuremath{\Omega}}%

\global\long\def\elsmp{\ensuremath{\omega}}%

\global\long\def\sigf#1{\mathcal{#1}}%

\global\long\def\sigfield{\ensuremath{\mathcal{F}}}%
\global\long\def\sigfieldg{\ensuremath{\mathcal{G}}}%

\global\long\def\flt#1{\mathcal{#1}}%

\global\long\def\filt{\mathcal{F}}%
\global\long\def\filtg{\mathcal{G}}%

\global\long\def\Borel{\ensuremath{\mathcal{B}}}%

\global\long\def\cyl{\ensuremath{\mathcal{C}}}%

\global\long\def\nulls{\ensuremath{\mathcal{N}}}%

\global\long\def\gauss{\mathfrak{g}}%

\global\long\def\leb{\mathfrak{m}}%


\global\long\def\prob{\ensuremath{\mathbb{P}}}%

\global\long\def\Prob{\ensuremath{\mathbb{P}}}%

\global\long\def\Probs{\mathcal{P}}%

\global\long\def\PROBS{\mathcal{M}}%

\global\long\def\expect{\ensuremath{\mathbb{E}}}%

\global\long\def\Expect{\ensuremath{\mathbb{E}}}%

\global\long\def\probspc{\ensuremath{(\smpl,\filt,\Prob)}}%


\global\long\def\iid{\ensuremath{\textnormal{i.i.d.}}}%

\global\long\def\as{\ensuremath{\textnormal{a.s.}}}%

\global\long\def\asp{\ensuremath{\textnormal{a.s.p.}}}%

\global\long\def\io{\ensuremath{\ensuremath{\textnormal{i.o.}}}}%

\newcommand\independent{\protect\mathpalette{\protect\independenT}{\perp}}
\def\independenT#1#2{\mathrel{\rlap{$#1#2$}\mkern2mu{#1#2}}}

\global\long\def\indep{\independent}%

\global\long\def\distrib{\ensuremath{\sim}}%

\global\long\def\distiid{\ensuremath{\sim_{\iid}}}%

\global\long\def\asydist{\ensuremath{\overset{a}{\distrib}}}%

\global\long\def\inprob{\ensuremath{\overset{p}{\goesto}}}%

\global\long\def\inprobu#1{\ensuremath{\overset{#1}{\goesto}}}%

\global\long\def\inas{\ensuremath{\overset{\as}{\goesto}}}%

\global\long\def\eqas{=_{\as}}%

\global\long\def\inLp#1{\ensuremath{\overset{\Ellp{#1}}{\goesto}}}%

\global\long\def\indist{\ensuremath{\overset{d}{\goesto}}}%

\global\long\def\eqdist{=_{d}}%

\global\long\def\wkc{\ensuremath{\rightsquigarrow}}%

\global\long\def\wkcu#1{\overset{#1}{\ensuremath{\rightsquigarrow}}}%

\global\long\def\plim{\operatorname*{plim}}%


\global\long\def\var{\operatorname{var}}%

\global\long\def\lrvar{\operatorname{lrvar}}%

\global\long\def\cov{\operatorname{cov}}%

\global\long\def\corr{\operatorname{corr}}%

\global\long\def\bias{\operatorname{bias}}%

\global\long\def\MSE{\operatorname{MSE}}%

\global\long\def\med{\operatorname{med}}%

\global\long\def\avar{\operatorname{avar}}%

\global\long\def\se{\operatorname{se}}%

\global\long\def\sd{\operatorname{sd}}%


\global\long\def\nullhyp{H_{0}}%

\global\long\def\althyp{H_{1}}%

\global\long\def\ci{\mathcal{C}}%


\global\long\def\simple{\mathcal{R}}%

\global\long\def\sring{\mathcal{A}}%

\global\long\def\sproc{\mathcal{H}}%

\global\long\def\Wiener{\ensuremath{\mathbb{W}}}%

\global\long\def\sint{\bullet}%

\global\long\def\cv#1{\left\langle #1\right\rangle }%

\global\long\def\smlcv#1{\langle#1\rangle}%

\global\long\def\qv#1{\left[#1\right]}%

\global\long\def\smlqv#1{[#1]}%


\global\long\def\trans{\mathsf{T}}%

\global\long\def\indic{\ensuremath{\mathbf{1}}}%

\global\long\def\Lagr{\mathcal{L}}%

\global\long\def\grad{\nabla}%

\global\long\def\pmin{\ensuremath{\wedge}}%
\global\long\def\Pmin{\ensuremath{\bigwedge}}%

\global\long\def\pmax{\ensuremath{\vee}}%
\global\long\def\Pmax{\ensuremath{\bigvee}}%

\global\long\def\sgn{\operatorname{sgn}}%

\global\long\def\argmin{\operatorname*{argmin}}%

\global\long\def\argmax{\operatorname*{argmax}}%

\global\long\def\Rp{\operatorname{Re}}%

\global\long\def\Ip{\operatorname{Im}}%

\global\long\def\deriv{\ensuremath{\mathrm{d}}}%

\global\long\def\diffnspc{\ensuremath{\deriv}}%

\global\long\def\diff{\ensuremath{\,\deriv}}%

\global\long\def\i{\ensuremath{\mathrm{i}}}%

\global\long\def\e{\mathrm{e}}%

\global\long\def\sep{,\ }%

\global\long\def\defeq{\coloneqq}%

\global\long\def\eqdef{\eqqcolon}%

